\documentclass[10pt]{amsart}
\usepackage{amsthm,amsmath,amssymb}
\usepackage[margin=1.3in]{geometry}
\usepackage{enumerate}[\indent 1.]
\usepackage{mathtools}
\usepackage{amsfonts}
\usepackage[comma, sort&compress]{natbib}
\usepackage{soul}
\usepackage{multirow}
\usepackage[utf8]{inputenc}
\usepackage{setspace}
\usepackage{subfig}
\usepackage{leftidx}
\usepackage{mathabx}
\usepackage{varwidth}
\usepackage{xcolor}
\usepackage{mathrsfs}
\usepackage{booktabs}
\usepackage{breqn}
\usepackage[linesnumbered, ruled, noline]{algorithm2e}
\usepackage[hidelinks]{hyperref}
\usepackage[htt]{hyphenat}
\usepackage{float} 
\usepackage{booktabs} 
\usepackage{tabularx} 
\usepackage{ulem} 

\usepackage[T1]{fontenc}
\usepackage[utf8]{inputenc}
\usepackage{babel}
\usepackage[font=small,labelfont=bf]{caption}

\usepackage{paralist}

\usepackage{todonotes}

\makeatletter
\newcommand\dotover{\leavevmode\cleaders\hb@xt@ .22em{\hss $\cdot$\hss}\hfill\kern\z@}

\makeatother

\newtheorem{theorem}{Theorem}
\newtheorem{assumption}{Assumption}
\newtheorem{conditionalt}{Condition}

\newtheorem{lemma}[theorem]{Lemma}
\newtheorem{proposition}[theorem]{Proposition}

\newtheorem{definition}{Definition}
\theoremstyle{remark}
\newtheorem{example}{Example}

\newtheorem{claim*}{Claim}
\newtheorem{claim}{Claim}
\DeclareMathOperator*{\argmax}{arg\,max}

\numberwithin{equation}{section}
\numberwithin{theorem}{section}
\numberwithin{example}{section}
\numberwithin{table}{section}
\numberwithin{definition}{section}
\numberwithin{figure}{section}

\newcommand{\cR}{\mathcal{R}}
\newcommand{\cHun}{\mathcal{H}^{u/[n]}}
\newcommand{\cH}{\mathcal{H}}
\newcommand{\cG}{\mathcal{G}}
\newcommand{\cJ}{\mathcal{J}}
\newcommand{\cA}{\mathcal{A}}
\newcommand{\cS}{\mathcal{S}}
\newcommand{\cL}{\mathcal{L}}
\newcommand{\cT}{\mathcal{T}}

\newcommand{\bR}{\mathbb{R}}
\newcommand{\bN}{\mathbb{N}}

\newcommand{\zetabf}{\boldsymbol{\zeta}}
\newcommand{\betabf}{\boldsymbol{\beta}}
\newcommand{\tvec}{\mathbf{t}}
\newcommand{\hattvec}{\hat{\mathbf{t}}}
\newcommand{\hatt}{\hat{t}}

\newcommand{\methodname}{ParFilter }
\newcommand{\methodnamens}{ParFilter}
\newcommand{\cI}{\mathcal{I}}
\newcommand{\cC}{\mathcal{C}}

\newcommand{\uiG}{u_{ik} / \cG_k}
\newcommand{\uiGc}{r_{ik} / \cC_k}

\newcommand{\FDR}{\text{FDR}}
\newcommand{\FDP}{\text{FDP}}

\newcommand{\rep}{\text{rep}}

\newcommand{\FDPhat}{\widehat{\FDP}}

\newcommand{\Simes}{\text{Simes}}

\newcommand{\1}{I}
\newcommand{\nubf}{\boldsymbol{\nu}}
\newcommand{\matleq}{\preccurlyeq}
\newcommand{\matgeq}{\succcurlyeq}
\newcommand{\angleb}[1]{\langle {#1} \rangle}

\newcommand{\Clfdrc}{\omega}

\newcommand{\Clfdr}{\psi}

\newcommand{\appropto}{\mathrel{\vcenter{
  \offinterlineskip\halign{\hfil$##$\cr
    \propto\cr\noalign{\kern2pt}\sim\cr\noalign{\kern-2pt}}}}}

\newcommand{\assumpref}[1]{Assumption~\ref{assump:#1}}
\newcommand{\assumpsref}[1]{Assumptions~\ref{assump:#1}}
\newcommand{\assumpssref}[1]{\ref{assump:#1}}
\newcommand{\figref}[1]{Figure~\ref{fig:#1}}

\newcommand{\condref}[1]{Condition~\ref{cond:#1}}
\newcommand{\condsref}[1]{Conditions~\ref{cond:#1}}

\newcommand{\secref}[1]{Section~\ref{sec:#1}}
\newcommand{\secsref}[1]{Sections~\ref{sec:#1}}
\newcommand{\secssref}[1]{\ref{sec:#1}}

\newcommand{\algref}[1]{Algorithm~\ref{alg:#1}}
\newcommand{\appref}[1]{Appendix~\ref{app:#1}}
\newcommand{\appsref}[1]{Appendices~\ref{app:#1}}
\newcommand{\appssref}[1]{\ref{app:#1}}

\newcommand{\defref}[1]{Definition~\ref{def:#1}}

\newcommand{\exref}[1]{Example~\ref{ex:#1}}

\newcommand{\lemref}[1]{Lemma~\ref{lem:#1}}

\newcommand{\propref}[1]{Proposition~\ref{prop:#1}}
\newcommand{\propsref}[1]{Propositions~\ref{prop:#1}}
\newcommand{\propssref}[1]{\ref{prop:#1}}
\newcommand{\thmref}[1]{Theorem~\ref{thm:#1}}

\newcommand{\thmsref}[1]{Theorems~\ref{thm:#1}}
\newcommand{\thmssref}[1]{\ref{thm:#1}}

\newcommand{\tabref}[1]{Table~\ref{tab:#1}}

\newcommand{\claimref}[1]{Claim~\ref{claim:#1}}

\newcommand{\AdaFilterBHCol}{dark gray }
\newcommand{\AdaPTCol}{green }
\newcommand{\AdaptiveBHCol}{deep orange }
\newcommand{\AdaptiveCoFilterBHCol}{lavender }

\newcommand{\BYCol}{brown }
\newcommand{\CAMTCol}{pink }
\newcommand{\CoFilterBHCol}{peach }
\newcommand{\IHWCol}{turquoise }
\newcommand{\InflatedAdaFilterBHCol}{light blue }
\newcommand{\InflatedParFilterCol}{purple }
\newcommand{\NoCovarParFilterCol}{maroon }
\newcommand{\NonAdaptiveParFilterCol}{olive }
\newcommand{\OracleCol}{light grey }
\newcommand{\ParFilterCol}{navy blue }

\allowdisplaybreaks

\title{
A covariate-adaptive test for replicability across multiple studies with false discovery rate control
}

\author[N.~Tran]{Ninh Tran}
\address[Ninh Tran]{School of Mathematics and Statistics, University of Melbourne, Australia}
\email{ninht@student.unimelb.edu.au}

\author[D.~Leung]{Dennis Leung}
\address[Dennis Leung]{School of Mathematics and Statistics, University of Melbourne, Australia}
\email{dennis.leung@unimelb.edu.au}


\thanks{This research was supported by The University of Melbourne's Research Computing Services and the Petascale Campus Initiative.}


\begin{document}

\maketitle

\begin{abstract}

Replicability  is a lynchpin for  credible  discoveries.  The partial conjunction (PC) p-value, which  combines individual base p-values  from multiple similar studies,  can  gauge  whether a feature of interest exhibits replicated  signals  across studies. However, when a large set of   features are examined as in high-throughput experiments,  testing for their replicated signals simultaneously can pose a very underpowered problem, due to both the  multiplicity burden  and   inherent limitations of PC $p$-values. This power deficiency is markedly severe when replication is demanded for all studies under consideration,  which is nonetheless the most natural and appealing benchmark for scientific generalizability a practitioner may  request. 
 
We propose ParFilter, a general framework that marries the ideas of filtering and covariate-adaptiveness to power up  large-scale testing for replicated signals as described above. It reduces the multiplicity burden by partitioning studies into smaller groups and borrowing the cross-group information  to filter out unpromising features. Moreover, harnessing side information offered by auxiliary covariates whenever they are available, it can train informative hypothesis weights to encourage rejections of features more likely to exhibit replicated signals. We prove its finite-sample control on the false discovery rate, under both independence and arbitrary dependence among the  base $p$-values across features.   In simulations as well as  a real case study on autoimmunity based on RNA-Seq data obtained from thymic cells, the ParFilter has demonstrated competitive performance against other existing methods for such replicability analyses.

\end{abstract}

\section{Introduction}\label{sec:intro}
The  replicability crisis \citep{Ioannidis2005}  has raised widespread concern about the credibility of  contemporary scientific findings. A recent article of \citet{nosek2022replicability} has recounted how psychologists have been actively confronting this  issue within their field over the past decade, as a case in point. To enhance the reliability of and assurance in scientific claims, the statistical community  have responded by  developing    rigorous methods for \emph{replicability analysis}, which aim to identify, in a principled manner,  replicable discoveries  that are supported by   data  from \emph{multiple studies} of a similar nature.  In an era of open science where  experimental data from different research groups become more  accessible,  these statistical tools have been increasingly employed in different domains, such as  immunology \citep{jaljuli2025revisiting}, natural language processing \citep{dror2017replicability} and developmental neuroscience \citep{meltzoff2018infant}. In this paper, we  will further contribute a timely method for such analyses of  modern  high-throughput experimental data.

Consider a typical meta-analytical  framework that consists of $n$ comparable studies, each examining the same  set of $m$ features. Adopting the shorthand $[\ell] \equiv \{ 1, \dots, \ell \}$ for any natural  number $\ell \in \bN \equiv \{1,2, \dots\}$,  suppose $\mu_{ij} \in \mathbb{R}$ is an unknown effect parameter corresponding to feature $i \in [m]$ in study $j \in [n]$, and $\mathcal{A}_i \subseteq \bR$ is a common \textit{null} parameter region across the $n$ studies for feature $i$. The effect specific to study $j$ is considered a scientifically interesting finding (i.e., a \textit{signal}) if $\mu_{ij}$ does not belong to $\mathcal{A}_i$. Consequently, the query of whether $\mu_{ij}$ is a signal can be  formulated as  testing the null hypothesis
\begin{equation} \label{individual_hypothesis}
H_{ij}: \mu_{ij} \in \cA_i
\end{equation}
against the alternative  $\mu_{ij} \not \in \cA_i$. In this paper, we  exemplify the problem with a meta-analysis on thymic gene differential expression in relation to autoimmunity, the details of which will be given  in \secref{realdata}. 

\begin{example}[Meta-analysis on differential gene expression in thymic medulla for autoimmune disorder]\label{ex:DGE}
 We consider a meta-analysis on RNA-Seq data obtained from medullary thymic epithelial cells (mTECs) for $m = 6,587$ genes compiled from $n = 3$ independent mouse studies. For each $(i,j) \in [6587] \times [3]$, $\mu_{ij} \in \mathbb{R}$ is  the mean difference in expression level of gene $i$ between healthy and autoimmune mice in study $j$. If $\mu_{ij} \neq 0$ (corresponding to $\mathcal{A}_i = \{0\}$), then gene $i$ in study $j$ is implicated as a potential marker for autoimmunity, as its expression differs between healthy and autoimmune mice, on average.
 \end{example}

Testing the  individual hypotheses $H_{ij}$ is certainly central to each study. However, to reinforce the generalizability of scientific findings, one may  ask whether a given feature $i$ has \textit{replicated} signals across  studies.  First advocated by \citet{FRISTON2005661}, \citet{Benjamini2008} and \citet{Benjamini2009}, the following paradigm  addresses the latter problem: 
For a given number $u \in [n]$ that represents the required \textit{level of replicability}, declare the effect of feature $i$ as being $u/[n]$ \textit{replicated} if at least $u$ of  $\mu_{i1},\dots,\mu_{in}$ are signals. This can be formulated as testing the \textit{partial conjunction} (PC) null hypothesis 
\begin{equation} \label{PC_null}
    H^{u/[n]}_{i} :  |\{ j \in [n]: \mu_{ij} \in \mathcal{A}_i \}| \geq n - u +1
\end{equation}
which states that at least $n-u+1$ of the effects $\mu_{i1}, \dots, \mu_{in}$ belong to the null set $\cA_i$, or equivalently, no more than $u-1$ of them are signals. Setting $u =1$  amounts to testing the \textit{global null hypothesis} that $H_{ij}$ is true for all $j \in [n]$, an inferential goal pursued in many past meta-analyses, but one that is the least demanding in terms of replicability. At the other end, setting $u = n$, the maximum replicability level, offers the  most conceptually appealing yet strictest benchmark for assessing the generalizability of scientifically interesting findings.


Assume the primary data available from the  studies allow a statistician to directly access or otherwise construct a   $p$-value $p_{ij}$ for each $H_{ij}$. For  meta-analyses like \exref{DGE}, the  gene expression data  are often available on public repositories such as the Gene Expression Omnibus\footnote{\url{https://www.ncbi.nlm.nih.gov/geo/}} (GEO)  and can be processed into  $p$-values using bioinformatics pipelines such as \texttt{limma} \citep{Ritchie2015}. 
In such large-scale experiments with many features (e.g. the $6,587$ genes  in \exref{DGE}),  it is natural to aim for controlling the \textit{false discovery rate} (FDR), i.e. the expected proportion of Type I errors among the discoveries, when simultaneously testing the PC nulls in \eqref{PC_null} for all  $i \in [m]$. In the present context, suppose  $\widehat{\cR} \subseteq [m]$ is a set of rejected PC nulls determined with  the available data, the FDR for $\widehat{\cR}$ is then defined as 
\begin{equation*} 
    \FDR_{\rep} = \FDR_{\rep}(\widehat{\cR}) \equiv \mathbb{E}\left[ \frac{ \sum_{i \in \cHun} \1 \{ i \in \widehat{\cR} \} }{ 1 \vee \sum_{i \in [m]} \1 \{ i \in \widehat{\cR} \} } \right],
\end{equation*}
where $a \vee b \equiv \max(a,b)$ for any real numbers $a,b \in \mathbb{R}$ and
\begin{equation} \label{PC_true_null_set}
    \cHun \equiv \{i \in [m]: H_i^{u/[n]} \text{ is true} \}
\end{equation}
is the set of all features whose PC null hypotheses are true. 

A ``vanilla'' protocol for controlling the $\FDR_{\rep}$  is to apply the Benjamini-Hochberg (BH) procedure \citep{Benjamini1995} to a set of $m$ $p$-values for assessing the PC nulls, also commonly referred to as partial conjunction (PC) $p$-values. In the literature, a PC $p$-value for testing $H_i^{{}_{u/[n]}}$, denoted as $p^{{}_{u/[n]}}_i$, is typically constructed by combining the base $p$-values $p_{i1}, \dots, p_{in}$  following the formulation proposed by \citet{Benjamini2008} and later extended by \citet{Wang2019}.  Although seemingly natural, this vanilla approach is known to have low power; this is  not only due to the multiplicity correction required for the FDR control, but also due to the inherent limitations of PC p-values. In \secref{deprived_power_sim}, we will re-examine this issue with a particular focus on  the maximum replicability level $u=n$, which represents the most natural benchmark for replicability of discoveries  a typical practitioner may wish for, but unfortunately also poses the most power-deficient testing problem.

The ParFilter is our  proposed method to  powerfully test for  replicated  signals in the  meta-analytical setting above with strong (finite-sample)  $\FDR_\rep$ control guarantees. In a nutshell, as the name suggests,  it operates by first \textbf{par}titioning the $n$ studies into different groups and  borrowing the between-group information to \textbf{filter} out features that are unlikely to be $u/[n]$ replicable; this reduces the multiplicity correction required and set the scene for high testing power. Moreover, it can incorporate useful \emph{side information} provided by auxiliary covariates, on top of the  $p$-values, to decide on the features that are   $u/[n]$ replicable. Specifically, for each study $j$, the statistician may optionally  have access to  a covariate $x_{ij} \in \mathbb{R}^{d_j}$  that is also informative for  testing $H_{ij}$. For instance, in \exref{DGE}, $x_{ij}$ can be taken as the differential expression of gene $i$ in cells from a different part of the thymus other than the medulla, such as the  cortex (\secref{realdata}). 
   This  reflects the recent trend in the multiple testing literature of integrating side information to enhance power.

Therefore, the ParFilter seamlessly marries two  lines of current  research -- replicability analysis and covariate-adaptive inference, and we highlight our  contributions in context: 
  \begin{enumerate} [(i)]
 \item  The ParFilter substantially generalizes the two $\FDR_\rep$ controlling procedures of   \citet[Sec. 3.2 \& 4.2]{Bogolomov2018}, which were only designed for  assessing $2/[2]$ replicability using the filtering strategy. Precisely, in the special case where $n = u = 2$ and covariates are unavailable, our \methodname  boils down to them under particular tuning parameter settings; see \secref{PF}. Their powerful procedures  provide   finite-sample  $\FDR_\rep$ control guarantees   under standard assumptions  \citep[Theorems 2 \& 4]{Bogolomov2018}, and can therefore be considered as gold standards to be pursued for further generalization.  
 
Nonetheless,  the original proofs for $\FDR_\rep$ control  in \citet{Bogolomov2018}  involve bounding sums of probabilities over intricate intersecting events, which appears too difficult to be directly extended for more than two studies. As a technical innovation, we take a different route by harnessing  techniques inspired by the recent developments in \citet{pfilter2019} and \citet{Katsevich2023} for  ``dependency control''  \citep{Blanchard2008}, and provide relatively transparent proofs of \methodnamens's finite-sample $\FDR_{\rep}$ guarantees  (\secref{main_results}) for arbitrary $u \leq n$.

 \item To  our best knowledge, the \methodname is   the only covariate-adaptive  procedure  specialized for $\FDR_\rep$ controlling multi-study replicability analysis thus far.  Due to their useful side information offered for testing, auxiliary covariates have been incorporated into multiple testing within a single study for some time. The related literature encompasses  Type I error rates beyond the FDR and resists a succinct summary; we refer to \citet{Lei2018, Ignatiadis2021, Leung2022, rencandesAOAS2023,  freestone2024semi} for examples of covariate-adaptive testing procedures with finite-sample FDR guarantees that are most relevant to our present work.

 To boost  testing power while ensuring its finite-sample $\FDR_\rep$  guarantees, the \methodname trains informative hypothesis weights, which can encourage rejections of features more likely to contain replicated signals,  from the covariate-augmented data 
 in a  manner that naturally respects the partitioning structure of the studies  in the filtering step (\secref{powerful_weights}). This novel  approach  shares a similar spirit as  the ``cross-weighting'' approach in the single-study FDR work of    \citet{Ignatiadis2021}, where  hypothesis weights are cross-trained on randomly or manually defined $p$-value folds alongside the covariates. 

 \end{enumerate}

We  also compare with two existing   methods  that  employ the general idea of filtering in different ways for power enhancement in $\FDR_{\rep}$ controlling testing. The  AdaFilter-BH procedure of \citet[Section 4.2]{adafilter}  first filters out  features that are unlikely to exhibit $(u-1)/[n]$ replicability and then computes a final rejection set for $u/[n]$ replicability from the remaining features. While being powerful, since its rejection set is derived from the same $p$-values used for  filtering, it only offers  asymptotic $\FDR_{\rep}$ control under  limiting  assumptions. On the other hand, the CoFilter of \citet{Dickhaus2024}, like the ParFilter, offers  finite-sample $\FDR_\rep$ control for an arbitrary number of studies $n$. It forms a selection set $\cS = \{ i \in [m]: p^{{}_{u/[n]}}_i \leq \tau \}$ using a pre-determined threshold $\tau \in (0,1)$, and then runs a BH procedure on the \textit{selection-adjusted} PC $p$-values in $\{ p^{u/[n]}_i/\tau : i \in \cS \}$. While its power benefits from filtering, the gain is limited because selection-adjusted PC $p$-values are more conservative than ordinary PC $p$-values. Further  comparisons will be explored in our simulation studies  (\secref{sim}).

Other inferential methods for partial conjunction hypotheses  not already discussed include those of \citet[Sec. 5.3]{Bogomolov2023} and \citet[Sec. 4(b)]{Benjamini2009}, which provide lower bounds on the number of studies in which feature $i$ shows an effect, while controlling the false coverage rate \citep{Benjamini2005}.  The methods proposed in \citet{Chung2014} and \citet{Heller2014} provide replicability inferences specifically under the popular two-group modeling assumption of \citet{Efron2001}. \citet{Li2021} developed a method --- building on the model-X framework of \citet{Candes2018} --- to assess replicability in a regression setting, where the goal is to determine whether explanatory variables are consistently associated with a response variable across $n$ sampled environments. \citet[Sec. 3.2 \& 4.2]{Sun2015} analyzed spatial signal clusters by formulating them as partial conjunction testing problems.  These references are certainly not exhaustive, and we refer to \citet{BogomolovHeller2023} for a comprehensive review. 


Our paper is structured as follows: \secref{PC} reviews PC $p$-values and re-examines their power limitations. \secref{MainPF} provides the layout of the \methodnamens, outlining its implementation and finite-sample FDR guarantees (\thmsref{error_control} and \ref{thm:error_control_dep}). \secref{sim} demonstrates the competitiveness of the \methodname against existing $\FDR_{\rep}$ procedures using simulated data. \secref{realdata} explores the \methodnamens's performance in the replicability analysis of the real data from \exref{DGE}. Finally, \secref{conclusion} concludes with a discussion on potential extensions. An R package implementing the \methodname is available at \url{https://github.com/ninhtran02/Parfilter}. 

\subsection{Additional notation}\label{sec:notation} 
$\Phi(\cdot)$ denotes a standard normal distribution function.   For two vectors ${\bf a}$ and ${\bf b}$ of the same length, $\textbf{a} \matleq \textbf{b}$ (resp. $\textbf{a} \matgeq \textbf{b}$) denotes that each element in $\textbf{a}$ is less (resp. greater) than or equal to its corresponding element in $\textbf{b}$;  for two matrices $\textbf{A}$ and $\textbf{B}$ of equal dimensions, $\textbf{A} \matleq \textbf{B}$ and $\textbf{A} \matgeq \textbf{B}$ are defined analogously by comparing their corresponding elements.  For any two real numbers $a,b \in \mathbb{R}$, $a \vee b \equiv \max(a,b)$ and $a \wedge b \equiv \min(a, b)$. 
  Given a natural number $\ell \in \mathbb{N}$, a real valued function $f: \mathcal{D}  \to \mathbb{R}$ defined on  a subset $\mathcal{D} \subseteq \mathbb{R}^{\ell} $  is said to be non-decreasing (resp. non-increasing) if $f(\mathbf{x}_1) \leq f(\mathbf{x}_2)$ (resp. $f(\mathbf{x}_1) \geq f(\mathbf{x}_2)$) for any vectors $\mathbf{x}_1,\mathbf{x}_2 \in \mathbb{R}^{\ell}$ such that $\mathbf{x}_1 \matleq \mathbf{x}_2$. For a vector of real numbers $\mathbf{a} \in \mathbb{R}^{\ell}$,   the vector $(1,\mathbf{a}) \in \mathbb{R}^{1 + {\ell}}$ denotes the concatenation of $1$ and $\mathbf{a}$.  For $\ell_1,\ell_2 \in \mathbb{N}$ such that $\ell_1 \leq \ell_2$, let $[\ell_1:\ell_2] \equiv \{ \ell_1,  \ell_1 + 1, \dots, \ell_2 \}$. For a set of numbers $\cJ \subseteq \bN$, let $|\cJ|$ denote its cardinality.

\section{Multiple testing with PC $p$-values}\label{sec:PC}
We first review some fundamentals of partial conjunction testing. Consider a generic context where  $\theta_1, \dots, \theta_n$ are $n$ parameters of interest and $\Theta_j$ is a set of null values for $\theta_j$ for each $j \in [n]$. Let $p_1,\dots,p_n \in [0,1]$ be the base $p$-values for testing the null hypothesis $ \theta_1 \in \Theta_1$, $\dots$, $\theta_n \in \Theta_n$ respectively, and suppose these $p$-values are \textit{valid}, i.e.,
\begin{equation} \label{generic_pv_valid}
    \Pr(p_j \leq t)\leq t \text{ for all } t \in [0, 1]  \text{ if the null hypothesis }
 \theta_j \in \Theta_j \text{ is true}. 
\end{equation}
For a given replicability level $u \in [n]$, the partial conjunction (PC) null hypothesis that at least $n- u + 1$ parameters among $\theta_1,\cdots,\theta_n$ are null is denoted as
\begin{equation*}
    H^{u/[n]}: |\{ j \in [n] : \theta_j \in \Theta_j \} | \geq n - u + 1.
\end{equation*}
In the literature, $H^{u/[n]}$ is typically tested by combining the base $p$-values $p_1,\dots,p_n$ to create a \textit{combined} $p$-value that is both \textit{valid} and \textit{monotone}:
 
\begin{definition}[Valid and monotone combined $p$-values for testing $H^{u/[n]}$]\label{def:monotone_valid_combined_pc_pv}
For a given combining function  $f: [0,1]^n \to [0,1]$, the combined $p$-value $f(p_1,\dots,p_n)$ is said to be
\begin{enumerate}[\indent (a)]
\item  valid for testing $H^{u/[n]}$ if 
    $\Pr(f(p_1,\dots,p_n) \leq t) \leq t  \text{ for all }    t \in [0, 1] \text{ under $H^{u/[n]}$, and }$
\item monotone if $f$ is non-decreasing.
\end{enumerate}
\end{definition}

The combined $p$-value properties outlined in \defref{monotone_valid_combined_pc_pv} are minimal because validity is required for ensuring Type I error control, and monotonicity is sensible for scientific reasoning; a more nuanced discussion on monotonicity can be found in \citet[Sec 3.2]{Wang2019}. For testing $H^{1/[n]}$, i.e., the \textit{global null hypothesis} that all $\theta_1 \in \Theta_1, \dots, \theta_n \in \Theta_n$ hold true, perhaps the most well known combining function is 
 \[
 f_{\text{Bonferroni}}(p_1, \dots, p_n)  \equiv (n \cdot \min(p_1, \dots, p_n) ) \wedge 1.
 \]
In words, $f_{\text{Bonferroni}}$ applies a Bonferroni correction to the minimum base $p$-value, and the resulting combined $p$-value is usually a very conservative. More powerful functions for combining $p$-values to test global nulls are available, and a comprehensive list of such functions is provided in \citet[Section 2.1]{Bogomolov2023}. Below, we present three common choices:
\begin{itemize}
 \item \textbf{\citet{Fisher1973}'s combining function:}
    \begin{equation*}
        f_{\text{Fisher}}(p_{1},\dots,p_n) \equiv 1 - F_{\chi^2_{(2 n)}} \left( -2 \sum_{j \in [n]} \log(p_{j}) \right),
    \end{equation*}
    where $F_{\chi^2_{(2 n)}}$ is the chi-squared distribution function  with $2  n$ degrees of freedom.
      \item \textbf{\citet{Stouffer1949}'s combining function:}
    \begin{equation*}
        f_{\text{Stouffer}}(p_{1},\dots,p_n) \equiv 1 - \Phi \left( \frac{\sum_{j \in [n]} \Phi^{-1}(1 - p_{j})}{\sqrt{n}} \right);
    \end{equation*}
   
    \item \textbf{\citet{Simes1986}'s combining function}:
    \begin{equation*}
        f_{\Simes}(p_{1},\dots,p_n) \equiv \min_{j \in [n]} \frac{n p_{(j)}}{ j },
    \end{equation*}
     where $p_{(1)}\leq \dots \leq p_{(n)}$ are the ordered statistics for $p_{1} \leq \dots \leq p_n$.
  \end{itemize}
Assuming $p_1, \dots, p_n$ are independent and property \eqref{generic_pv_valid} holds, all four combining functions above (i.e., Bonferroni, Fisher, Stouffer, and Simes) give valid combined $p$-values for testing $H^{1/[n]}$. Moreover, these combining functions are all evidently non-decreasing.

To test $H^{u/[n]}$ for a general $u > 1$, the corresponding combined $p$-value, commonly referred to as a PC $p$-value, is typically derived using the \textit{generalized Benjamini-Heller} combining method \citep{Benjamini2008, Wang2019}. The resulting PC $p$-value is known as a \textit{GBHPC} (\textbf{g}eneralized \textbf{B}enjamini-\textbf{H}eller \textbf{p}artial \textbf{c}onjunction) $p$-value:

\begin{definition}[GBHPC  $p$-value for testing $H^{u/[n]}$]  \label{def:gbhpc} 
Suppose $u \in [n]$ is given.
For each subset $\cJ = \{j_1, \dots, j_{n-u+1}\}\subseteq [n]$ of $n - u +1$ ordered elements  $1 \leq j_1 < \dots <  j_{n - u +1} \leq n$, let 
\[
\mathbf{p}_\cJ = (p_{j_1}, \dots, p_{j_{n - u +1}}) \text{ and } {\boldsymbol \theta}_\cJ = (\theta_{j_1}, \dots, \theta_{j_{n - u +1}})
\]
 be the  vectors  of $p$-values and parameters corresponding to $\cJ$. Moreover, let 
 \[
 f_{\cJ}: [0, 1]^{n-u+1} \longrightarrow [0, 1]
\]
 be a non-decreasing function where $f_\cJ( \mathbf{p}_\cJ)$ is valid for testing 
 \[H^{1/\cJ}: {\boldsymbol \theta}_\cJ  \in \Theta_{j_1} \times \cdots \times \Theta_{j_{n-u+1}},\]
   the restricted global null hypothesis that all the base nulls corresponding to $\cJ$ are true. Define the GBHPC $p$-value for testing $H^{u/[n]}$ to be
\begin{equation}\label{u_combine_p}
   f^\star(p_1,\dots,p_n;u)\equiv 
     \max \{ f_\cJ( \mathbf{p}_\cJ): \mathcal{J} \subseteq [n]  \text{ such that } |\cJ| = n-u+1\},
\end{equation}
i.e., the largest combined $p$-value $f_\cJ( \mathbf{p}_\cJ)$ for testing the restricted global nulls $H^{1/\cJ}$ among all $\cJ \subseteq [n]$ with cardinality $|\cJ| = n-u+1$.
\end{definition}
As a concrete example, if we take each $f_\cJ$ in \eqref{u_combine_p} to be Stouffer's combining function in $n - u +1$ arguments, then
\begin{equation}\label{stouffer_GBHPC}
f^\star(p_1,\dots,p_n) = f_{\text{Stouffer}}(p_{(u)}, \dots, p_{(n)}) = 1 - \Phi \left( \frac{\sum^n_{j = u} \Phi^{-1}(1 - p_{(j)})}{\sqrt{n - u + 1}} \right).
\end{equation}
Intuitively, \eqref{stouffer_GBHPC} makes sense as a $p$-value for testing $H^{u/[n]}$ because it is a function of the largest $n - u +1$ base $p$-values only. By \citet[Proposition 1]{Wang2019}, GBHPC $p$-values (\defref{gbhpc}) are valid for testing $H^{u/[n]}$ (\defref{monotone_valid_combined_pc_pv}$(a)$) because every $f_\cJ(\mathbf{p}_\cJ)$ in \eqref{u_combine_p} is valid for testing their corresponding restricted global null  $H^{1/\cJ}$; moreover, they are monotone (\defref{monotone_valid_combined_pc_pv}$(b)$) because every $f_\cJ$ in \eqref{u_combine_p} is non-decreasing. Under mild conditions, \citet[Theorem 2]{Wang2019} has also demonstrated that any admissible monotone combined $p$-value must necessarily take the form of \eqref{u_combine_p}.

\begin{figure}
    \centering
    \includegraphics[width=0.75\linewidth]{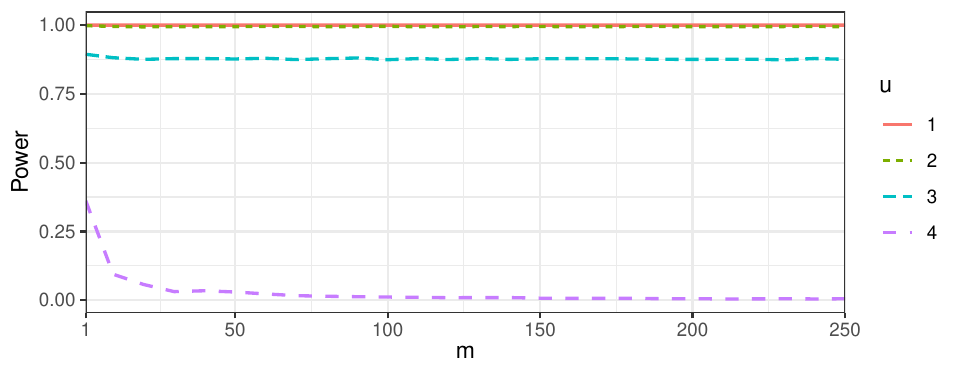}
    \caption{Power of the BH procedure (with $\FDR$ target $0.05$) applied to $m = 1,10,20,\cdots,250$ PC $p$-values under   replicability levels $u = 1, 2, 3, 4$, based on a simulation experiment with a total of $n = 4$ studies. By design, $H^{u/[n]}_i$ is false for each $(i, u) \in [m] \times [4]$, and its corresponding PC $p$-value $p^{u/[4]}_i$ is constructed from the base $p$-values $p_{i1}, \dots, p_{i4}$ using the generalized Benjamini–Heller formulation in \defref{gbhpc}, with Stouffer’s combining function.}
    \label{fig:PCdemo}
\end{figure}

  \subsection{The deprived power of partial conjunction test when $u = n$} \label{sec:deprived_power_sim}
 
Maximum replicability ($u = n$) is the most natural and appealing requirement for ensuring the generalizability of scientific findings. However, it also posits the most power-deprived partial conjunction testing problem, because the underlying composite null parameter space of $H^{u/[n]}$ becomes larger as $u$ increases \citep[Sec. 2.1]{BogomolovHeller2023}. What may be surprising is how disproportionately underpowered the case of $u = n$ can be, relative to other cases with $u < n$, especially when multiple PC nulls $H_1^{{}_{u/[n]}}, \dots, H_m^{{}_{u/[n]}}$ are to be tested simultaneously as in \secref{intro}. 

To illustrate, consider a simple scenario with $n = 4$ studies, where all the base null hypotheses $H_{ij}$ are \emph{false} with their corresponding $p$-values generated as 
\[
p_{ij} \overset{\tiny \text{iid}}{\sim} \text{Beta}(0.26,7) \text{ for each } (i, j) \in [m] \times [4],
\]
which are highly right-skewed and stochastically smaller than a uniform distribution. At a given replicability level $u = 1, \dots, 4$, we assess $H_i^{u/[4]}$ by forming a PC $p$-value $p^{{}_{u/[4]}}_i$ from $p_{i1}, \dots, p_{i4}$ for each feature $i$, using the Stouffer-GBHPC formulation given in \eqref{stouffer_GBHPC}.  We then apply the BH procedure to the $m$ PC $p$-values at a nominal FDR level of 0.05. For various numbers of hypotheses tested, i.e., $m \in \{ 1, 10, 20, \dots, 250 \}$, \figref{PCdemo} shows the power of this standard testing procedure, estimated by the average proportion of rejected PC null hypotheses out of the total $m$ over 500 repeated experiments. 

For a single test at $m = 1$, the power  is $1.00$, $1.00$, $0.89$, and $0.37$ for $u = 1, 2, 3, 4$, respectively; the significant drop in power from $u = 3$ to $u = 4$ highlights how highly conservative a single PC $p$-value can already be under $u = n$. As $m$ increases, the power for testing $4/[4]$ replicability further drops rapidly towards $0$, while the power for other replicability levels remains quite stable throughout. This illustrates that, in addition to their poor individual power,  PC $p$-values at maximum replicability level $u = n$ can be extremely sensitive to standard multiplicity correcting methods such as the BH procedure, and yield power too low to be detected. 

In \appref{deprived_power}, we repeat these simulations with $n$ varying from 2 to 8, and also examine the power when Fisher’s combining function is used in place of Stouffer’s in the GBHPC $p$-value formulation of $p^{{}_{u/[n]}}_1,\dots,p^{{}_{u/[n]}}_m$. The results
 there
exhibit the same pattern: a sharp drop in power when moving from $u = n - 1$ to $u = n$, and a rapid decline to zero power as $m$ increases when $u = n$. A general discussion on the low power of combined $p$-values in partial conjunction hypothesis testing can be found in \citet[Sec. 1.3]{Liang2025} and the references therein. However, a complete theoretical explanation or lower bound for the sharp drop in power of a GBHPC $p$-value at $u = n$ currently eludes us, though we suspect that the results of \citet{Weiss1965, Weiss1969} on the distribution of sample spacings $p_{(j+1)} - p_{(j)}$ for $j \in [n-1]$ may prove useful in developing one.

\section{The \methodname}\label{sec:MainPF}
In this section, we introduce our \methodname for simultaneously testing the PC nulls  in \eqref{PC_null}. In the sequel, we will use 
\begin{equation*}
\mathbf{P} \equiv (p_{ij})_{i \in [m],j \in [n]} \quad \text{ and } \quad \mathbf{X} \equiv (x_{ij})_{i \in [m], j \in [n]}
\end{equation*}
to respectively denote the $m \times n$ matrices of  $p$-values and covariates for the base hypotheses in \eqref{individual_hypothesis}, with the understanding that the covariates  $x_{ij}$ within the same study $j$  have a common dimension $d_j \geq 1$. If covariates are unavailable for a given study $j$, the $j$-th column of $\mathbf{X}$ is  taken to contain entries of $0$'s as constants. 

We  also need to denote submatrices or subvectors of $\mathbf{P}$ or $\mathbf{X}$. If we are given  $i \in [m]$, $j \in [n]$, $\cI \subseteq [m]$ and $\cJ \subseteq [n]$ as row or column indices, then $(-i)$, $(-j)$, $(-\cI)$ and $(-\cJ)$ are shorthands for their complements $[m] \backslash \{i\}$, $[n] \backslash \{j\}$, $[m] \backslash \cI$ and $[n] \backslash \cJ$ respectively, and ``$\cdot$'' represents either all the rows or all the columns. For examples,  $\mathbf{P}_{i \cdot} \equiv (p_{i1},\cdots,p_{in})$ denotes the $i$th row  of $\mathbf{P}$ as a vector,  $\mathbf{P}_{\cdot \cJ}$ denotes the whole matrix $\mathbf{P}$ restricted to the columns in  $\cJ$,  $\mathbf{P}_{i \cJ}$ denotes the $i$th row of $\mathbf{P}_{\cdot \cJ}$, 
and $\mathbf{P}_{(-\cI) \cJ}$ denotes $\mathbf{P}_{\cdot \cJ}$ restricted to the rows in $[m] \backslash \cI$. Other such submatrices or subvectors should be self-explanatory from their context. 



These elementary  assumptions  will ground  our theoretical results throughout:
 
 \begin{assumption}[Independence across studies conditional on ${\bf X}$]  \label{assump:indepstudies}
$\mathbf{P}_{\cdot 1}, \dots, \mathbf{P}_{\cdot n}$ are mutually independent 
conditional on $\mathbf{X}$.
 \end{assumption}

\begin{assumption}[Valid base $p$-values]\label{assump:superuniformp}
    Each base $p$-value in $\mathbf{P} \in [0,1]^{m \times n}$ satisfies
    \begin{equation*}
    \Pr(p_{ij} \leq t) \leq t \quad \text{for all $\ t \in [0,1]$ whenever $H_{ij}$ is true}.
\end{equation*}
    In other words, $p_{ij}$ is stochastically no smaller than a uniform distribution  (i.e. superuniform) when $H_{ij}$ is true. 
\end{assumption}

\begin{assumption}[Independence between $p_{ij}$ and ${\bf X}$ under the null]\label{assump:pxindependence}
Given $H_{ij}$ is true, then $p_{ij}$ is independent of ${\bf X}$.  
\end{assumption}

 When the covariates can be considered as deterministic, such as our case study in \secref{realdata}  where the covariates are computed from external datasets,  \assumpref{indepstudies}  is natural  since meta-analyses typically consist of independently conducted studies. In other occasions, within each study, the covariates are induced from the same data used to compute the base $p$-values and are  more appropriately treated as random.  An example is   \citet{bottomly2011evaluating}'s well known dataset in the single-study  FDR testing literature \citep{Lei2018, Leung2022,Zhang2022, ignatiadis2016data}, where for each gene (a feature), the $p$-value corresponds to its differential  expression  between two mouse strains and the covariate is the mean normalized read counts across samples. With multiple studies, as long as the $n$ covariate-augmented datasets $\{{\bf P}_{\cdot 1}, {\bf X}_{\cdot 1}\}, \cdots \{{\bf P}_{\cdot n}, {\bf X}_{\cdot n}\}$ are independent, \assumpref{indepstudies} is obviously satisfied.

The superuniformity of null $p$-values  (\assumpref{superuniformp}) is a standard starting point for hypothesis testing.  In virtually all covariate-adaptive testing literature in the single-study context, it has always been assumed that the covariates do not influence the distribution of a null $p$-value as a   ``principle for information extraction'' \citep{CARS2019}, and such covariates naturally exist in many applications 
\citep{ignatiadis2016data, boca2018direct}. \assumpref{pxindependence}  is therefore a multi-study generalization of this principle, and is trivial if ${\bf X}$ is deterministic. If ${\bf X}$ is random and the covariate-augmented datasets $\{{\bf P}_{\cdot 1}, {\bf X}_{\cdot 1}\}, \cdots \{{\bf P}_{\cdot n}, {\bf X}_{\cdot n}\}$ are independent, \assumpref{pxindependence}  can be restated as ``\emph{$p_{ij}$ is independent of ${\bf X}_{\cdot j}$ when $H_{ij}$ is true}''  as in the single-study case; see \citet[Assumption 1(b) and 2(b)]{Ignatiadis2021} for instance.

%


The \methodname begins by specifying a \textit{testing configuration for $u/[n]$ replicability}.

\begin{definition}[\text{Testing configuration for $u/[n]$ replicability}]\label{def:partition}
For a given number $K \in [u]$, let 
\begin{equation*}
    \cG_1,\cG_2,\dots,\cG_K \subseteq [n]
\end{equation*}
be disjoint non-empty subsets, referred to as groups, that form a partition of the set $[n]$.
Let $w_1,\cdots,w_K \in (0,1]$ be numbers, referred to as local error weights, that sum to one, i.e.,
\begin{equation*}
    \sum_{k \in [K]} w_k = 1.
\end{equation*}
For each feature $i \in [m]$, let $\mathbf{u}_i \equiv (u_{i1},\dots,u_{iK})$ be a vector of $K$ natural numbers, referred to as local replicability levels, that satisfies
\begin{align}\label{uvec}
    u_{ik} \in [|\cG_k|]  \text{ for all } k \in [K]  \quad \text{ and } \quad  \sum_{k \in [K]} u_{ik} = u.
\end{align}
Then the combination
\begin{equation}\label{testing_config_notation}
	\{ \cG_k, w_k , \mathbf{u}_i\}_{k \in [K], i \in [m]}
\end{equation}
is known as a testing configuration for $u/[n]$ replicability. This configuration can be fixed or dependent on the covariates $\mathbf{X}$, but must not be dependent on the base $p$-values $\mathbf{P}$.
\end{definition}

With respect to a testing configuration  \eqref{testing_config_notation}, for each feature $i$ in a given group $k$, one can  define the \textit{local} PC null hypothesis  
\begin{equation}\label{local_PC_null}
    H_i^{\uiG}:  |\{ j \in \cG_k: \mu_{ij} \in \mathcal{A}_i \}| \geq |\cG_k| - u_{ik} +1,
\end{equation}
which asserts that there are no more than $u_{ik} -1$ true signals among the $|\cG_k| $ effects, as well as its corresponding \emph{local} PC $p$-value
\begin{equation} \label{generic_local_pv}
    p^{\uiG}_i \equiv f_{ik} (\mathbf{P}_{i \cG_k}; u_{ik}),
\end{equation}
where $f_{ik}$ is a function that combines the base $p$-values in $\mathbf{P}_{i \cG_k}$ and also depends on $u_{ik}$. The \methodname will leverage the  property associated with a given testing configuration in the following lemma to control the $\FDR_{\rep}$ when testing the PC nulls  in \eqref{PC_null}:

\begin{lemma}[$\FDR_{k}$ control $\Rightarrow$ $\FDR_{\rep}$ control]\label{lem:FDR_rep_control_from_partition}
Suppose $\{ \cG_k, w_k , \mathbf{u}_i\}_{k \in [K], i \in [m]}$ is a testing configuration  and 
\begin{equation*}
    \cH_k \equiv \{ i \in [m]: \text{$H^{\uiG}_i$ is true} \}
\end{equation*}
is the set of true local PC nulls for each group $k \in [K]$. Let $\widehat{\cR} \subseteq [m]$ be an arbitrary data-dependent rejection set and 
  \begin{equation*}
    \FDR_{k} = \FDR_{k}(\widehat{\cR}) \equiv \mathbb{E} \left[ \frac{\sum_{i \in  \cH_k}  I \{ i \in \widehat{\cR} \}}{|\widehat{\cR}| \vee 1} \right]
\end{equation*}
be the FDR entailed by rejecting the local PC nulls in group $k$ for the features in $\widehat{\cR}$. For an $\FDR_{\rep}$ target $q \in (0,1)$, if $\widehat{\cR}$ is such that $\FDR_{k}(\widehat{\cR}) \leq w_k \cdot q \text{ for each } k \in [K]$,
    then $\FDR_{\rep} (\widehat{\cR})\leq q$.
\end{lemma}

\lemref{FDR_rep_control_from_partition} is proved in \appref{FDR_k_control_gives_FDR_rep_control}. In addition to the testing configuration,  for each group $k$, the \methodname  employs two related functions $f_{\cS_k}$ and $f_{\nubf_k}$ that take $\mathbf{P}$ and $\mathbf{X}$ as inputs: The first function output a subset 
\begin{equation} \label{selected_set}
  \cS_k \equiv f_{\cS_k}(\mathbf{P}; \mathbf{X}) \subseteq [m]
\end{equation}
that aims to have as few intersecting features with the set of null features $\cH^{u/[n]}$ in \eqref{PC_true_null_set}  as possible, and the second function creates non-negative weights $\nu_{1k}, \dots, \nu_{mk} \in [0, \infty)$ by outputting a vector
\begin{align}\label{nufun}
    \nubf_k =  (\nu_{1k},\nu_{2k},\dots,\nu_{mk}) \equiv  f_{\nubf_k}(\mathbf{P};\mathbf{X}) \quad \text{ that satisfies } \quad \sum_{\ell \in \cS_k} \nu_{\ell k} = |\cS_k|.
\end{align}
Finally, for given fixed tuning parameters
\[
\lambda_1, \lambda_2, \dots,\lambda_K \in (0,1),
\]
the ParFilter considers the following \textit{candidate} rejection set
\begin{multline}\label{RrepEasy}
  \cR(\tvec) \equiv \bigcap_{k \in [K]} \left\{ i \in \cS_k :  p^{\uiG}_{i} \leq (\nu_{ik} \cdot t_k) \wedge \lambda_k \right\}  \\
  \text{ for any vector of thresholds } {\bf t} = (t_1, \dots, t_K) \in [0,\infty)^{K},
\end{multline}
and decides on a data-dependent  vector $\hattvec = (\hatt_1,\dots,\hatt_K)$, which is  plugged into \eqref{RrepEasy} to obtain a final rejection set $\widehat{\cR} = \cR(\hattvec)$. Under suitable assumptions, $\cR(\hattvec)$ will control $\FDR_{k} (\cR(\hattvec))$ below $w_k \cdot q$ for each $k \in [K]$, and therefore control $\FDR_\rep (\cR(\hattvec))$ below $q$ by virtue of \lemref{FDR_rep_control_from_partition}. 

In the following subsections, we will discuss the specifics of how  the testing configuration, local PC $p$-values, $f_{\cS_k}$,  $f_{\nubf_k}$, and $\hattvec$ should be formed. But before getting into these downstream details, one can  already observe from our overview above how the \methodname derives its testing power: 
\begin{enumerate} [(i)]
\item  The selected sets $\cS_1, \dots, \cS_K$ serve to filter out unpromising features that have little hope of demonstrating $u/[n]$ replicability. If the \emph{selection rules}  $f_{\cS_k}$'s can successfully exclude many features belonging to  $\cH^{u/[n]}$, we can substantially reduce the amount of multiplicity correction required when determining the data-dependent threshold $\hattvec$ for the rejection set. 
\item The \emph{local PC weights} $\nubf_1, \cdots, \nubf_K$ rescale the rejection thresholds, as seen in \eqref{RrepEasy}. A simple choice  is to set $\nu_{ik} = 1$ for every $i$ and $k$. However, if the weights can be learned from the data, ideally also with the help of auxiliary covariates, in such a way that $\nu_{ik}$ is large for $i \notin \cH^{u/[n]}$ and small for $i \in \cH^{u/[n]}$, then the final reject set $\cR(\hattvec)$ renders a more powerful procedure.

\end{enumerate}

\subsection{Testing configuration for the \methodname}\label{sec:choose_grouping_config}
Our suggestion for setting the testing configuration for maximum replicability is straightforward:
\begin{example}[Simple testing configuration for $u = n$]\label{ex:suggested_testing_config_ueqn}
    Let $K = n$. For each $k \in [n]$, let
    \begin{equation*}
        \cG_k = \{ k \}, \quad w_k = 1/n, \quad \text{and} \quad  u_{ik} = 1 \quad \text{for each $i \in [m]$ }.
    \end{equation*}
\end{example}
Since we primarily focus on the power-derived case of $u = n$, we defer our suggested testing configuration for the case of $u < n$ to \appref{uleqn}. Adapting a testing configuration to the covariates $\mathbf{X}$ is permitted by \defref{partition} and may lead to improved power. However, developing a general methodology for creating an adaptive testing configuration is challenging, as covariates can vary drastically in type (e.g., continuous, categorical, etc.), dimensionality, and level of informativeness across different scientific contexts. Based on our experiments, we find that the  testing configuration provided in \exref{suggested_testing_config_ueqn} performs well across various settings, particularly when paired with the selection rules and the local PC weights proposed in \secsref{selection}  and \secssref{powerful_weights}.

\subsection{Local PC $p$-values for the \methodnamens}\label{sec:local_PC}
The local PC $p$-values  in \eqref{generic_local_pv} should generally satisfy the following condition for our theoretical results established in \secref{main_results}. It resembles the validity property in \defref{monotone_valid_combined_pc_pv}$(a)$ but additionally account for the covariates:

\begin{conditionalt}[Validity of $p^{\uiG}_i$ conditional on ${\bf X}$]\label{cond:valid_monotone_local_PC_pv}
For any given feature $i$ and group $k$, its corresponding local PC $p$-value $p^{\uiG}_i$ defined in \eqref{generic_local_pv} is said to be valid for testing $H_i^{\uiG}$ conditional on ${\bf X}$ if
\begin{equation*}
    \Pr( p^{\uiG}_i \leq   t \mid \mathbf{X} ) \leq  t \text{ for all }  t \in [0,1] \ \text{when $H_i^{\uiG}$ is true}.
\end{equation*}

\end{conditionalt} 

Under our basic assumptions, defining the combining function $f_{ik}$ in \eqref{generic_local_pv} with the generalized Benjamini-Heller combining method (\defref{gbhpc}) and  the four examples of combining functions for testing global nulls mentioned in \secref{PC} can render conditionally valid local PC $p$-values.


\begin{lemma}[Validity of local GBHPC $p$-values conditional on ${\bf X}$] \label{lem:construct_local_PC}
Suppose \assumpsref{indepstudies}--\assumpssref{pxindependence} hold.  
 Given a testing configuration for $u/[n]$ replicability (\defref{partition}) and $(i, k) \in [m] \times [K]$, let
\begin{equation}\label{local_GBHPC}
    p^{\uiG}_i  = \max \left\{ f_{i, \cJ}( \mathbf{P}_{i \cJ}): \cJ \subseteq \cG_k \text{ and } | \cJ | = |\cG_k| - u_{ik} + 1 \right\},
\end{equation}
where each $f_{i, \cJ}: [0, 1]^{|\cG_k| - u_{ik} + 1} \longrightarrow [0,1]$ is one of the Bonferroni, Fisher, Stouffer, or Simes combining function in $|\cG_k| - u_{ik} + 1$  arguments.  Then   $p^{\uiG}_i$ meets \condref{valid_monotone_local_PC_pv}, i.e. it is valid for testing $H_i^{\uiG}$ conditional on ${\bf X}$.

%
%
%
\end{lemma}

The proof of \lemref{construct_local_PC} is given in \appref{valid_local_GBHPC_pvalue}. Note that the $p^{\uiG}_i$  in \eqref{local_GBHPC} is also monotone as in \defref{monotone_valid_combined_pc_pv}$(b)$, because the Bonferroni, Fisher, Stouffer and Simes combining functions  are all non-decreasing.

\subsection{Selection rules for the \methodname}\label{sec:selection}
The choice of the selection rule $f_{\cS_k}$ for each group $k$ is flexible, provided it satisfies the following two conditions for establishing the finite-sample $\FDR_{\rep}$ control guarantees of the \methodname  in \secref{main_results}.

\begin{conditionalt}[$\cS_k$  depends only on $\mathbf{P}_{ \cdot (-\cG_k)}$ and ${\bf X}$]\label{cond:indep}
For each $k \in [K]$, the selected set $\cS_k \equiv f_{\cS_k}(\mathbf{P};\mathbf{X})$ only depends on ${\bf P}$ via ${\bf P}_{\cdot (- \cG_k)}$. In other words, $\cS_k$ cannot be determined using knowledge of the base $p$-values in the same group $k$. 
\end{conditionalt}

\begin{conditionalt}[Stability of $f_{\cS_k}$]\label{cond:stable}
 For each $k \in [K]$ and  $i \in \cS_{k}\equiv f_{\cS_k} (\mathbf{P};\mathbf{X})$, fixing $\mathbf{P}_{(-i)\cdot}$ and $\mathbf{X}$ and changing $\mathbf{P}_{i \cdot}$ so that $i \in \cS_{k}$ still holds will not change $\cS_k$. In other words,  each selection rule $f_{\cS_k}$ is stable.
\end{conditionalt}

 Many common multiple testing procedures are stable (\condref{stable}) and can be employed as selection rules; see \citet{Bogomolov2023} for examples. For simplicity, we suggest the following thresholding selection rule which only depends on the $p$-values and can be easily verified to satisfy both \condsref{indep} and \ref{cond:stable}.

\begin{example}[Selection by simple PC $p$-value thresholding] \label{ex:suggested_select_rules}
A simple selection rule can be defined by taking
\begin{equation}\label{nonadaptiveSk}
    \cS_k = \bigcap_{\ell \in [K]\backslash \{k\} } \left\{ i \in [m] : p^{u_{i \ell}/\cG_{\ell}}_i \leq \left( w_{\ell} \cdot q \right) \wedge \lambda_{\ell}  \right\}  \quad \text{for each }  k \in [K]. 
\end{equation}
\end{example}

This suggested selection rule is an extension of what was adopted by \citet[Sec. 5]{Bogolomov2018} in their two-study special case where $n = u = K = 2$, $\cG_1 = \{ 1 \}$, $\cG_2 = \{ 2 \}$, and $u_{i1} = u_{i2} = 1$ for all $i \in [m]$.  The idea behind \eqref{nonadaptiveSk} is to select for each group $k$ the features that are likely to be $u_{i\ell}/\cG_{\ell}$ replicated for all other groups $\ell \neq k$; such features are also more likely to be  $u / [n]$ replicated since $\sum_{k \in [K]} u_{ik} = u$.

\subsection{Local PC weights for the \methodnamens} \label{sec:powerful_weights}

Given the form of  a candidate rejection set in \eqref{RrepEasy}, for each group $k$,  it suffices to only construct the local PC weights $\nu_{ik}$ for those $i \in \cS_k$. To ensure finite-sample $\FDR_{\rep}$ control, we require the local PC weight functions $f_{\nubf_k}$'s to create weights that satisfy either \condref{nu}$(a)$ or $(b)$ below:

\begin{conditionalt}[Restrictions on the construction of $\nu_{ik}$]\label{cond:nu}
 For a given group $k\in [K]$ and  $i \in \cS_k$, the local PC weight $\nu_{ik}$ is a function of
 \begin{enumerate}[\indent (a)]
  	\item $\mathbf{P}_{(-\cS_{k}) \cG_{k}}$, $\mathbf{X}$, and $\cS_k$ only. 
	\item $\mathbf{P}_{\cdot (-\cG_{k})}$, $\mathbf{X}$, and $\cS_k$ only. 
\end{enumerate}
\end{conditionalt}
In other words, 
both conditions allow $\nu_{ik}$ to depend on all  the covariates, where \condref{nu}$(a)$  allows further dependence on those base $p$-values in the same group $k$ that are not selected by $\cS_k$, and  \condref{nu}$(b)$  allows further dependence on those base $p$-values not in the same group $k$. Under  either condition,  
our  $\FDR_\rep$ control  in  \secref{main_results}   is established in the present frequentist setting where the truth statuses of the base hypotheses $H_{ij}$'s in \eqref{individual_hypothesis} are fixed.  

To motivate our specific construction of the local PC weights, we will nonetheless borrow a  working model for each study where the statuses of its base hypotheses  can be \emph{random}. Such Bayesian formulations   are also common for multiple testing in the  single-study context; see \citet{Leung2022} and \citet{Heller2020}, as well as the older references therein for instance. Our working model is as follows:  Conditional on ${\bf X}$, the duples  $(p_{1j}, H_{1j}), \dots, (p_{mj}, H_{mj})$ in study $j$ are independent, and the distribution of  each $p_{ij}$ only depends on $x_{ij}$ as
 \begin{equation}\label{pmod}
 p_{ij} | x_{ij} \sim  f_{\zetabf_j,\betabf_j}(p_{ij}|x_{ij} ) \equiv \pi_{\zetabf_j}(x_{ij}) + (1 - \pi_{\zetabf_j}(x_{ij})) f_{\betabf_j}(p_{ij}|x_{ij} )
\end{equation}
with parameter vectors $\boldsymbol{\beta}_j, \zetabf_j \in \mathbb{R}^{1 + d_j}$, where
 \begin{equation*}\label{nullx}
     \pi_{\zetabf_j}(x_{ij}) \equiv \frac{1}{1 + e^{-\zetabf^T_{j} (1,x_{ij}) }}
\end{equation*}
 models  $ \Pr(H_{ij} \text{ is true} \mid x_{ij})$, the conditional probability that $H_{ij}$ is true given $x_{ij}$, and
 \begin{equation*}
    f_{\boldsymbol{\beta}_j}(p_{ij}|x_{ij}) \equiv (1 - k_{\boldsymbol{\beta}_j}(x_{ij}) ) \cdot p^{-k_{\boldsymbol{\beta}_j}(x_{ij})},
\end{equation*}
with $ k_{\boldsymbol{\beta}_j}(x_{ij}) \equiv (1 + e^{-\boldsymbol{\beta}_j^T (1,x_{ij}) })^{-1}$, is a right-skewed beta density that models the conditional density of $p_{ij}$ when $H_{ij}$ is not true. Note that $p_{ij}$ is  modeled to be uniformly distributed on $[0, 1]$ given $x_{ij}$ when $H_{ij}$ is true, and \eqref{pmod} boils down to a marginal mixture model for $p_{ij}$ with only intercept parameters if no covariates are available for study $j$. The model can also be  generalized by  incorporating splines to capture possible nonlinearities in $x_{ij}$.

 Under the working assumption that $(p_{i1},H_{i1} ), \dots, (p_{in}, H_{in} )$ are independent conditional on ${\bf X}$ \footnote{A natural analog to \assumpref{indepstudies} when the base hypotheses can be random.}, our approach to constructing sensible local PC weights revolves around estimating 
\[
\Clfdrc^{\upsilon/\cJ}_i  \equiv \Pr ( \text{$H^{u/\cJ}_i$ is false} \mid \mathbf{P}, \mathbf{X} )
\]
 for a given $i \in [m]$, $\cJ \subseteq [n]$ and $\upsilon \in |\cJ|$, where $H^{v/\cJ}_i$ is the partial conjunction null hypothesis
 \[
 H^{v/\cJ}_i: |\{ j \in \cJ: \mu_{ij} \in \mathcal{A}_i \}| \geq |\cJ| - v  +1.
 \]
 In  \appref{PCClfdrc}, we will derive that $\Clfdrc^{\upsilon/\cJ}_i $ has the generic form
 \begin{equation} \label{form_for_omega_ivcj}
 \Clfdrc^{\upsilon/\cJ}_i =\Clfdrc^{\upsilon/\cJ}_i \big(\zetabf_j, \betabf_j, p_{ij}\big) = \sum_{ \cL \subseteq \cJ : |\cL| \geq \upsilon }  \Bigg( \prod_{j \in \cJ \backslash \cL}  \tfrac{ \pi_{\zetabf_{j}}(x_{ij}) }{ f_{\zetabf_j, \betabf_j}(p_{ij}|x_{ij} )  }  \cdot  \prod_{j \in \cL}  \tfrac{ (1 -\pi_{\zetabf_{j}}(x_{ij}) ) \cdot f_{\betabf_{j}}(p_{ij}|x_{ij}) }{ f_{  \zetabf_j, \betabf_j}(p_{ij}|x_{ij} )  } \Bigg),
 \end{equation}
 which is treated as a function in $\zetabf_j, \betabf_j, p_{ij}$.
 We  now illustrate  how to leverage this expression  to construct  PC weights that satisfies \condref{nu}$(a)$. A similar way to construct such weights  satisfying \condref{nu}$(b)$ is deferred to \appref{localPCweights_b}.


\begin{example}[Forming $  \nu_{ik}$ that satisfies  \condref{nu}$(a)$]\label{ex:suggested_PC_weights}
The general idea is to let
\begin{equation}\label{suggested_feature_weights}
  \nu_{ik} = \frac{{\footnotesize \tilde{\Clfdrc}^{\uiG}_i} |\cS_k| }{ \sum_{\ell \in \cS_k} { \footnotesize \tilde{\Clfdrc}^{u_{\ell k} / \cG_k}_{\ell} } }  ,
\end{equation}
where  ${\footnotesize \tilde{\Clfdrc}^{\uiG}_i}$ is an estimate of  $\Clfdrc^{{}_{\uiG}}_{i}$,
the  probability that $H^{\uiG}_i$ is false conditional on $\mathbf{P}$ and $\mathbf{X}$.  The local PC weights defined by \eqref{suggested_feature_weights} not only satisfy $\sum_{\ell \in \cS_k} \nu_{\ell k} = |\cS_k|$ as required in \eqref{nufun}, but also allows feature $i$ to be more easily rejected if $H^{{}_{\uiG}}_i$ is false. 
 Therefore, to the end of estimating $\Clfdrc^{{}_{\uiG}}_{i}$,  we will plug in estimates $(\tilde{\zetabf}_{j}, \tilde{\betabf}_{j}, \tilde{p}_{ij})$ of $(\zetabf_j, \betabf_j, p_{ij})$ constructed only based on $(\mathbf{P}_{(-\cS_{k}) \cG_{k}}, \mathbf{X}, \cS_k)$   into \eqref{form_for_omega_ivcj} above, by taking $\cJ = \cG_k$ and $v = u_{ik}$.

Let $k(j) \equiv \{k  \in [K] : j \in \cG_{k(j)} \}$ be  the group that study $j$ belongs to, and 
\begin{equation*}
    (\hat{\zetabf}_{j}, \hat{\betabf}_{j}) \equiv \argmax_{(\zetabf_{j}, \betabf_{j})} \sum_{i \notin \cS_{k(j)} } \log\left( f_{\zetabf_j,\betabf_j}(p_{ij}|x_{ij} )    \right)
\end{equation*}
be estimators for $(\zetabf_{j}, \betabf_{j})$, obtained by maximizing the working likelihood of \eqref{pmod} on the subset of data from study $j$ corresponding to the non-selected features in group $k(j)$. Since non-selected features tend to exhibit sparse signals, we find that inflating $\hat{\zetabf}_j$ by a multiplicative factor of $1.5$  generally improves the  fit, so we will let
\[
\tilde{\zetabf}_j = 1.5  \hat{\zetabf}_j \text{ and }\tilde{\betabf}_{j} = \hat{\betabf}_j.
\]
Moreover, by substituting them into the expression of $\mathbb{E} \left[ p_{ij} | x_{ij} \right]$ in $({\zetabf}_{j}, {\betabf}_{j})$ derived in \appref{epx}, an estimate for $p_{ij}$ can  also be taken as
\begin{equation*}
\tilde{p}_{ij} \equiv \frac{2 \cdot (1-k_{\tilde{\betabf}_j}(x_{ij})) + k_{\tilde{\betabf}_j}(x_{ij}) \cdot \pi_{ \tilde{\zetabf}_j}(x_{ij})  }{4 - 2 \cdot k_{\tilde{\betabf}_j}(x_{ij})}.
\end{equation*}
Plugging $(\tilde{\zetabf}_{j}, \tilde{\betabf}_{j}, \tilde{p}_{ij})$ into \eqref{form_for_omega_ivcj} and taking $\cJ = \uiG$ and $v = u_{ik}$, our estimate for $\Clfdrc^{{}_{\uiG}}_{i}$ is given by
\begin{align*}
    \tilde{\Clfdrc}^{\uiG}_{i} \equiv  \sum_{ \cL \subseteq \cG_k : |\cL| \geq u_{ik} } & \Bigg( \prod_{j \in \cG_k \backslash \cL}  \tfrac{ \pi_{ \tilde{\zetabf}_{j}}(x_{ij}) }{ f_{ \tilde{\zetabf}_j, \tilde{\betabf}_j}(\tilde{p}_{ij}|x_{ij} )  }  \cdot  \prod_{j \in \cL}  \tfrac{ (1 -\pi_{ \tilde{\zetabf}_{j}}(x_{ij}) ) \cdot f_{\tilde{\betabf}_{j}}(\tilde{p}_{ij}|x_{ij}) }{ f_{  \tilde{\zetabf}_j, \tilde{\betabf}_j}(\tilde{p}_{ij}|x_{ij} )  } \Bigg). \nonumber
\end{align*}
Obviously, $ \tilde{\Clfdrc}^{\uiG}_{i}$ only depends on $(\mathbf{P}_{(-\cS_{k}) \cG_{k}}, \mathbf{X}, \cS_k)$  because all of $\tilde{\zetabf}_{j}, \tilde{\betabf}_{j}, \tilde{p}_{ij}$ are so. 
\end{example}

Lastly, we re-emphasize that the Bayesian working model \eqref{pmod} is only a device for   constructing sensible local PC weights as in \eqref{suggested_feature_weights}, and is not assumed to hold true. As long as \condref{nu}$(a)$ or $(b)$ is satisfied, our theoretical $\FDR_\rep$ guarantees  in   \secref{main_results} can be established in the pure frequentist sense where the base hypothesis statuses are held fixed.

\subsection{Determining the thresholds}\label{sec:PF}
For each group $k \in [K]$, let    
\begin{equation} \label{group_pi}
    \text{$\pi_{k} \equiv \frac{\sum_{ i \in \cS_{k} \cap   \cH_k} \nu_{ik} }{|\cS_{k}| \vee 1}$} \in [0,1]
\end{equation}
be the \textit{(post-selection) weighted proportion of true local PC nulls}.
Moreover, for a candidate rejection set $\cR(\tvec)$, let  
\begin{align*}
    \FDP_{k}(\tvec) \equiv  \frac{ \sum_{i \in  \cH_k }   \1 \left\{ i \in \cR(\tvec) \right\} }{ |\cR (\tvec)  | \vee 1 } 
\end{align*}
be the \textit{false discovery proportion (FDP)} of local PC nulls in group $k$. 
A natural data-driven estimator for $\FDP_{k}(\tvec)$ can be taken as
\begin{equation}\label{FDPhat_repk}
    \FDPhat_{k}(\tvec) \equiv \frac{|\cS_{k}| \cdot \hat{\pi}_{k} \cdot t_k}{|\cR(\tvec)| \vee 1}, 
\end{equation}
where $t_k$ is the $k$-th coordinate of $\tvec$, and $\hat{\pi}_{k}  \equiv f_{\hat{\pi}_{k}}(\mathbf{P}; \mathbf{X}, \lambda_k ) \in (0,\infty)$ is an estimator for $\pi_k$; possible choices for the function $f_{\hat{\pi}_k}$ are discussed in \secref{main_results}. Heuristically, \eqref{FDPhat_repk} provides a conservative estimate for $\FDP_{k}(\tvec)$ because
\begin{align}
     \FDP_{k}(\tvec) \equiv  \frac{ \sum_{i \in \cH_k}  \1 \left\{ i \in \cR(\tvec) \right\} }{ |\cR(\tvec)  | \vee 1 }  &\leq  \frac{ \sum_{i \in  \cS_k \cap \cH_k }  \1 \left\{ p^{\uiG}_{i} \leq \nu_{ik} \cdot t_{k} \right\} }{ | \cR(\tvec)  | \vee 1}  \nonumber \\
     &\lesssim \frac{  |\cS_k| \cdot {\pi}_{k}  \cdot t_{k} }{ | \cR(\tvec)  | \vee 1}  \approx  \frac{  |\cS_k| \cdot \hat{\pi}_{k}  \cdot t_{k} }{ | \cR(\tvec)  | \vee 1} \equiv \FDPhat_{k}(\tvec). \nonumber 
\end{align}
Above, the approximate inequality ``$\lesssim$'' follows from the fact that, under conditionally valid local PC $p$-values (\condref{valid_monotone_local_PC_pv}$(a)$), the number of rejected true local PC nulls is approximately $\sum_{ i \in \cS_{k} \cap   \cH_k}( \nu_{ik} t_k) \approx |\cS_k| \cdot \pi_k \cdot t_k$  from \eqref{group_pi}. As mentioned in our overview, the \methodname aims to control $\FDR_{\rep}$ under level $q$, by making a data-driven choice of $\tvec$ for $\cR (\tvec)$ in \eqref{RrepEasy} that controls each $\FDR_{k}$  under level $w_k \cdot q$. As such, it  computes a data-driven threshold vector $\hattvec$ that satisfies 
\begin{equation}\label{desired_inequalities}
    \FDPhat_{k}(\hattvec) \leq w_k \cdot q \text{ for all }k \in [K],
\end{equation}
because one can then expect that
\begin{equation*}
   \FDR_{k}(\cR(\hattvec)) \equiv \mathbb{E}\left[ \FDP_{k}(\hattvec)  \right] \lesssim \mathbb{E}\left[ \FDPhat_{k}(\hattvec)  \right] \leq w_k \cdot q \text{ for each }  k . 
\end{equation*}
Rearranging \eqref{desired_inequalities} into
\begin{equation*}
   \hat{t}_k \leq \frac{ |\cR(\hattvec)| \vee 1 }{|\cS_k| \cdot \hat{\pi}_{k}} \cdot w_k \cdot q  
\end{equation*}
highlights how, as a result of filtering, the multiplicity correction associated with $\hatt_k$ scales with $|\cS_k|$, instead of with the possibly much larger $m$. Below, we define the set of all feasible values for  $\hattvec$ that satisfy all the constraints in \eqref{desired_inequalities}: 

\begin{definition}[Set of all feasible values for $\hattvec$]\label{def:set_of_feasible_thresholds}
Given the false discovery proportion estimator in \eqref{FDPhat_repk} for each group $k \in [K]$, define
\begin{equation*} \label{set_of_feasible_thresholds}
    \cT \equiv \Big\{  \tvec \in [0,\infty)^K: \FDPhat_{k} (\tvec) \leq w_k \cdot q \text{ for all } k \in [K] \Big\}
\end{equation*}
to be the set of possible values for $\hattvec$ for which the constraints in \eqref{desired_inequalities} are all satisfied.
\end{definition}

Now, among all the vectors in $\cT$, which one should $\hattvec$ be? Since the number of rejections $|\cR (\tvec)|$ is a non-decreasing function in $\tvec$, the intuition is to choose the coordinates of $\hattvec$ to be as large as possible. Turns out, an unequivocally best choice of $\hattvec$ exists:

\begin{proposition}[Existence of an optimal vector $\hattvec$ in $\cT$]\label{prop:optimal_thresholds}
   Let $\cT$ be   as in \defref{set_of_feasible_thresholds}, and also define
    \begin{equation}\label{optimal_tau_repk} 
       \hat{t}_{k} \equiv \max\left\{ t_k \in [0,\infty) : \exists (\tvec_{[1:(k-1)]}, t_k,\tvec_{ [(k+1):K]}) \in \cT \right\}
    \end{equation}
    for each $k\in [K]$. 
     Then $\hattvec \equiv (\hat{t}_1,\dots,\hat{t}_K) \in \cT$.
In particular, for any ${\bf t}= (t_1, \dots, t_K)  \in \cT$, we have that $t_k \leq \hatt_k$ for all $k \in [K]$; 
in other words, $\hattvec$ is the largest element-wise vector in $\cT$.
\end{proposition}

To find this optimal $\hattvec$ efficiently, we recommend using our iterative procedure in  \algref{tau} below. Its correctness and termination are ensured by the following proposition.

\begin{proposition}[Correctness and Termination  of \algref{tau}]\label{prop:tau_correctness}
\algref{tau} will
\begin{inparaenum}[(i)]
\item
output the optimal vector of thresholds $\hat{\tvec}$   in \propref{optimal_thresholds}, and
\item 
 terminate in finitely many steps.
  \end{inparaenum}
\end{proposition}

Both the proofs of \propsref{optimal_thresholds} and  \propssref{tau_correctness} involve interesting limiting arguments,   and are provided in \appsref{optimal_thresholds} and \appssref{tau_correctness} respectively. Once $\hattvec$ is determined, the \methodname outputs $\widehat{\cR} = \cR(\hattvec)$ as its final rejection set.  We    provide a full recap of the entire procedure  in \algref{PF}. In \appref{BogHellerEquiv}, we demonstrate that the ParFilter, with particular tuning settings, can exactly recover the $\FDR_\rep$ controlling procedures of \citet{Bogolomov2018} when  $n = u = 2$, as noted in \secref{intro}.

\ \


\begin{algorithm}[H]
\caption{Computing the optimal vector of thresholds $\hattvec$}\label{alg:tau}

    Set $s = 1$ and let $\hattvec^{(0)} = (\hatt^{(0)}_1, \dots, \hatt^{(0)}_K) = (\infty, \dots, \infty)$. 

 Let $\hattvec^{(s)} \equiv (\hatt^{(s)}_1, \dots, \hatt^{(s)}_K)$ and for each $k \in [K]$, define its component recursively as 
    \begin{align*}
       \hatt^{(s)}_{k} \leftarrow \max \Big\{ t_k \in [0, \hatt_k^{(s-1)} ] : \FDPhat_{k} \left( (\hattvec^{(s)}_{[1:(k-1)]},t_{k},\hattvec^{(s-1)}_{[(k+1):K]}) \right) \leq w_k \cdot q \Big\}.
    \end{align*}
    If $\hattvec^{(s)} = \hattvec^{(s-1)}$, where the equality is understood to hold element-wise, move to the output step. Otherwise, let $s \leftarrow s + 1$ and go back to Step 2.

    \KwOut{$\hattvec = \hattvec^{(s)}.$}
\end{algorithm}

\vspace{0.6cm}

\begin{algorithm}[H] 
\caption{\methodname at $\FDR_{\rep}$ target level $q$ }
\label{alg:PF}
\KwData{ $\mathbf{P}$, $m \times n$ matrix of base $p$-values;  $\mathbf{X}$, $m \times n$ matrix of covariates.}

\KwIn{
 $\cG_k, w_k, \mathbf{u}_i, \lambda_k, f_{ik}, f_{\cS_k}, f_{\nubf_{k}}, f_{\hat{\pi}_k}$ for $k \in [K]$ and $i \in [m]$; \\ 
\qquad \qquad $\FDR_{\rep}$ target $q \in (0,1)$.} 

Compute the local PC $p$-value $p^{\uiG}_i$ for each $(i,k) \in [m]\times [K]$ according to \eqref{generic_local_pv}.

Compute the set of selected features $\cS_k$ for each $k \in [K]$ according to \eqref{selected_set}.


Compute the local PC weights $\nubf_k$ for each $k \in [K]$ according to \eqref{nufun}.

Compute the weighted null proportion estimate $\hat{\pi}_{k} \equiv f_{\hat{\pi}_{k}}(\mathbf{P}; \mathbf{X}, \lambda_k )$ for each  $k \in [K]$. Options for $f_{\hat{\pi}_{k}}$ are provided in \secref{main_results}.

Compute  the optimal threshold vector $\hattvec \in \cT$ implied by  \propref{optimal_thresholds} with \algref{tau}.

\KwOut{Take the rejection set as $\widehat{\cR} = \cR(\hattvec)$ based on   \eqref{RrepEasy}.} 
\end{algorithm}

\subsection{Main theoretical results}\label{sec:main_results}
The theorem below establishes the finite-sample $\FDR_{\rep}$ control guarantees of the \methodname when the base $p$-values are conditionally independent across the $m$ features given $\mathbf{X}$. 

\begin{theorem}[$\FDR_{\rep}$ control under independence]\label{thm:error_control}
    Suppose \assumpref{indepstudies} is true, and the components of  \methodname  (\algref{PF}) meet \condsref{valid_monotone_local_PC_pv}, \ref{cond:indep}, \ref{cond:stable}, and \ref{cond:nu}$(a)$, i.e., 
       \begin{itemize}
    \item Each local PC $p$-value is conditionally valid;
    \item For each $k$, $\cS_k$ is a function only of  $\mathbf{P}_{\cdot (-\cG_k)}$ and $\mathbf{X}$;
    \item Each selection rule $f_{\cS_k}$ is stable;
    \item Each local PC weight $\nu_{ik}$ is a function of $\mathbf{P}_{(-\cS_{k}) \cG_{k}}$, $\mathbf{X}$, and $\cS_k$ only.
    \end{itemize}
    If the base $p$-values $p_{1j},p_{2j},\dots,p_{mj}$ are independent conditional on $\mathbf{X}$ for each study $j \in [n]$, then the \methodname at level $q \in (0,1)$  has the  false discovery rate controlling properties
    \begin{equation}\label{FDR_control}
        \FDR_{k} (\cR(\hattvec)) \leq  w_k \cdot q  \;\ \text{ for each } k \in [K] \;\ \text{ and } \;\  \FDR_{\rep}(\cR(\hattvec)) \leq q,
    \end{equation}
provided that  the tuning parameters $\{\lambda_k\}_{k \in [K]}$ and weighted null proportion estimators $\{\hat{\pi}_{k}\}_{k \in [K]}$ are taken as one of the following: 
     \begin{enumerate}[(i)]
        \item  For each group $k \in [K]$, $\lambda_k = 1$ and $\hat{\pi}_k = 1$.\\
        \item For each group  $k \in [K]$ and some fixed choice of $\lambda_k \in (0,1)$,       
         \begin{equation}\label{pi_hat_adaptive}
            \hat{\pi}_k = \frac{ \max_{i \in \cS_k} \nu_{ik} + \sum\limits_{i \in \cS_k} \nu_{ik} \cdot \1 \left\{ p^{\uiG}_{i} > \lambda_{k} \right\} }{(1-\lambda_{k}) \cdot |\cS_k|}  \in (0, \infty).
        \end{equation} 

     \end{enumerate}
\end{theorem}

The control of $\FDR_{\rep}$ at level $q$ in \thmref{error_control} follows from the control of $\FDR_{k}$ at level $w_k \cdot q$ for each $k \in [K]$ (\lemref{FDR_rep_control_from_partition}). Hence, its proof in \appref{main_theorem_pf} focuses on showing the latter result by extending the FDR control techniques developed by \citet{Katsevich2023}.  The estimator \eqref{pi_hat_adaptive} has its root in \citet[equation (2.2)]{pfilter2019}, where the choice of tuning parameter $\lambda_k$ entails a bias-variance trade-off in estimating $\pi_k \in [0,1]$; setting $(\lambda_1,\cdots,\lambda_K) = (0.5, \cdots, 0.5)$ usually strikes a good balance in our experience. We generally recommend using this  data-driven estimator for $\pi_k$
rather than setting $\hat{\pi}_k = 1$ as in \thmref{error_control}$(i)$, since the former is usually less conservative.

We now turn to the finite-sample $\FDR_{\rep}$ control guarantees of  the \methodname when the base $p$-values can be arbitrarily dependent across the $m$ features.

\begin{theorem}[$\FDR_{\rep}$ control under arbitrary dependence]\label{thm:error_control_dep}
    Suppose \assumpref{indepstudies} is true, and  the components of  \methodname   (\algref{PF}) meet \condsref{valid_monotone_local_PC_pv}, \ref{cond:indep}, \ref{cond:stable}, and \ref{cond:nu}$(b)$, i.e.,
    
\begin{itemize}
    \item Each local PC $p$-value is conditionally valid;
        \item For each $k$, $\cS_k$ is a function only of  $\mathbf{P}_{\cdot (-\cG_k)}$ and $\mathbf{X}$;
            \item Each selection rule $f_{\cS_k}$ is stable;
            \item Each local PC weight $\nu_{ik}$ is a function of $\mathbf{P}_{\cdot (-\cG_k)}$, $\mathbf{X}$, and $\cS_k$ only.
\end{itemize}
       Provided that for each $k \in [K]$, the tuning parameter is taken as $\lambda_k = 1$ and  the weighted null proportion estimator is taken as
        \begin{equation} \label{misnomer}
            \hat{\pi}_{k} =  \sum_{i \in [|\cS_{k}| \vee 1]}\frac{1}{i},
        \end{equation}
 the \methodname at level $q \in (0,1)$  has the false discovery rate controlling properties stated in \eqref{FDR_control}.
\end{theorem}

The proof of \thmref{error_control_dep} given in  \appref{main_theorem_dep_pf} relies on showing the $\FDR_k$ control using results similar to the ``dependency control conditions'' originating from the important work of \citet{Blanchard2008}. \thmref{error_control_dep} mainly differs from \thmref{error_control} by requiring the local PC weights to satisfy a different set of constraints -- \condref{nu}$(b)$ instead of $(a)$ -- and by adopting a more conservative null proportion estimator -- \eqref{misnomer} instead of \eqref{pi_hat_adaptive} or $\hat{\pi}_k = 1$. We note that calling the form in \eqref{misnomer} a weighted ``null proportion'' estimator is somewhat of a misnomer, as a harmonic number cannot be less than 1. Rather, \eqref{misnomer} serves to suitably inflate the FDP estimator in \eqref{FDPhat_repk} to account for the arbitrary dependence among the base $p$-values.


\section{Simulation Studies}\label{sec:sim}
For $n \in \{2, 3, 4, 5\}$ studies, 
we conduct simulations to assess the performance of the \methodname  in a maximum replicability analysis with $u=n$, where $H^{n/[n]}_1,\dots,H^{n/[n]}_m$ are tested simultaneously for $m = 5000$ features. The data are generated  with the steps below:

\begin{enumerate}[(i)]
\item We first let $x_i  \overset{\tiny \text{iid}}{\sim} \text{N}(0,1)$ for $i = 1, \dots, 5000$, and set 
 \begin{equation}\label{same_covar_in_sim}
    x_{i1} = x_{i2} = \cdots = x_{in} = x_i.
\end{equation}
In other words, for each feature $i$,  the covariates $x_{i1}, \dots, x_{in}$ are all equal and generated from the same standard normal distribution.

\item Next, given all the  covariates, the statuses of the base null hypotheses are decided as follows: For each $(i,j) \in [5000] \times [n]$, its corresponding base hypothesis only depend on $x_{ij}$ and is generated as 
\begin{equation*}
    \1 \{ H_{ij} \text{ is true} \} | x_{ij} \sim  \text{Bernoulli} \left( \frac{\exp(\gamma_0(n) + \gamma_1 \cdot x_{ij})}{1 + \exp(\gamma_0(n) + \gamma_1 \cdot x_{ij})}  \right),
\end{equation*}
with the parameters $\gamma_0(n) = \log( (0.01/n)^{-1/n} - 1)$ and $\gamma_1 \in \{ 0, 1, 1.5 \}$. 

To shed light on our choices for the simulation parameters, note that when $x_{ij} = 0$ for all $(i,j) \in [5000] \times [n]$, the value of $\gamma_0(n)$ ensures the proportion of false PC nulls  is approximately $0.5 \%$, $0.3 \%$, $0.25 \%$ and $0.2 \%$ for $n = 2,3,4$, and $5$  respectively, as the probability of any $H_i^{n/[n]}$ being false can be computed as  $1/(1+\exp(\gamma_0(n)))^n$.  When $\gamma_1 = 0$, $x_{ij}$ has no influence on $H_{ij}$, but as $\gamma_1$ increases above 0, $x_{ij}$ becomes more correlated with $H_{ij}$.
\item Lastly, given all the covariates and hypothesis statuses, the base $p$-values are independently generated as
\begin{equation*}
    p_{ij} \sim
    \begin{cases}
        \text{Unif}(0,1) &\text{if $H_{ij}$ is true} \\
        \text{Beta}(1 - \xi,7) &\text{if $H_{ij}$ is false}
    \end{cases},
\end{equation*}
where $\xi$ is a signal size parameter ranging in $\{ 0.74, 0.76, 0.78, 0.80, 0.82 \}$. As $\xi$ increases from $0.74$ to $0.82$, $\text{Beta}(1 - \xi,7)$ becomes increasingly right-skewed and stochastically smaller than a uniform distribution.
\end{enumerate}

Under this data generating mechanism, the base $p$-values across features are independent given ${\bf X}$ and \thmref{error_control} applies. 
For each feature $i \in [5000]$, we also let 
\begin{equation}\label{maxPCp}
p^{n/[n]}_i = \max(\mathbf{P}_{i \cdot}),
\end{equation}
be the PC $p$-value for testing $H_i^{n/[n]}$, 
formed using the GBHPC method in \defref{gbhpc}, where $f_{\cJ}$ is taken as either the Bonferroni, Stouffer, Fisher, or Simes combining function, as they all lead to \eqref{maxPCp} when $u=n$.


\subsection{Compared methods}\label{sec:sim_methods}
We compare the following methods designed for multi-study replicability analysis, including two versions of our ParFilter methodology:

\begin{enumerate} [(a)]
 \item \textbf{\methodnamens}: This is \algref{PF}, using 
 \begin{itemize}
 \item the testing configuration suggested in \exref{suggested_testing_config_ueqn},
 
 \item the local PC $p$-values in \eqref{local_GBHPC}  with all the $f_{i, \cJ}$'s taken as Stouffer's combining function,\footnote{Under the testing configuration suggested in \exref{suggested_testing_config_ueqn}, these local PC $p$-values are simply the base $p$-values.}
 \item the selection rules  in \exref{suggested_select_rules},
 \item  the local PC weights suggested in \exref{suggested_PC_weights},
 \item and  the weighted null proportion estimators for each $k \in [K]$ as  in \eqref{pi_hat_adaptive}, with tuning parameters $(\lambda_1,\dots,\lambda_K) = (0.5,\cdots,0.5)$. 
 
 \end{itemize}
            \item \textbf{No-Covar-\methodnamens}: An implementation of \algref{PF} almost exactly the same as item $(a)$ above, except that the local PC weights are fixed at $\nu_{ik} = 1$ for all $(i, k)$. Consequently, it does not adapt to the available covariate information.

\item      \textbf{AdaFilter-BH}: The procedure presented in \citet[Definition 3.2]{adafilter}.

  \item \textbf{Inflated-AdaFilter-BH}: The procedure presented in   \citet[Definition 3.2]{adafilter} but with the $\FDR_{\rep}$ target multiplied by $(\sum_{i \in [m]} \frac{1}{i})^{-1} < 1$.    
      \item \textbf{CoFilter-BH}: The CoFilter procedure described in  \citet[Algorithm 2.1]{Dickhaus2024} with the BH procedure implemented in step $(iii)$ of the algorithm.
\end{enumerate}

We note that (c)-(e) make use of  the PC $p$-values defined in \eqref{maxPCp} for  $H_i^{n/[n]}$ in their steps. Next we describe four other methods originally intended for single-study FDR analysis that are applicable in the context of this simulation study. These methods all make use of the PC $p$-values in \eqref{maxPCp} and  selected because of their state-of-the-art performance and/or widespread use: \begin{enumerate} [(a)]\setcounter{enumi}{5}
    \item  \textbf{BH}: The BH procedure by \citet{Benjamini1995}.
    \item \textbf{AdaPT}: The adaptive $p$-value thresholding method by \citet{Lei2018}, with default settings based on the R package \texttt{adaptMT}.

    \item \textbf{CAMT}: The covariate-adaptive multiple testing method by \citet{Zhang2022} with default settings based on their R package \texttt{CAMT}.

    \item \textbf{IHW}: The independent hypothesis weighting framework by \citet{Ignatiadis2021} with default settings based on the R package \texttt{IHW}. Note that this method supports only univariate covariates.
   
       \end{enumerate}
       
Except for the BH method, all other methods in (g)-(h) are covariate-adaptive and make use of the covariates $\{x_i\}_{i \in [5000]}$ from \eqref{same_covar_in_sim}. Note that this  setting with the same covariate for each feature across studies makes $x_i$ a rather easy covariate choice for these methods. In other settings with different covariates across studies, these methods may become less applicable, particularly for IHW which only supports univariate covariates. In \tabref{method_summary}, we summarize the theoretical $\FDR_{\rep}$ guarantees of all  methods under the data-generating mechanism of this simulation study (i.e. independent base $p$-values across features) and their covariate-adaptiveness. 

\begin{table}[h]
\small
\begin{tabular}{|l|c|l|l|}
\hline
Method                                               & Covariate-adaptive? & $\FDR_{\rep}$ control \\ \hline
{\methodname}                & Yes              & Finite  (\thmref{error_control}$(ii)$)       \\
{No-Covar-\methodname} & No & Finite (\thmref{error_control}$(ii)$)  \\
{AdaFilter-BH}               & No               & Asymptotic (Theorem 4.4, \citet{adafilter})  \\
{Inflated-AdaFilter-BH}      & No               & Finite  (Theorem 4.3, \citet{adafilter})     \\
{CoFilter-BH}                       & No              & Finite (Proposition 4.1, \citet{Dickhaus2024}) \\
{BH}                & No               & Finite (Theorem 1, \citet{Benjamini1995})          \\
{AdaPT}                      & Yes              & Finite (Theorem 1, \citet{Lei2018})         \\
{CAMT}                       & Yes              & Asymptotic (Theorem 3.8, \citet{Zhang2022}) \\
{IHW}                        & Yes              & Asymptotic (Proposition 1 (b), \citet{Ignatiadis2021})  \\

 \hline
\end{tabular}
\caption{Summary of  methods compared in the simulation study of \secref{sim}. Specifically, whether the method is covariate-adaptive and the type of $\FDR_{\rep}$ control it possesses, i.e. finite-sample or asymptotic as $m \longrightarrow \infty$.}\label{tab:method_summary}
\end{table}

\subsection{Simulaton results}
Each method in \secref{sim_methods} is applied at the nominal $\FDR_{\rep}$ level $q = 0.05$ and their power is measured by the \textit{true positive rate}, 
\begin{equation*}
    \text{TPR}_{\rep}  \equiv  \text{TPR}_{\rep}(\widehat{\cR}) = \mathbb{E} \left[ \frac{\sum_{i \in [m] \backslash \cH^{u/[n]} }  \1 \left\{ i \in \widehat{\cR} \right\} }{ |[m] \backslash \cH^{u/[n]}| \vee 1} \right],
\end{equation*}
where $\widehat{\cR}$ is the method's rejection set and $[m] \backslash \cH^{u/[n]}$ is the set of false PC nulls. The simulation results are presented in \figref{Indepueqn} where the  empirical $\FDR_{\rep}$ and $\text{TPR}_{\rep}$ levels are computed based on $500$ repetitions. The following observations can be made:
\begin{itemize}
    \item \textbf{Observed $\FDR_{\rep}$ levels:} $\FDR_{\rep}$ control at $q = 0.05$ is generally observed for most methods. A notable exception is AdaFilter-BH (\AdaFilterBHCol line) under the setting $(n, \gamma_1) = (5, 0)$ (i.e., $5/[5]$ replicability with uninformative covariates) as its $\FDR_{\rep}$ levels exceeds $q = 0.05$ across most values of $\xi$. As noted in \tabref{method_summary}, AdaFilter-BH provides only asymptotic $\FDR_{\rep}$ guarantees.
    
    \item \textbf{Power when covariates are uninformative ($\gamma_1 = 0$):} When the covariates are uninformative, AdaFilter-BH outpowers the other methods across all settings of $n$ and $\xi$. 
    The  ParFilter and No-Covar-ParFilter methods (\ParFilterCol and \NoCovarParFilterCol lines), which visually overlap each other, are the most powerful methods after AdaFilter-BH. The other covariate-adaptive methods -- AdaPT (\AdaPTCol line), CAMT (\CAMTCol line), and IHW (\IHWCol line) -- perform worse than the covariate-free methods CoFilter-BH (\CoFilterBHCol line) and Inflated-AdaFilter-BH (\InflatedAdaFilterBHCol line).
   
    \item \textbf{Power when covariates are informative ($\gamma_1 = 1.0$ or $1.5$):}  The \methodname generally achieves the highest power when the covariates are informative. Its power advantage over competing methods is most pronounced when $n$ is small or when the signal size $\xi$ is large. For instance, in the setting of $2/[2]$ replicability with highly informative covariates ($\gamma_1 = 1.5$), the \methodname achieves over 0.10 higher $\text{TPR}_{\rep}$ than the next best method, AdaFilter-BH, for $\xi \in \{0.76, 0.78, 0.80, 0.82 \}$. 
    CAMT is consistently more powerful than the covariate-adaptive methods AdaPT and IHW but still, in general, fails to be more competitive than AdaFilter-BH, a covariate-free method dedicated for replicability analyses, even when the covariates are highly informative.
    
\end{itemize}

For maximal replicability analyses, our simulations highlight the competitiveness of the ParFilter, particularly when covariates are moderately to highly informative and the signal size is not too low. In simulation scenarios with uninformative covariates, AdaFilter-BH is the most powerful method but may violate $\FDR_{\rep}$ control. Therefore, the \methodname may be preferable if ensuring $\FDR_{\rep}$ control is a high priority. In \appref{additional_sim}, we explore further simulation setups, including the case of $u < n$ and the case of autoregressively correlated base $p$-values, as well as more methods, including alternative implementations of the \methodname with different algorithmic settings.
\begin{figure}[H]
    \centering
    \includegraphics[width=0.9\textwidth]{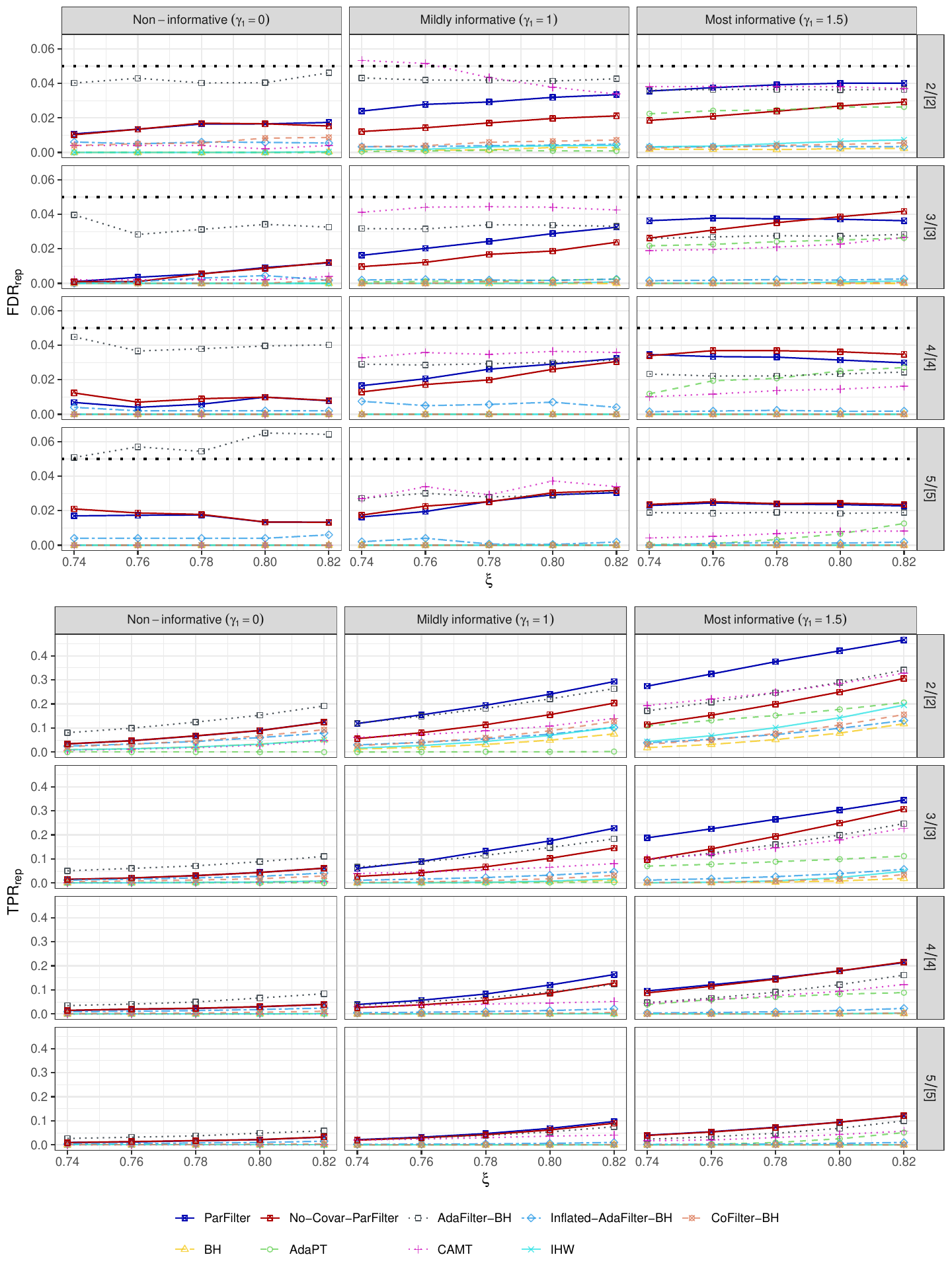}
    \caption{Empirical $\text{FDR}_{\rep}$ and $\text{TPR}_{\rep}$ of the compared methods for the simulations in \secref{sim} based on $500$ simulated datasets, for testing $2/[2]$, $3/[3]$, $4/[4]$, $5/[5]$ replicability under  independence of base $p$-values across features. Logistic parameters $\gamma_1 = 0, 1, 1.5$, denoted as non-informative, mildly informative, and most informative respectively. Effect size parameters $\xi = 0.74, 0.76, 0.78, 0.80, 0.82$.}
    \label{fig:Indepueqn}
\end{figure}

\newpage

\section{Case Study}\label{sec:realdata}
The immune system mounts an immune response to defend against foreign substances in the body (e.g., viruses and bacteria). An equally important task, called self-tolerance, is to recognize substances made by the body itself to prevent autoimmunity, where the immune system mistakenly attacks healthy cells and tissues. In medullary thymic epithelial (mTEC) cells expressing high Major Histocompatibility Complex class (MHC) II levels, commonly known as mTEChi cells, a protein called the autoimmune regulator (AIRE) drives the expression of genes to synthesize proteins critical for self-tolerance \citep{Derbinski2005}. Impaired AIRE transcription disrupts self-tolerance and leads to autoimmune diseases \citep{Anderson2002}.

In this case study, we conduct a maximum replicability analysis to identify genes implicated in autoimmune diseases arising from defective AIRE transcription. For our analysis, we found $n=3$ RNA-Seq datasets from the GEO that compared control and Aire\footnote{The Aire gene in mice is equivalent to the AIRE gene in humans. It is the convention for mouse gene symbols to have the first letter capitalized, followed by lowercase letters.}-inactivated mice mTEChi cells: GSE221114 \citep{GSE221114}, GSE222285 \citep{GSE222285}, and GSE224247 \citep{GSE224247}. We indexed these datasets as study $1$, $2$, and $3$ respectively. After removing genes with overly low read counts (fewer than 10), we were left with $m = 6,587$ genes for our analysis, each of which was arbitrarily assigned a unique index $i \in [6587]$.

For each $(i,j) \in [6587] \times [3]$, we set the base null hypothesis to be $H_{ij}: \mu_{ij} = 0$ against the alternative $\mu_{ij} \neq 0$, where $\mu_{ij} \in \bR$ denotes the difference in mean expression levels between control and Aire-inactivated mice. Using the \texttt{limma} package \citep{Ritchie2015} in R, we applied the \texttt{eBayes} function to the RNA-Seq data to compute \textit{moderated} $t$-statistics \citep{Smyth2004}, which were then transformed into the base $p$-values $\mathbf{P} = (p_{ij})_{i \in [6587], j \in [3]}$ for this replicability analysis.

We now turn to the task of constructing a covariate $x_{ij}$ for each base hypothesis $H_{ij}$. Consider the following side information on self-tolerance and cells surrounding mTEChi:
\begin{itemize}
    \item The expression of the gene KAT7 supports AIRE in driving the expression of nearby genes to produce proteins essential for self-tolerance \citep{KAT72022}.
    \item Cortical thymic epithelial cells (cTECs), located in the thymus cortex, play a role in helping the body recognize its own proteins without attacking them \citep{Klein2014}.
    \item mTEC cells expressing low MHC II levels, i.e., mTEClo cells, develop into mTEChi cells over time as they mature within the thymus \citep{Akiyama2013}.
\end{itemize}
Given the side information above, consider the following null hypotheses:
\begin{equation}\label{aux_hyp}
    H^{c}_i : \mu^{c}_i = 0 \text{ for each $c \in \{ \text{cTEC}, \text{mTEChi}, \text{mTEClo} \}$ and $i \in [6587]$},
\end{equation}
where $\mu^{c}_i \in \mathbb{R}$ is the mean difference in expression levels between control and Kat7-inactivated mice for gene $i$ in cell $c$. From the GEO, we extracted RNA-Seq data relevant for testing \eqref{aux_hyp} from dataset GSE188870 \citep{KAT72022}. Then, for each $i \in [6587]$, we used \texttt{limma} again to form a vector $\mathbf{y}_i =(y^{\text{cTEC}}_{i},y^{\text{mTEChi}}_{i},y^{\text{mTEClo}}_{i}) \in \bR^3$ where $y^{c}_i \in \mathbb{R}$ is the moderated $t$-statistic for testing $H^{c}_i$. We transform $\mathbf{y}_i$ using a natural cubic spline basis expansion with four interior knots (using the \texttt{ns} function in the R package \texttt{spline} with the argument \texttt{df = 6}) to create a vector $\dot{\mathbf{y}}_i \in \bR^{18}$. The covariates $\mathbf{X} = (x_{ij})_{i \in [6587], j \in [3]}$ for this replicability analysis are then set by letting $x_{ij} = \dot{\mathbf{y}}_i \text{ for each $(i,j) \in [6587] \times [3]$}$. 
We apply the same $\FDR$ methods examined in \secref{sim} to this case study, at the same nominal $\FDR_{\rep}$ target of $q = 0.05$. Since the IHW procedure only supports univariate covariates, we use $=(y^{\text{cTEC}}_{i} + y^{\text{mTEChi}}_{i} + y^{\text{mTEClo}}_{i})/3 \in \bR$ as the covariate  for this method, whereas the other two covariate-adaptive methods for single-study FDR analysis (AdaPT and CAMT) use $\dot{\mathbf{y}}_i$.

\subsection{Replicability analysis results}
The number of genes identified as being $3/[3]$ replicated by each of the compared methods is presented in \figref{AIREanalysis}. There, we see that \methodname achieves a larger number of rejections than the remaining methods across. In particular, the \methodname achieves 11 more rejections than No-Covar-\methodname and 34 more than AdaFilter-BH -- the only other two procedures with more than 500 rejections. Together, the \methodname and  No-Covar-\methodname identified 32 genes that are not detected by any other method. We ranked these 32 genes by their PC $p$-values -- computed using \eqref{stouffer_GBHPC} -- from smallest to largest, and present the top fifteen in \tabref{GeneRanks}. The first three genes in this table are all related to autoimmunity according to the literature:
\begin{itemize}
	\item F11r: \citet{Bonilha2021} found links between F11r and the regulation of immune responses in pathological conditions such as autoimmune diseases and cancers;
	\item Mknk2: \citet{Karampelias2022} identified Mknk2 as a therapeutic target for Type 1 diabetes, a common autoimmune disease that affects the production of insulin;
	\item Mreg: \citet{DamekPoprawa2009} suggest that Mreg is required for lysosome maturation in cells. Lysosomal dysfunction has been implicated in autoimmune diseases \citep{Gros2023}.
\end{itemize}  
In \appref{real_data_extra}, we highlight seven additional genes from \tabref{GeneRanks} with known associations to autoimmune diseases.


\begin{figure}
  \begin{minipage}[b]{.50\linewidth}
    \centering
    \includegraphics[width=0.9\linewidth]{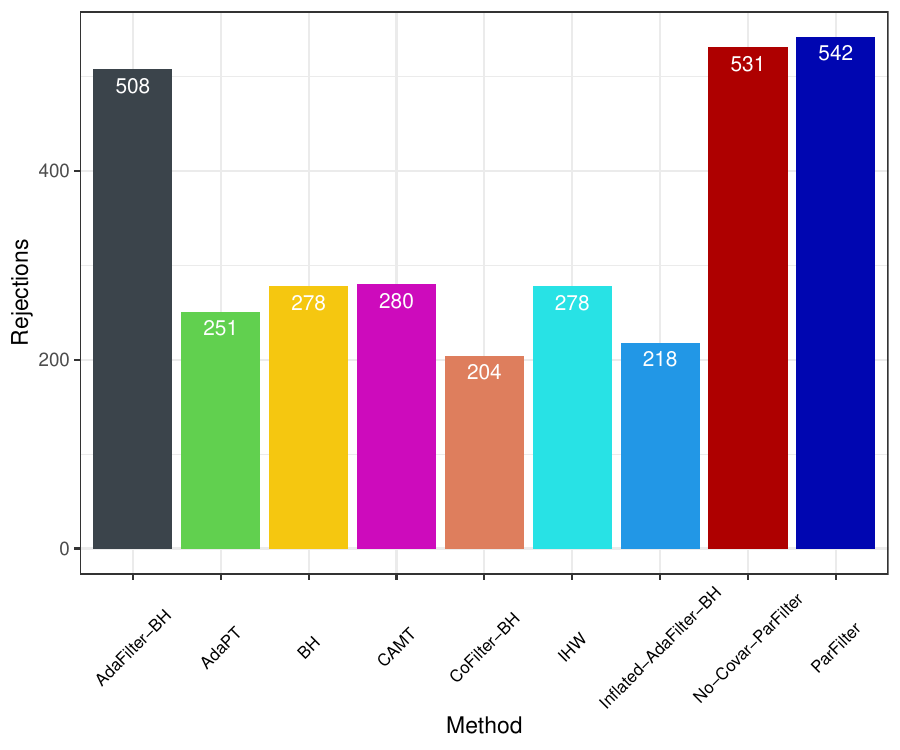}
    \captionsetup{width=0.9\linewidth}
    \captionof{figure}{Number of rejections across the compared methods in a $3/[3]$ replicability analysis.}\label{fig:AIREanalysis}
  \end{minipage}\hfill
  \begin{minipage}[b]{.50\linewidth}\footnotesize
    \centering
    \begin{tabular}{ *{2}{c} }
      Gene & PC $p$-value  \\
      \hline
	F11r    & 0.01259083 \\
	Mknk2   &  0.01260681  \\
	Mreg      & 0.01266433  \\
	Plac8     & 0.01267082  \\
	 Srsf3      & 0.01275518  \\
	 Ecscr     & 0.01278160  \\
	 Jarid2       & 0.01286667   \\
	 Ktn1        & 0.01309430  \\
	 Ncl      & 0.01313040  \\
	 Nhsl1  & 0.01320058 \\
	  Bcl2l2   & 0.01328083  \\
	  Sst   & 0.01329332  \\
	  Fxyd3  & 0.01336191  \\
	  Rell1    & 0.01344367  \\
	  Pnn   & 0.01347595  \\
    \end{tabular}
    \captionsetup{width=0.9\linewidth}
    \captionof{table}{Top 15 genes identified as $3/[3]$ replicated by the \methodname and No-Covar-ParFilter but not by other methods at $q = 0.05$.}\label{tab:GeneRanks}
  \end{minipage}
\end{figure}

\section{Discussion}\label{sec:conclusion}

It may be of interest to perform   \emph{directional} inference  in some applications involving two-sided hypotheses where $\cA_i  = \{0\}$ for all $i \in [m]$ in \eqref{individual_hypothesis}.  Formally, one can  define two  base null hypotheses
\begin{equation*}
    \text{$H^{-}_{ij}: \mu_{ij} \geq 0 \quad$ and $\quad H^{+}_{ij} : \mu_{ij} \leq 0 \quad$  for each $(i,j) \in [m] \times [n]$},
\end{equation*}
and subsequently the directional PC null hypotheses  
\begin{align*}
    H^{u/[n], -}_{i} : & | \{ j \in [n] : H^{-}_{ij} \text{ is false } \}| \geq n - u + 1 \quad \text{and} \\
    &\quad \quad 
     H^{u/[n], +}_{i} : | \{ j \in [n] : H^{+}_{ij} \text{ is false }  \}| \geq n - u + 1 \quad \text{ for each $i \in [m]$}.
\end{align*}
Rejecting $H^{u/[n], -}_{i}$ (resp. $H^{u/[n], +}_{i}$) declares there are at least $u$ negative (resp. positive) signals for feature $i$. We can apply the ParFilter  to test $\{H^{u/[n], -}_{i}\}_{i \in [m]}$ and $\{H^{u/[n], +}_{i}\}_{i \in [m]}$ under the $\FDR_\rep$  targets $q^-$ and $q^{+}$ separately, and the directional false discovery rate for the overall replicability analysis can be controlled under level $q^- + q^+$; we refer to \citet[Sec. 4.4.3]{adafilter} and \citet[Remark 5.1]{BogomolovHeller2023}  for more exposition on this approach.

In summary, the  ParFilter is a highly customizable framework (with flexible choices of testing configuration, selection rules, local replicability levels, etc.) to test for replicated signals in  large-scale settings with strong false discovery rate control in finite samples.  Its adaptiveness to auxiliary covariate information has opened up new avenues to tackle  partial conjunction testing, which is by nature very power-deficient for high replicability levels if the inference is based on $p$-values alone. There is certainly ample room for further improving our the current implementation, and  we  discuss some of these extensions and  their related issues below.


\subsection{Future  directions} \label{sec:future_directions}

One natural direction for future work is to develop more powerful data-driven strategies for constructing $\cG_k$, $w_k$, $\mathbf{u}_i$, $p^{{}_{\uiG}}_i$, $\cS_k$, and $\nu_{ik}$ than those proposed in this paper.  In particular, we suspect that $\nu_{ik}$ -- the local PC weight for feature $i$ in group $k$ -- may have an optimal construction under data generated from a two-group model \citep{Efron2008},  by using a weighting scheme similar to that in Proposition 3.2 of \citet{Roquain2008OptimalWF}. Subsequently, it may be possible to derive an “oracle inequality” for the optimal power; see  \citet[Theorem 4.2]{Roquain2008OptimalWF}.

Another direction is to convert the \methodname from a $p$-value-based procedure to an $e$-value-based procedure, in an analogous spirit to how \citet{WangRamdas2022} converts the BH procedure to an $e$-BH procedure. Interest in $e$-values has grown enormously in recent years as an alternative to $p$-values in hypothesis testing \citep{WangVovk2021}. Based on the theoretical underpinnings of the e-BH procedure introduced in \citet{WangRamdas2022}, we suspect that an explicit $e$-value version of the \methodname could offer improved power when the data among the $m$ features within a study are arbitrarily dependent, as in our \thmref{error_control_dep}; we are actively investigating this extension at the time of writing.

\subsection{Current limitations when $u < n$}

Throughout, we have largely emphasized  on the important case of maximum replicability level $u = n$. While this is  the most natural  level a practitioner may request, it is sometimes worthwhile to choose $u < n$, either because of power consideration or because maximum replicability is  not demanded for some applications.

When assessing  less-than-maximum replicability with $u <n$, our  simulation studies in \appref{uleqnsim} suggest that our implementation of the ParFilter therein is less competitive than its performance demonstrated  in \secref{sim} for testing maximum replicability. This is mostly due to the difficulty of selecting appropriate local replicability levels $\mathbf{u}_i$ for $i \in [m]$ in a way that makes $u/[n]$ replicated features more detectable. In  \appref{uleqn}, we have offered a  detailed discussion of this issue, which arises from the ``imbalance'' of a testing configuration with respect to some $u/[n]$ replicated features,  and why the \methodname does not encounter the same difficulty when  $u = n$. 

Given its state-of-the-art performance for $u=n$,  the ParFilter in its current form is a much-needed addition to practitioners'  toolkits  for conducting meta-analyses, and 
complements other existing $\FDR_\rep$ controlling methods, such as the AdaFilter-BH procedure \citep{adafilter}, that have exhibited stronger performance for  $u <n$ as per our simulations in \appref{uleqnsim}. In the same appendix, we have also discussed potential directions for future research aimed at improving the power of the \methodname when $u < n$.

\subsection{Post hoc inference for individual studies}

The \methodname can easily  incorporate a \textit{post hoc} FDR procedure for testing the base hypotheses in a given study $j$. Consider the following candidate rejection set
\begin{equation*}
	\cR_j(\tau_j) \equiv \cR(\hattvec) \cap \{ i \in [m]: p_{ij}  \leq \tau_j \} 
\end{equation*}
 for testing  $H_{1j}, \dots, H_{mj}$, 
where $\tau_j \in [0,1]$ is a rejection threshold for the base $p$-values in study $j$. Above, the base null $H_{ij}$ can only be rejected by $\cR_j(\tau_j)$ if the PC null $H^{u/[n]}_i$ is already rejected by $\cR(\hattvec)$, the output of the ParFilter (\algref{PF}). Thus, $\cR_j(\tau_j)$ provides \textit{post hoc} inferences by building upon the initial discoveries made by $\cR(\hattvec)$. Let
\begin{equation*}
\FDR_{j} \equiv \FDR_{j}(\tau_j) = \mathbb{E} \left[ \frac{\sum_{i \in \cR_j(\tau_j)} \1 \left\{ \text{$H_{ij}$ is true} \right\}   }{|\cR_j(\tau_j)| \vee 1} \right]
\end{equation*}
denote the false discovery rate associated with the base null hypotheses rejected by $\cR_j(\tau_j)$. $\FDR_j$ can be controlled by implementing the post hoc procedure outlined in the theorem below, whose proof provided in \appref{post_hoc} is based on \citet[Theorem 1]{Katsevich2023}:

\begin{theorem}[Post hoc FDR control for an individual study]\label{thm:post_hoc}
Suppose the assumptions and conditions of either \thmref{error_control}$(i)$ or $(ii)$  hold. For a given study $j \in [n]$ and an $\FDR_{j}$  target $q_j \in (0,1)$, let
\begin{equation*}
\hat{\tau}_j \equiv \max \left\{ \tau_j \in \{ p_{ij} : i \in [|\cS_{k(j)}|]  \} : \frac{ |\cS_{k(j)}| \cdot \tau_j}{|\cR_j(\tau_j)| \vee 1}  \leq q_j \right\},
\end{equation*}
where $k(j) \in [K]$ is the group in the testing configuration to which  study $j$ belongs. Then the post hoc rejection set $\cR_j(\hat{\tau}_j)$ has $\FDR_{j}(\hat{\tau}_j)$ control at level $q_j$, i.e., $\FDR_{j}(\hat{\tau}_j) \leq q_j$.
\end{theorem} 

We note that, if each group in  the testing configuration contain only a single study as in \exref{suggested_testing_config_ueqn}, such post hoc FDR inference  is only meaningful if $q_j \leq w_{k(j)} q$, because \thmsref{error_control} and \thmssref{error_control_dep} already ensure that each group-wise $\FDR_{k(j)}$ can  be controlled under level $w_{k(j)} \cdot q$.

\bibliographystyle{apalike}
\bibliography{references}

\appendix

\section{Miscellaneous implementational details} \label{app:implementation_other}

\subsection{Setting the testing configuration when $u < n$}\label{app:uleqn}
Generally speaking, one should aim to choose a testing configuration that is \textit{balanced} for the vast majority of features. In other words, one should avoid choosing a configuration that is \textit{imbalanced} for most features.
\begin{definition}[Imbalance and Balance]\label{def:imbalance}
Let $\{ \cG_k, w_k, \mathbf{u}_{\ell} \}_{k \in [K], \ell \in [m]}$ be a testing configuration for $u/[n]$ replicability, and $i$ be a given $u/[n]$ replicated feature, i.e., the PC null $H^{u/n}_i$ is false. Then the testing configuration is said to be:
\begin{itemize}
    \item imbalanced for feature $i$ if a local PC null $H^{\uiG}_i$ is true for some $k \in [K]$. 
    \item  balanced for feature $i$ if the local PC null $H^{\uiG}_i$ is false for all $k \in [K]$.
\end{itemize}
\end{definition}
To illustrate why minimizing imbalance is important, suppose \emph{every} $u/[n]$ replicated feature is imbalanced because of a chosen testing configuration.  Thus, by \defref{imbalance}, the set of non-nulls $[m] \backslash \cH^{u/[n]}$ is also the set of all imbalanced features and  satisfies
\begin{align}\label{I_inclusion}
[m] \backslash \cH^{u/[n]} \subseteq \bigcup_{{}_{k \in [K]}} \cH_k.
\end{align}
For each $k \in [K]$, the number of false discoveries made by the \methodname among the local PC nulls in group $k$ is ${|\tiny \widehat{\cR} \cap \cH_k|}$, where $\widehat{\cR} = \cR(\hattvec)$ is the output of \algref{PF}. Since the \methodname controls $\FDR_k(\widehat{\cR})$ at level $w_k \cdot q$, where $q$ is typically small (e.g., $q = 0.05$), $|\widehat{\cR} \cap \cH_k|$ is generally small for all $k$. It follows from \eqref{I_inclusion} that the number of correct rejections $|\widehat{\cR} \cap ([m] \backslash \cH^{u/[n]} )|$ is small too, which means that $\widehat{\cR}$ is unlikely to include many $u/[n]$ replicated features.

When $u = n$, it is impossible for a testing configuration to be imbalanced for any feature. Hence, imbalance is not a concern for maximum replicability analyses, which explains the strong performance of the ParFilter in the simulations of \secref{sim}. Unfortunately, the same cannot be said for the case of $u < n$. For a testing configuration in the latter case, our recommendation  is to use $K = 2$ groups with approximately equal sizes and local error weights, as using less groups makes the testing configuration less likely to be imbalanced for any $u/[n]$ replicated feature.  For instance, if  $\text{mod}(a,b)$ denotes the operation $a$ modulo $b$ for integers $a,b \in \bN$, letting
\begin{equation}\label{uleqnconfig}
    \cG_k = \{ j \in [n] : \text{mod}(j,2) + 1 = k\} \quad \text{and} \quad w_k = \frac{|\cG_k|}{n} \quad \text{for each $k \in [2]$}
\end{equation}
assigns the studies $\{1, 3, 5, \dots\}$ to $\cG_2$  and the studies $\{2, 4, 6, \dots\}$ to $\cG_1$, with the  two local error weights proportional to their group sizes. If possible, we recommend using expert knowledge to construct the local replicability levels by analyzing the covariates $x_{i1}, \dots, x_{in}$ to identify which base hypotheses among $H_{i1}, \dots, H_{in}$ are most likely false, and assigning $\mathbf{u}_i \equiv (u_{i1}, \dots, u_{iK})$ accordingly, while ensuring the constraints in \eqref{uvec} are satisfied. 

It may be true that several testing configurations are  considered equally effective in mitigating imbalance. In such scenarios, we recommend  randomly sampling one such configuration to be used in the \methodname  to hedge against overly relying on a single configuration that may not yield powerful results. It is immediately clear that as long as the testing configuration is sampled  independently of $\mathbf{P}$ and $\mathbf{X}$, the $\FDR_{\rep}$ guarantees provided by \thmsref{error_control} and \ref{thm:error_control_dep}  still hold; we state this formally as a theorem:

\begin{theorem}[$\FDR_{\rep}$ control under random testing configurations]\label{thm:random_error_control}
    Suppose 
    \begin{equation}\label{L}
	\{ \cG^{(b)}_k, w^{(b)}_{k}, \mathbf{u}^{(b)}_i \}_{k \in [K^{(b)}], i \in [m]},  \quad \text{$b \in [B]$},
\end{equation}
are  $B \in \mathbb{N}$ testing configurations, and the testing configuration used in the \methodname (\algref{PF}) is $ \{ \cG^{b^*}_k, w^{b^*}_{k}, \mathbf{u}^{b^*}_i \}_{k \in [K_{b^*}], i \in [m]}$, where $b^* \in [B]$ is a random variable that is independent of $\mathbf{P}$ and $\mathbf{X}$. If the assumptions and conditions of either \thmref{error_control} or \ref{thm:error_control_dep} hold, then the \methodname at level $q \in (0,1)$  returns a rejection set $\cR(\hattvec)$ satisfying 
 \begin{equation*}
        \FDR_{k} (\cR(\hattvec)) \leq  w_k \cdot q  \;\ \text{ for each } k \in [K] \;\ \text{ and } \;\  \FDR_{\rep}(\cR(\hattvec)) \leq q.
    \end{equation*}
\end{theorem}

\subsection{Results for the working model \eqref{pmod} in \secref{powerful_weights}}\label{app:working_model}
\subsubsection{Derivation of $\Clfdrc^{\upsilon/\cJ}_i$ for any $\cJ \subseteq [n]$ and $\upsilon \in |\cJ|$ in \eqref{form_for_omega_ivcj}}\label{app:PCClfdrc}
Let $\Clfdrc_{ij} \equiv \Pr(\text{$H_{ij}$ is false} | \mathbf{P},\mathbf{X})$ be the probability that $H_{ij}$ is false given $\mathbf{P}$ and $\mathbf{X}$. 
We have under the working model that
\begin{align*}
\Clfdrc_{ij} &= \Pr(\text{$H_{ij}$ is false}|\mathbf{P},\mathbf{X}) \\
&= \Pr(\text{$H_{ij}$ is false}|p_{ij},x_{ij}) \\
&= \frac{\Pr(\text{$H_{ij}$ is false},p_{ij},x_{ij})}{ \Pr(p_{ij},x_{ij}) } \\
&= \frac{\Pr(p_{ij}| \text{$H_{ij}$ is false}, x_{ij}) \Pr(\text{$H_{ij}$ is false}| x_{ij})  \Pr(x_{ij}) }{ \Pr(p_{ij}|x_{ij}) \Pr(x_{ij})  } \\
&= \frac{ (1 - \pi_{\zetabf_j}(x_{ij})) \cdot  f_{\betabf_{j}}(p_{ij}|x_{ij}) }{  f_{\zetabf_j, \betabf_j}(p_{ij}|x_{ij}) }.
\end{align*}

It follows that, under the independence of $(p_{i1},H_{i1} ), \dots, (p_{in}, H_{in} )$ conditional on ${\bf X}$,
\begin{align*}
\Clfdrc^{\upsilon/\cJ}_i &= \Pr(\text{$H^{\upsilon/\cJ}_i$ is false}| \mathbf{P},\mathbf{X}) \\
&= \Pr(\text{At least $\upsilon$ null base hypotheses among $H_{ij}$ for $j \in \cJ$ are false}| \mathbf{P},\mathbf{X} ) \\
&=  \sum_{ \cL \subseteq \cJ : |\cL| \geq \upsilon }  \prod_{j \in \cJ \backslash \cL} \Pr(\text{$H_{ij}$ is true}| \mathbf{P}, \mathbf{X})  \prod_{j \in \cL} \Pr(\text{$H_{ij}$ is false}| \mathbf{P},\mathbf{X}) \\
&=  \sum_{ \cL \subseteq \cJ : |\cL| \geq \upsilon }  \prod_{j \in \cJ \backslash \cL} \Pr(\text{$H_{ij}$ is true}| p_{ij}, x_{ij})  \prod_{j \in \cL} \Pr(\text{$H_{ij}$ is false}| p_{ij}, x_{ij}) \\
&=  \sum_{ \cL \subseteq \cJ : |\cL| \geq \upsilon }  \prod_{j \in \cJ \backslash \cL} (1 - \Clfdrc_{ij})  \prod_{j \in \cL}  \Clfdrc_{ij}.
\end{align*}

\subsubsection{Derivation of $\mathbb{E}[p_{ij}|x_{ij}]$}\label{app:epx}
Under the working model, we have that
\begin{align*}
\mathbb{E}[p_{ij}|x_{ij}] &= \int^1_0  p \cdot \left(  \pi_{\zetabf_j}(x_{ij}) + (1 - \pi_{\zetabf_j}(x_{ij})) \cdot  (1 - k_{\boldsymbol{\beta}_j}(x_{ij}) ) \cdot p^{-k_{\boldsymbol{\beta}_j}(x_{ij})}) \right) \ dp \\
&= \left[ \frac{\left(1 - \pi_{\zetabf_j}(x_{ij}) \right) \left(1 - k_{\boldsymbol{\beta}_j}(x_{ij})  \right) p^{2 - k_{\boldsymbol{\beta}_j}(x_{ij}) }}{2 - k_{\boldsymbol{\beta}_j}(x_{ij}) } + \frac{p^{2} \pi_{\zetabf_j}(x_{ij}) }{2} \right]^1_0 \\
&= \frac{\left(1 - \pi_{\zetabf_j}(x_{ij}) \right) \left(1 - k_{\boldsymbol{\beta}_j}(x_{ij})  \right)}{2 - k_{\boldsymbol{\beta}_j}(x_{ij}) } + \frac{\pi_{\zetabf_j}(x_{ij})}{2} \\
&= \frac{2 \cdot (1-k_{{\betabf}_j}(x_{ij})) + k_{{\betabf}_j}(x_{ij}) \cdot \pi_{{\zetabf}_j}(x_{ij})  }{4 - 2 \cdot k_{{\betabf}_j}(x_{ij})}.
\end{align*}

\subsection{Forming $  \nu_{ik}$ that satisfies  \condref{nu}$(b)$}\label{app:localPCweights_b}

Let $r_{ik} \equiv u - u_{ik}$ and $\cC_k \equiv [n] \backslash \cG_k$, the set of studies not in group $k$. Let
\begin{equation*}
H^{\uiGc}_i: |\{ j \in \cC_k : \mu_{ij} \in \cA_{i} \}| \geq |\cC_k| - r_{ik} + 1  
\end{equation*}
denote a PC null hypothesis stating that there are no more than $r_{ik} - 1$ signals among the studies in $\cC_k$ for feature $i$. Moreover, let $\hat{\Clfdrc}^{{}_{\uiGc}}_{i}$ be an estimate of $\Clfdrc^{{}_{\uiGc}}_{i} \equiv \Pr ( \text{$H^{{}_{\uiGc}}_i$ is false} | \mathbf{P}, \mathbf{X} )$, probability that $H^{{}_{\uiGc}}_i$ is false given $\mathbf{P}$ and $\mathbf{X}$.


We will use the working model \eqref{pmod} in \secref{powerful_weights} to construct a $\hat{\Clfdrc}^{{}_{\uiGc}}_i$ that depends only on $(\mathbf{P}_{\cdot (- \cG_k)}, \mathbf{X}, \cS_k)$.  We propose substituting this $\hat{\Clfdrc}^{{}_{\uiGc}}_i$ in place of $\tilde{\Clfdrc}^{{}_{\uiG}}_i$ in \eqref{suggested_feature_weights} for each $i \in \cS_k$, so that the resulting local PC weight
\begin{equation*}
 \nu_{ik} = \frac{{\footnotesize \hat{\Clfdrc}^{\uiGc}_i} |\cS_k| }{ \sum_{\ell \in \cS_k} { \footnotesize \hat{\Clfdrc}^{r_{\ell k} / \cC_k}_{\ell} } }
\end{equation*}
satisfies \condref{nu}$(b)$. Provided that $\hat{\Clfdrc}^{{}_{r_{\ell k} / \cC_k}}_{\ell}$ is a good estimate for each $(\ell,k) \in [m] \times [K]$, the local PC weight design above allows $i$ to be more easily rejected by \eqref{RrepEasy} if $H^{{}_{u/[n]}}_i$ is false, since $\nu_{ik}$ is likely to be large  \emph{for all} $k \in [K]$.

 Let  
\begin{equation*}
    (\hat{\zetabf}_{j}, \hat{\betabf}_{j}) \equiv \argmax_{(\zetabf_{j}, \betabf_{j})} \sum_{i \in [m] } \log\left( f_{\zetabf_j, \betabf_j}(p_{ij}| x_{ij})  \right),
\end{equation*}
be estimates for $(\zetabf_{j}, \betabf_{j})$, obtained by maximizing the working likelihood on data from study $j$. Subsequently, let an estimator for $p_{ij}$ be
\begin{equation*}
\hat{p}_{ij} \equiv \frac{2 \cdot (1-k_{\hat{\betabf}_j}(x_{ij})) + k_{\hat{\betabf}_j}(x_{ij}) \cdot \pi_{\hat{\zetabf}_j}(x_{ij})  }{4 - 2 \cdot k_{\hat{\betabf}_j}(x_{ij})},
\end{equation*}
obtained by substituting $\zetabf_{j} = \hat{\zetabf}_j$ and  $\betabf_{j} = \hat{\betabf}_{j}$ into the expression of $\mathbb{E} \left[ p_{ij} | x_{ij} \right]$ in  \appref{epx} under the working model. Let 
\begin{align*}\label{PC_local_hatClfdrc}
    \hat{\Clfdrc}^{\uiGc}_{i} \equiv  \sum_{ \cL \subseteq \cC_k : |\cL| \geq r_{ik} } & \Bigg( \prod_{j \in \cC_k \backslash \cL}  \tfrac{ \pi_{\hat{\zetabf}_{j}}(x_{ij}) }{f_{\hat{\zetabf}_j, \hat{\betabf}_j}( \hat{p}_{ij}|x_{ij} ) }  \cdot  \prod_{j \in \cL}  \tfrac{ (1 -\pi_{\hat{\zetabf}_{j}}(x_{ij}) ) \cdot f_{\hat{\betabf}_{j}}(\hat{p}_{ij}|x_{ij}) }{f_{\hat{\zetabf}_j, \hat{\betabf}_j}(\hat{p}_{ij}|x_{ij} ) } \Bigg), 
\end{align*}
obtained by substituting $\zetabf_{j} = \hat{\zetabf}_j$,  $\betabf_{j} = \hat{\betabf}_{j}$, and $p_{ij} = \hat{p}_{ij}$, for each $j \in \cC_k$, into the expression \eqref{form_for_omega_ivcj} with $v = r_{ik}$ and $\cJ = \cC_k$. By construction, $\hat{\Clfdrc}^{{}_{\uiGc}}_{i}$ only depends on $(\mathbf{P}_{\cdot (- \cG_k)}, \mathbf{X}, \cS_k)$ because all of $\zetabf_j, \betabf_j, \hat{p}_{ij}$ are so.

\subsection{The \methodname as a generalization of \citet{Bogolomov2018}'s two procedures}\label{app:BogHellerEquiv}

For a replicability analysis where $n = 2$ and $u = 2$, we will show that the \methodnamens, under specific algorithm settings, and the adaptive procedure by \citet[Section 4.2]{Bogolomov2018} are equivalent. The equivalence between the \methodname and the non-adaptive procedure by \citet[Section 3.2]{Bogolomov2018} can be shown in a similar fashion, and so, we omit it for brevity. Here, the qualifier ``adaptive'' refers to  adaptiveness to the post-selection weighted proportions of local PC nulls in \eqref{group_pi} (as opposed to adaptiveness to the covariates ${\bf X}$)

Let the number of groups be $K = 2$, the groups be $\cG_{1} = \{ 1 \}$ and $\cG_{2} = \{ 2 \}$, and the local replicability levels be $u_{i1} = u_{i2} = 1$ and the local PC weights be $\nu_{i1} = \nu_{i2} = 1$ for each $i \in [m]$. Moreover, let the weighted null proportion estimator $\hat{\pi}_k$ be in the form of \eqref{pi_hat_adaptive} for each $k \in [2]$ and let the local PC $p$-values be defined as
\begin{equation*}
    p^{u_{i1}/ \cG_1 }_i = p^{1/\{1\}}_i = p_{i1} \quad \text{and} \quad p^{u_{i2}/ \cG_2 }_i =  p^{1/\{2\}}_i = p_{i2} \quad \text{for each $i \in [m]$}.
\end{equation*}
Using the suggestion in \exref{suggested_select_rules}, let the selected features be taken as
\begin{equation*}
    \cS_1 = \{ i \in [m] : p_{i2} \leq (w_2 \cdot q ) \wedge  \lambda_2  \}  \quad \text{and} \quad \cS_2 = \{ i \in [m] : p_{i1} \leq  (w_1 \cdot q) \wedge \lambda_1  \} ,
\end{equation*}
respectively. Under the aforementioned  settings, \algref{PF} returns the following rejection set:
\begin{align}
    \cR(\hattvec) &= \{ i \in \cS_1 : p_{i1} \leq \hatt_{1} \wedge \lambda_1 \} \cap \{ i \in \cS_2 : p_{i2} \leq \hatt_{2} \wedge \lambda_2 \}  \nonumber \\
    &= \{ i \in \cS_1 \cap \cS_2: (p_{i1},p_{i2}) \matleq (\hatt_1 \wedge \lambda_1,\hatt_2 \wedge \lambda_2 ) \}  \label{2Rrep}
\end{align}
where $\hattvec = (\hatt_1,\hatt_2)$ is the optimal vector of thresholds implied by \propref{optimal_thresholds}. The adaptive procedure by \citet[Section 4.2]{Bogolomov2018}, under selected features $\cS_1$ and $\cS_2$ for group 1 and 2 respectively\footnote{\citet{Bogolomov2018} uses $\cS_2$ and $\cS_1$ to denote the selected features for groups 1 and 2, respectively; we reverse their notation to maintain consistency with the convention used throughout this paper.}, yields a rejection set of the form:
\begin{align}\label{BogHeller_rejection_set}
    \tilde{\cR} \equiv  \left\{ \ell \in \cS_1 \cap \cS_2 : (p_{\ell 1},p_{\ell 2}) \matleq \left( \frac{R \cdot w_1 \cdot q}{|\cS_1 | \cdot \hat{\pi}_{1}  } \wedge \lambda_1,  \frac{R \cdot w_2 \cdot q}{|\cS_2 | \cdot \hat{\pi}_{2} } \wedge \lambda_2 \right)   \right\}
\end{align}
where 
\begin{align}\label{R}
    R \equiv \max \left\{ r : \sum_{\ell \in \cS_1 \cap \cS_2} \1 \left\{ (p_{\ell 1},p_{\ell 2}) \matleq \left( \frac{r \cdot w_1 \cdot q}{|\cS_1 | \cdot \hat{\pi}_{1}  } \wedge \lambda_1,  \frac{r \cdot w_2 \cdot q}{|\cS_2| \cdot \hat{\pi}_{2}  } \wedge \lambda_2 \right) \right\}  = r \right\}.
\end{align}
If $R$ does not exist, then $\tilde{\cR}$ is defined to be an empty set. To complete this proof, we will show that $\tilde{\cR} = \cR(\hattvec)$ by proving the  two inclusions
\begin{equation*}
    \tilde{\cR} \subseteq \cR(\hattvec) \quad \text{and} \quad  \cR(\hattvec) \subseteq \tilde{\cR}.
\end{equation*}

\begin{proof} [Proof of $\tilde{\cR} \subseteq \cR(\hattvec)$:]
Assume $\tilde{\cR}$ is non-empty and let $\tilde{\tvec} = (\tilde{t}_1,\tilde{t}_2)$ where
\begin{equation}\label{tildet}
    \tilde{t}_1 =  \frac{R \cdot w_1 \cdot q}{|\cS_1| \cdot \hat{\pi}_{1}  }  \quad \text{ and } \quad \tilde{t}_2 =  \frac{R \cdot w_2 \cdot q}{|\cS_2| \cdot \hat{\pi}_{2}  } .
\end{equation}
It follows from \eqref{BogHeller_rejection_set} and the form of the \methodnamens's rejection set given in \eqref{2Rrep} that
\begin{align*}
    \tilde{\cR} = \left\{ i \in \cS_1 \cap \cS_1 : (p_{i1},p_{i2}) \matleq \left( \tilde{t}_1 \wedge \lambda_1, \tilde{t}_2 \wedge \lambda_2 \right)  \right\} = \cR(\tilde{\tvec}).
\end{align*}
Since $R = |\tilde{\cR}|$ by \eqref{R}, it follows from the relationship above and \eqref{tildet} that
\begin{align*}
    \tilde{t}_1 = \frac{|\cR(\tilde{\tvec})| \cdot w_1 \cdot q}{|\cS_{1}| \cdot \hat{\pi}_{1}  }  \quad \text{ and } \quad \tilde{t}_2 = \frac{|\cR(\tilde{\tvec})| \cdot w_2 \cdot q}{|\cS_{2}| \cdot \hat{\pi}_{2}  } .
\end{align*}
The equalities above can be rearranged to yield 
\begin{equation*}
    \FDPhat_{1}(\tilde{\tvec}) \equiv \frac{|\cS_1| \cdot \hat{\pi}_{1} \cdot \tilde{t}_1}{|\cR(\tilde{\tvec})|} = w_1 \cdot q  \quad \text{ and }  \quad  \FDPhat_{2}(\tilde{\tvec}) \equiv \frac{|\cS_2| \cdot \hat{\pi}_{2} \cdot \tilde{t}_2}{|\cR(\tilde{\tvec})|} = w_2 \cdot q,
\end{equation*}
which implies that $\tilde{\tvec} \in \cT$ by \defref{set_of_feasible_thresholds}.  Since $\hattvec$ is the largest element-wise vector in $\cT$ by \propref{optimal_thresholds}, it must be the case that $\hattvec \matgeq \tilde{\tvec}$. Hence, it follows from \lemref{Rrep} and $\hattvec \matgeq \tilde{\tvec}$ that
\begin{equation*}
    \tilde{\cR} = \cR(\tilde{\tvec}) \subseteq \cR(\hattvec).
\end{equation*}
\end{proof}

\begin{proof} [Proof of $\cR(\hattvec) \subseteq \tilde{\cR}$:]
Assume $\cR(\hattvec)$ is non-empty. We have that
\begin{align*}
    \sum_{i \in \cS_1 \cap \cS_2 } & \1 \left\{ (p_{i1},p_{i2}) \matleq \left( \frac{|\cR(\hattvec)| \cdot w_1 \cdot q}{|\cS_{1}| \cdot \hat{\pi}_{1}  } \wedge \lambda_1 ,  \frac{|\cR(\hattvec)| \cdot w_2 \cdot q}{|\cS_{2}| \cdot \hat{\pi}_{2} } \wedge \lambda_2 \right) \right\} \\
     &\overset{\tiny (a)}{=} \sum_{i \in \cS_1 \cap \cS_2 } \1 \left\{ (p_{i1},p_{i2}) \matleq (\hatt_1 \wedge \lambda_1, \hatt_2 \wedge \lambda_2) \right\} \\ 
     & = \left| \{ i \in \cS_1 \cap \cS_2 : (p_{i1},p_{i2}) \matleq (\hatt_1 \wedge \lambda_1,\hatt_2 \wedge \lambda_2 ) \} \right| \\
     &\overset{\tiny (b)}{=} |\cR(\hattvec)|
\end{align*}
where $(a)$ is a result of \lemref{tvechat_equality} and $(b)$ is a result \eqref{2Rrep}.  By the definition of $R$ in \eqref{R}, it must follow from the display above that
\begin{align}\label{Rgeq}
    R \geq |\cR(\hattvec)|.
\end{align}
Hence, 
\begin{align*}
\cR(\hattvec) &= \left\{ i \in \cS_1 \cap \cS_2 : (p_{i1},p_{i2}) \matleq \left( \hatt_1 \wedge \lambda_1, \hatt_2 \wedge \lambda_2 \right)  \right\} \\
&\overset{\tiny (c)}{=}  \left\{ i \in \cS_1 \cap \cS_2 : (p_{i1},p_{i2}) \matleq \left( \frac{|\cR(\hattvec)| \cdot w_1 \cdot q}{|\cS_{1}| \cdot \hat{\pi}_{1}  } \wedge \lambda_1,  \frac{|\cR(\hattvec)| \cdot w_2 \cdot q}{|\cS_{2}| \cdot \hat{\pi}_{2} } \wedge \lambda_2 \right)  \right\} \\
&\overset{\tiny (d)}{\subseteq} \left\{ i \in \cS_1 \cap \cS_2 : (p_{i1},p_{i2}) \matleq \left( \frac{R \cdot w_1 \cdot q}{|\cS_1| \cdot \hat{\pi}_{1} } \wedge \lambda_1,  \frac{R \cdot w_2 \cdot q}{|\cS_2| \cdot \hat{\pi}_{2} } \wedge \lambda_2 \right)  \right\}  \\
&= \tilde{\cR}
\end{align*}
where $(c)$ is a result of \lemref{tvechat_equality} and $(d)$ is a result of \eqref{Rgeq}. Thus, $\cR(\hattvec) \subseteq \tilde{\cR}$.
\end{proof}

\subsection{Proof of \thmref{post_hoc} (Post hoc FDR control for individual studies)}\label{app:post_hoc}

First, note that
\begin{equation*}
\cR_j(\hat{\tau}_j) \equiv  \cR(\hattvec) \cap \{ i \in [m] : p_{ij} \leq \hat{\tau}_j \} \overset{\tiny (a)}{=}  \cR(\hattvec) \cap \{ i \in \cS_{k(j)} : p_{ij} \leq \hat{\tau}_j \}
\end{equation*}
where $k(j)$ is the group where study $j$ belongs, and $(a)$ follows from the fact that $\cR(\hattvec) \subseteq \cS_k$ for any group $k$ by \eqref{RrepEasy}. 

$\cR_j(\hat{\tau}_j) = \cR(\hattvec) \cap \{ i \in \cS_{k(j)} : p_{ij} \leq \hat{\tau}_j \} $ can be interpreted as the output of the Focused BH procedure by \citet[Procedure 1]{Katsevich2023} for simultaneously testing
\begin{center}
 $H_{ij}$ for $i \in \cS_{k(j)}$,
\end{center}
where $\hat{\tau}_j$ is equivalent to their equation $(6)$ and $\cR(\hattvec)$ is what they call a \textit{screening function}. By \lemref{stable} in \appref{technical_lemmas}, $\cR(\hattvec)$ is a stable screening function according to Theorem 1$(ii)$ of \citet{Katsevich2023} when $\mathbf{X}$ is given. Since conditioning on $\{ \mathbf{P}_{\cdot \cG_{k(j)} }, \mathbf{X} \}$ fixes $\cS_{k(j)}$ by \condref{indep}, it will hold by Theorem 1 of \citet{Katsevich2023} that 
\begin{equation*}
 \mathbb{E}\left[ \frac{\sum_{i \in \cR_j(\hat{\tau}_j)} \1 \{ \text{$H_{ij}$ is true} \} }{|\cR_j(\hat{\tau}_j)| \vee 1}  \Bigg| \mathbf{P}_{\cdot \cG_{k(j)} }, \mathbf{X}  \right] \leq q_j.
\end{equation*}
Taking expectation in the prior display gives  that $\FDR_{j}(\hat{\tau}_j) \leq q_j$.

\section{Technical  lemmas}\label{app:technical_lemmas}

\subsection{Lemmas regarding properties of the rejection sets}
This section states and proves technical lemmas that  concern some properties of the rejection sets considered by the ParFilter. These lemmas have been used in a few  proofs in \appref{implementation_other}, and will primarily be used  in 
the proofs  in \appref{error_control}.  
Hereafter, we will occasionally use
\begin{center}
$\cR(\tvec) = \cR(\tvec; \mathbf{P}, \mathbf{X})$, $\quad \FDPhat_{k}(\tvec) = \FDPhat_{k}(\tvec; \mathbf{P}, \mathbf{X})$, and $\quad \cT = \cT (\mathbf{P}; \mathbf{X})$
\end{center}
to emphasize that $\cR(\tvec)$ in \eqref{RrepEasy}, $\FDPhat_{k}(\tvec)$ in \eqref{FDPhat_repk} and $\cT$ in \defref{set_of_feasible_thresholds}  also depend on the base $p$-values $\mathbf{P}$ and covariates $\mathbf{X}$.



\begin{lemma}[$\cR(\tvec'') \subseteq \cR(\tvec')$ for  $\tvec' \matgeq \tvec''$]\label{lem:Rrep}
Let $\tvec' , \tvec''  \in [0,\infty)^{K}$ be vectors that satisfy $\tvec' \matgeq \tvec''$. Then $\cR(\tvec'') \subseteq \cR(\tvec')$.
\end{lemma} 

\begin{proof}[Proof of \lemref{Rrep}] 
Write $\tvec' = (t'_1,\dots,t'_K)$ and $\tvec'' = (t''_1,\dots,t''_K) \in [0,\infty)^{K}$.
Since $\tvec' \matgeq \tvec''$, it follows  that 
$$\left\{ i \in \cS_k : p^{\uiG}_{i}  \leq \nu_{ik} \cdot t''_k \wedge \lambda_k\right\} \subseteq \left\{ i \in \cS_k :  p^{\uiG}_{i} \leq \nu_{ik} \cdot t'_k\wedge \lambda_k \right\}$$
for each $k \in [K]$. Thus, $cR(\tvec'') \subset \cR(\tvec')$ from their definitions in \eqref{RrepEasy}.
\end{proof}

The following lemma shows that $\cR(\hattvec)$ satisfies a property akin to the \textit{self-consistency} condition described in  \citet[Definition 2.5]{Blanchard2008}.
\begin{lemma}[Self-consistency with respect to each group]\label{lem:gw_sc}
 Let $\cR(\hattvec)$ be the rejection set output of \algref{PF}, where $\hattvec$ is the optimal vector of thresholds in $\cT$ from \propref{optimal_thresholds}. If the selected features $\cS_k$ for group $k$  is non-empty, then it holds that
\begin{equation}\label{gw_indicator_inequality}
        I\left\{ i \in \cR(\hattvec) \right\} \leq I\left\{  p^{\uiG}_{i} \leq \frac{ \nu_{ik} \cdot (|\cR(\hattvec)| \vee 1) \cdot w_k \cdot q}{(|\cS_k| \vee 1) \cdot \hat{\pi}_k} \right\} \quad \text{for any $i \in [m]$.}
\end{equation}
\end{lemma}
\bigskip

\begin{proof}[Proof of \lemref{gw_sc}]
    If $\cR(\hattvec)$ is empty, then $I\left\{ i \in \cR(\hattvec) \right\} = 0$ and \eqref{gw_indicator_inequality} trivially holds.
    Now consider the case where $\cR(\hattvec)$ is non-empty. Since $\hattvec \in \cT$ by \propref{optimal_thresholds} (whose proof in \appref{optimal_thresholds} only depends on \lemref{Rrep}), it follows by \defref{set_of_feasible_thresholds} that $\FDPhat_k(\hattvec) \leq w_k \cdot q$, which upon rearrangement yields
    \begin{align}
         \hat{t}_{k} &\leq \frac{ |\cR(\hattvec)| \cdot w_k \cdot q}{|\cS_k| \cdot \hat{\pi}_{k}}.  \label{FDPhatk_rearrange}
    \end{align}
    Since $\cR(\hattvec) \subseteq \{ \ell \in \cS_k : p^{u_{\ell k} / \cG_k }_{\ell} \leq \nu_{\ell k} \cdot \hatt_k \}$ by \eqref{RrepEasy}, it follows from \eqref{FDPhatk_rearrange} that
    \begin{equation*}
         i \in \cR(\hattvec) \quad \text{ implies } \quad  p^{\uiG}_{i} \leq \nu_{ik} \cdot \frac{|\cR(\hattvec)| \cdot w_k \cdot q}{|\cS_k| \cdot \hat{\pi}_{k}},
    \end{equation*}
   and the proof is finished, as $|\cR(\hattvec)| = |\cR(\hattvec)| \vee 1$ and $|\cS_k| = |\cS_k| \vee 1$ when $\cR(\hattvec)$ is non-empty.
\end{proof}

\begin{lemma}[$\hattvec$ as a function of $\cR(\hattvec)$]\label{lem:tvechat_equality}
    Let $\cR(\hattvec)$ be the output of \algref{PF} where $\hattvec \in \cT$ is the optimal vector of thresholds  from \propref{optimal_thresholds}. If $\cR(\hattvec)$ is non-empty, then
    \begin{equation*}
        \hattvec = (\hatt_1,\dots,\hatt_K) = \left( \frac{|\cR(\hattvec)|  \cdot w_1 \cdot q}{|\cS_1| \cdot \hat{\pi}_{1} }, \dots, \frac{|\cR(\hattvec)| \cdot w_K \cdot q}{|\cS_K| \cdot \hat{\pi}_{K} } \right).
    \end{equation*}
\end{lemma}

\bigskip

\begin{proof}[Proof of \lemref{tvechat_equality}]
    Since $\hattvec \in \cT$ by \propref{optimal_thresholds}, we have by \defref{set_of_feasible_thresholds} that
\begin{align}\label{error_inequalities}
    \left( \FDPhat_1(\hattvec), \dots, \FDPhat_{K}(\hattvec) \right) \matleq (w_1 \cdot q, \dots,  w_K \cdot q).
\end{align}
Using the definition for $\FDPhat_{k}$  in \eqref{FDPhat_repk}, rearranging the above relationship when $\cR(\hattvec)$ is non-empty yields
\begin{align}\label{hattvecinequality}
   \underbrace{\left( \frac{|\cR(\hattvec)| \cdot w_1 \cdot q}{|\cS_1| \cdot \hat{\pi}_{1}  }, \dots,   \frac{|\cR(\hattvec)| \cdot w_K \cdot q}{|\cS_K| \cdot \hat{\pi}_{K}  } \right) }_{ 
   \begin{matrix}
       \equiv \hattvec^*
   \end{matrix}
   } 
   \matgeq  (\hatt_1, \dots, \hatt_{K})  = \hattvec. 
\end{align}
Note that if $\cR(\hattvec)$ is non-empty, then $\cS_k$ must also be non-empty (i.e., $|\cS_k| > 0$) for each $k \in [K]$. From the inequality in \eqref{hattvecinequality}, we have by \lemref{Rrep} that
\begin{align}\label{RhatsubsetUB}
    \cR(\hattvec) \subseteq \cR \left(  \hattvec^* \right).
\end{align}
Moreover, 
\begin{align*}
    \FDPhat_{k} \left(  \hattvec^*   \right) &= \frac{ |\cS_k| \cdot \hat{\pi}_{k}  \cdot \left(    \frac{|\cR(\hattvec)| \cdot w_k \cdot q}{|\cS_k| \cdot \hat{\pi}_{k} } \right)   }{ \left| \cR \left(  \hattvec^*   \right) \right| \vee 1 } \\
    & = \frac{ |\cR(\hattvec)| \cdot w_k \cdot q  }{ \left| \cR \left(  \hattvec^* \right) \right| \vee 1 }  \\
    &\overset{\tiny (a)}{\leq} \frac{ |\cR(\hattvec)| \cdot  w_k \cdot q }{ |\cR(\hattvec) | \vee 1 } \\
    &\leq w_k \cdot q \quad \text{ for all $k \in [K]$ }
\end{align*}
where (a) is a result of \eqref{RhatsubsetUB}. Hence, it follows from the display above and \defref{set_of_feasible_thresholds} that
\begin{equation}\label{rearrangedinsideT}
    \hattvec^* \in \cT.
\end{equation}
Since $\hattvec$ is the largest element-wise vector in $\cT$ (\propref{optimal_thresholds}), it follows from \eqref{hattvecinequality} and \eqref{rearrangedinsideT} that $ \hattvec^* = \hattvec$, thereby completing the proof.
\end{proof}

\begin{lemma}[Stability of $\cR(\hattvec)$]\label{lem:stable}
Consider an implementation of the \methodname (\algref{PF}) where
\begin{itemize}
    \item $f_{\cS_k}$ ensures that $\cS_k$ stable for each $k \in [K]$, i.e., \condref{stable} is met;
    \item $f_{\nubf_k}$ ensures that $\nubf_k$ a function of  $\mathbf{P}_{(-\cS_k)\cG_k}$,  $\mathbf{X}$, and   $\cS_k$ only  for each $k \in [K]$, i.e., \condref{nu}$(a)$ is met;
    \item $f_{\hat{\pi}_{k}}$ ensures that $\hat{\pi}_k$ is either one or in the form of \eqref{pi_hat_adaptive} for each $k \in [K]$.  
    \end{itemize} 
    Then the rejection set  $\cR(\hattvec)$ outputted by the \methodname is stable, i.e. for any $i \in \cR(\hattvec)$,
    \begin{center}
        fixing $\mathbf{P}_{(-i) \cdot}$ and $\mathbf{X}$, and changing $\mathbf{P}_{i \cdot}$ so that $i \in \cR(\hattvec)$ still holds
    \end{center}
    will not change the set $\cR(\hattvec)$.
\end{lemma}

\begin{proof}[Proof of \lemref{stable}]
The following proof shares arguments similar the proof of   \citet[Theorem 1$(ii)$]{Katsevich2023}.    Suppose $\mathbf{P}'$ and $\mathbf{P}''$ are two $m \times n$ matrices of base $p$-values. For each  $k \in [K]$, let
    \begin{equation}\label{local_PC_pv_for_stability}
        {p'}^{{u}_{ \ell k} / \cG_k}_{ \ell } = f_{ \ell k}(\mathbf{P}'_{ \ell \cG_k}; {u}_{ \ell k}) \quad \text{and} \quad {p''}^{u_{ \ell k} / \cG_k}_{ \ell } = f_{ \ell k}(\mathbf{P}''_{ \ell \cG_k}; u_{ \ell k}) \text{ for each feature } \ell \in [m],
    \end{equation}
    \begin{equation*}
        \cS'_k = f_{\cS_k}(\mathbf{P}'; \mathbf{X}) \quad \text{and} \quad \cS''_k = f_{\cS_k}(\mathbf{P}''; \mathbf{X}),
    \end{equation*}
      \begin{align}\label{stable_weights}
        \nubf'_k = (\nu'_{1k},\dots,\nu'_{mk}) = f_{\nubf_k}(\mathbf{P}'; \mathbf{X})  \quad \text{and}  \quad
        \nubf'' = (\nu''_{1k},\dots,\nu''_{mk}) = f_{\nubf_k}(\mathbf{P}''; \mathbf{X}),
    \end{align}
    as well as
       \begin{equation}\label{stable_pihat}
       \hat{\pi}'_{k} = f_{\hat{\pi}_{k}}(\mathbf{P}'; \mathbf{X}, \lambda_k) \quad \text{ and } \quad \hat{\pi}''_{k} = f_{\hat{\pi}_{k}}(\mathbf{P}''; \mathbf{X}, \lambda_k) 
    \end{equation}
    denote the  local PC $p$-values, selected features, local PC weights and weighted null proportion estimators for group $k$  constructed from $\mathbf{P}'$ and $\mathbf{P}''$ respectively, with other arguments such as ${\bf X}$ and $\lambda_k$ being the same.    
    Let 
        \begin{equation*}
        \text{$\hattvec' = (\hat{t}'_1,\dots,\hat{t}'_K)  \quad$ and $\quad \hattvec'' = (\hat{t}''_1,\dots,\hat{t}''_K) $}
    \end{equation*}
   be the optimal thresholds implied by \propref{optimal_thresholds} in 
   \begin{equation*}
   \cT (\mathbf{P}';\mathbf{X})  \quad  \text{and} \quad \cT (\mathbf{P}'';\mathbf{X})
   \end{equation*}
   respectively. By \propref{optimal_thresholds}, we have that
   \begin{align}
    \FDPhat_k(\hattvec';\mathbf{P}',\mathbf{X}) = &\frac{|\cS'_k| \cdot \hat{\pi}'_k \cdot \hatt'_k}{|\cR(\hattvec';\mathbf{P}', \mathbf{X})|\vee 1} \leq w_k \cdot q  \quad \text{and} \label{FDPdash} \\
    & \quad  \FDPhat_k(\hattvec'';\mathbf{P}'',\mathbf{X}) = \frac{|\cS''_k| \cdot \hat{\pi}''_k \cdot \hatt''_k}{|\cR(\hattvec'';\mathbf{P}'', \mathbf{X})|\vee 1} \leq w_k \cdot q  \label{FDPdashdash}
   \end{align}
    for each group $k \in [K]$.
    
    We will \textit{assume} that
  \begin{equation}\label{Rrep_ass}
   \mathbf{P}'_{(-i) \cdot}   =   \mathbf{P}_{(-i) \cdot}'', \quad \text{and} \quad
   i \in \cR( \hattvec'; \mathbf{P}', \mathbf{X}) \cap \cR( \hattvec''; \mathbf{P}'', \mathbf{X}) \text{ for a given } i \in [m],
    \end{equation}
    i.e. the matrices $\mathbf{P}'$ and $\mathbf{P}''$ only possibly differ in the $i$-th row in such a way that feature $i$ is still in both of the \methodname rejection sets based on  $(\mathbf{P}',\mathbf{X}$) and $(\mathbf{P}'',\mathbf{X})$. To finish the proof, it suffices to show under the assumption in \eqref{Rrep_ass} that
    \begin{equation} \label{stability_parfilter_goal}
    \cR(\hattvec';\mathbf{P}', \mathbf{X}) = \cR(\hattvec'';\mathbf{P}'', \mathbf{X}).
     \end{equation}

Since $f_{\cS_k}$ is a stable selection rule and
\begin{equation}\label{R_subset_S}
    \text{$\cR(\hattvec';\mathbf{P}', \mathbf{X}) \subseteq   \cS'_k$ and $\cR(\hattvec'';\mathbf{P}'', \mathbf{X}) \subseteq   \cS''_k$ for each $k \in [K]$,}
\end{equation}
it follows from the assumption in \eqref{Rrep_ass} that
\begin{equation}\label{Sk_equality}
    \cS'_k = \cS''_k \text{ for each } k \in [K].
\end{equation}
A consequence of \eqref{Rrep_ass}, \eqref{R_subset_S}, and \eqref{Sk_equality} is that 
\begin{equation}\label{training_p}
    \mathbf{P}'_{(-\cS'_k) \cG_k} = \mathbf{P}''_{(-\cS''_k) \cG_k} \text{ for each $k \in [K]$.}
\end{equation}
By \condref{nu}$(a)$, we have by \eqref{stable_weights}, \eqref{Sk_equality}, and \eqref{training_p} that
\begin{align}\label{equal_feature_weights}
    \nubf'_k  = \nubf''_k \text{ for each $k \in [K]$.}
\end{align}
By the definition of $\cR({\bf t})$ in \eqref{RrepEasy}, we have by \eqref{Rrep_ass} that 
    \begin{equation}\label{i_in_Rrep_result}
        i \in \left\{ \ell \in \cS'_k: {p'}^{u_{\ell k} / \cG_k}_{\ell}  \leq (\nu'_{\ell k} \cdot \hatt'_{k}) \wedge \lambda_k \right\} \cap \left\{ \ell \in \cS''_k: {p''}^{u_{\ell k}/\cG_k}_{\ell} \leq (\nu''_{\ell k} \cdot \hatt''_{k}) \wedge \lambda_k \right\}  \text{ for each } k \in [K].
    \end{equation}
Since $\mathbf{P}'$ and $\mathbf{P}''$ only possibly differ in the $i$-th row, it must follow from \eqref{local_PC_pv_for_stability} that
\begin{equation} \label{local_pc_pv_are_equal_except_row_i}
    {p'}^{u_{\ell k}/\cG_k}_{\ell} = {p''}^{u_{\ell k}/\cG_k}_{\ell} \text{ for all } \ell \in [m] \backslash \{ i \} \text{ and } k \in [K].
    \end{equation}
Together, \eqref{Sk_equality}, \eqref{equal_feature_weights}, \eqref{i_in_Rrep_result}, and \eqref{local_pc_pv_are_equal_except_row_i} imply that
    \begin{align}\label{R_equality}
        \left\{ \ell \in \cS'_k: {p'}^{u_{\ell k} / \cG_k}_{\ell}  \leq (\nu'_{\ell k} \cdot \hatt^{\max}_{k}) \wedge \lambda_k \right\} = \left\{ \ell \in \cS''_k:  {p''}^{u_{\ell k} / \cG_k}_{\ell} \leq (\nu''_{\ell k} \cdot \hatt^{\max}_{k}) \wedge \lambda_k \right\} \text{ for each } k \in [K],
\end{align}
where
\begin{equation}\label{t_k_max_def}
    \hatt^{\max}_{k} = \hatt'_{k} \vee \hatt''_{k}.
\end{equation}
It follows from \eqref{R_equality} and the definition of $\cR({\bf t})$ in \eqref{RrepEasy} that
\begin{equation}\label{Rrep_equality}
    \cR(\hattvec^{\max};\mathbf{P}', \mathbf{X}) = \cR(\hattvec^{\max};\mathbf{P}'', \mathbf{X}) \text{ \ where \ $\hattvec^{\max} \equiv  (\hatt^{\max}_1,\dots,\hatt^{\max}_K)$.}
\end{equation}
Thus, to finish proving \eqref{stability_parfilter_goal}, it is enough to show that 
    \begin{equation} \label{tmax_equal_everything}
        \hattvec^{\max} = \hattvec' \text{ and }       \hattvec^{\max} =  \hattvec''.
     \end{equation}
      

    We will now proceed to show that 
    \begin{equation}\label{pihat_rep_k_equality}
        \hat{\pi}'_{k} = \hat{\pi}''_{k} \text{ for each } k \in [K],
    \end{equation}
    which holds trivially when each $f_{\hat{\pi}_{k}}$ outputs $1$ regardless of its inputs. Therefore, we will focus on establishing \eqref{pihat_rep_k_equality} in the case where each $f_{\hat{\pi}_{k}}$ takes the form given in \eqref{pi_hat_adaptive}. Assumption \eqref{Rrep_ass}, together with the definition of $\cR(\mathbf{t})$ in \eqref{RrepEasy}, implies that
    \begin{equation}\label{pi_less_than_lambda_k}
        i \notin \left\{ \ell \in \cS'_k: {p'}^{u_{\ell k}/\cG_k}_{\ell} > \lambda_k \right\}  \cup \left\{ \ell \in \cS''_k: {p''}^{u_{\ell k}/\cG_k}_{\ell} > \lambda_k \right\} \text{ for each } k \in [K].
    \end{equation}
    It follows from \eqref{stable_pihat}, \eqref{Sk_equality}, \eqref{equal_feature_weights}, \eqref{local_pc_pv_are_equal_except_row_i}, \eqref{pi_less_than_lambda_k}, and the form of \eqref{pi_hat_adaptive} that
    \begin{align*}
       \hat{\pi}'_{k} &=  \frac{ \max\limits_{\ell \in \cS'_k} \nu'_{\ell k} + \sum\limits_{\ell \in \cS'_k \backslash \{i\}} \nu'_{\ell k} \cdot I \left\{ {p'}^{u_{\ell k}/\cG_k}_{\ell} > \lambda_k \right\}}{(1 - \lambda_k) \cdot |\cS'_k| } \\
       &= \frac{ \max\limits_{\ell \in \cS''_k} \nu''_{\ell k} + \sum\limits_{\ell \in \cS''_k \backslash \{i\}} \nu''_{\ell k} \cdot I \left\{ {p''}^{u_{\ell k}/\cG_k}_{\ell} > \lambda_k \right\}}{(1 - \lambda_k) \cdot |\cS''_k| } = \hat{\pi}''_{k}
    \end{align*}
    for each $k \in [K]$. Thus, \eqref{pihat_rep_k_equality} holds if each $f_{\hat{\pi}_{k}}$ is taken in the form of \eqref{pi_hat_adaptive}.
    
    Note that, for each $k \in [K]$,
    \begin{align}\label{FDPhat_repk_TBD}
        \FDPhat_{k}(\hattvec^{\max};\mathbf{P}', \mathbf{X}) = \frac{|\cS'_{k}| \cdot \hat{\pi}'_{k} \cdot \hatt^{\max}_{k}}{|\cR(\hattvec^{\max};\mathbf{P}', \mathbf{X})| \vee 1}.
    \end{align}
    Then, consider the following two exhaustive cases in light of the definition of $\hatt^{\max}_{k}$ in \eqref{t_k_max_def}:
    \begin{itemize}
        \item Case 1: $\hatt^{\max}_{k} = \hatt'_k$. From \eqref{FDPhat_repk_TBD}, we have that 
        \begin{align*}
            \FDPhat_{k}(\hattvec^{\max};\mathbf{P}', \mathbf{X})  &= \frac{|\cS'_{k}| \cdot \hat{\pi}'_{k} \cdot \hatt'_k}{|\cR(\hattvec^{\max};\mathbf{P}', \mathbf{X})| \vee 1} \\
            &\overset{\tiny (a)}{\leq} \frac{|\cS'_{k}| \cdot \hat{\pi}'_{k} \cdot \hatt'_k}{|\cR(\hattvec';\mathbf{P}', \mathbf{X})| \vee 1} \\
            & = \FDPhat_{k}(\hattvec';\mathbf{P}', \mathbf{X})  \\
            &\overset{\tiny (b)}{\leq} w_k \cdot q
        \end{align*}
        where (a) is a result of \lemref{Rrep}  since $\hattvec^{\max} \matgeq \hattvec'$, and (b) is a result of \eqref{FDPdash}.
        \item Case 2: $\hatt^{\max}_{k} = \hatt''_k$. From \eqref{FDPhat_repk_TBD}, we have that 
        \begin{align*}
             \FDPhat_{k}(\hattvec^{\max};\mathbf{P}', \mathbf{X}) &\overset{\tiny (c)}{=} \frac{|\cS''_{k}| \cdot \hat{\pi}''_{k} \cdot \hatt''_k}{|\cR(\hattvec^{\max};\mathbf{P}'', \mathbf{X})| \vee 1} \\
            &\overset{\tiny (d)}{\leq} \frac{|\cS''_{k}| \cdot \hat{\pi}''_{k} \cdot \hatt''_k}{|\cR(\hattvec'';\mathbf{P}'', \mathbf{X})| \vee 1} \\
            &= \FDPhat_{k}(\hattvec'';\mathbf{P}'', \mathbf{X}) \\
            &\overset{\tiny (e)}{\leq} w_k \cdot q
        \end{align*}
        where $(c)$ is a result of \eqref{Sk_equality}, \eqref{Rrep_equality}, and \eqref{pihat_rep_k_equality}, $(d)$ is a result of \lemref{Rrep} since $\hattvec^{\max} \matgeq \hattvec''$, and $(e)$ is a result of \eqref{FDPdashdash}.
    \end{itemize}
    Since we have proven  $\FDPhat_{k}(\hattvec^{\max};\mathbf{P}', \mathbf{X}) \leq w_k \cdot q$ for all $k \in [K]$, it follows by \defref{set_of_feasible_thresholds} that
    \begin{equation}\label{taurepstar}
        \hattvec^{\max} \in \cT(\mathbf{P}'; \mathbf{X}).
    \end{equation}
    Furthermore, since $\hattvec'$ is the largest element-wise vector in $\cT(\mathbf{P}'; \mathbf{X})$ (\propref{optimal_thresholds}) and $\hattvec^{\max} \matgeq \hattvec'$ is true by construction, it must follow from \eqref{taurepstar} that $\hattvec^{\max} = \hattvec'$. A completely analogous proof will also yield $\hattvec^{\max} = \hattvec''$. Hence, \eqref{tmax_equal_everything} is proved.
\end{proof}

\subsection{Additional lemmas on conditional validity } 
The additional lemmas below concerning the conditional validity of $p^{u_{i k} / \cG_k}_{i}$ play key roles in the proofs of \thmsref{error_control} and \ref{thm:error_control_dep} in \appref{error_control}.

\begin{lemma}[Additional conditional validity  of $p^{u_{i k} / \cG_k}_{i}$  under arbitrarily dependence]\label{lem:supertX}
Suppose \assumpref{indepstudies} holds.  Furthermore, suppose for any given feature $i$ and group $k$, the local PC $p$-value $p^{\uiG}_i$ is  valid conditional on ${\bf X}$, and the selection set $\cS_k$  depends only on $\mathbf{P}_{ \cdot (-\cG_k)}$ and ${\bf X}$ (i.e.,  \condsref{valid_monotone_local_PC_pv} and \ref{cond:indep} are met). Then,  for all $t \in [0,1]$,
\begin{center}
$ \Pr ( p^{u_{i k} / \cG_k}_{i} \leq t | i \in \cS_k, \mathbf{P}_{\cdot (-\cG_k)},  \mathbf{X}  ) \leq t$ if $H^{\uiG}_i$ is true.
\end{center}
\end{lemma}

\begin{proof}[Proof of \lemref{supertX}]
Under \assumpref{indepstudies} and \condref{indep}, it is obvious that
\begin{equation}\label{indepPX}
\text{$\mathbf{P}_{i \cG_k}$ and $\{i \in \cS_k,  \mathbf{P}_{\cdot (-\cG_k)}\}$ are independent conditional on $\mathbf{X}$.}
\end{equation}
The lemma then holds by the definition of $p^{\uiG}_i$ as a function of ${\bf P}_{i \cG_k}$ in \eqref{generic_local_pv}, the conditional independence  stated in \eqref{indepPX}, and  \condsref{valid_monotone_local_PC_pv}.
\end{proof}

\begin{lemma}[Additional conditionally validity  of $p^{u_{i k} / \cG_k}_{i}$ under independence]
\label{lem:supert}
Suppose \assumpref{indepstudies} holds, and the base $p$-values in $\mathbf{P}_{\cdot j} = (p_{1j},p_{2j},\dots,p_{mj})$ are independent conditional on $\mathbf{X}$ for each study $j \in [n]$. Furthermore, suppose for any given feature $i$ and group $k$,  the local PC $p$-value $p^{\uiG}_i$ is  valid conditional on ${\bf X}$, and the selection set $\cS_k$  depends only on $\mathbf{P}_{ \cdot (-\cG_k)}$ and ${\bf X}$ (i.e.,  \condsref{valid_monotone_local_PC_pv} and \ref{cond:indep} are met). Then, for all $t \in [0,1]$,
\begin{enumerate}[(i)]
\item  $\Pr ( p^{u_{i k} / \cG_k}_{i} \leq t | i \in \cS_k, \mathbf{P}_{\cdot (-\cG_k)},  \mathbf{P}_{(-\cS_k) \cG_k}, \mathbf{X}  )  \leq t$  if $H^{\uiG}_i$ is true; 
 \item  $ \Pr ( p^{u_{i k} / \cG_k}_{i} \leq t |  i \in \cS_k,  \mathbf{P}_{(-i) \cdot}, \mathbf{P}_{i (-\cG_k)}, \mathbf{X}  ) \leq t$ if $H^{\uiG}_i$ is true. 
\end{enumerate}
\end{lemma}


\begin{proof}[Proof of \lemref{supert}]
  Under \assumpref{indepstudies} and the fact that  $\cS_k$  depends only on $\mathbf{P}_{ \cdot (-\cG_k)}$ and ${\bf X}$ (\condref{indep}), we know that 
  \begin{equation*}
  {\bf P}_{\cdot \cG_k}  \text{ and } \{ i \in \cS_k, \mathbf{P}_{\cdot (-\cG_k)} \} \text{ are independent conditional on } {\bf X},
  \end{equation*}
  which, under the mutual independence of $p_{1j},p_{2j},\dots,p_{mj}$ conditional on ${\bf X}$ for each $j$, further implies that 
  \[
  \mathbf{P}_{1 \cG_k} ,\dots,   \mathbf{P}_{m \cG_k} \text{ are mutually independent conditional on } \{ i \in \cS_k, \mathbf{P}_{\cdot (-\cG_k)} , {\bf X}\}.
  \]
  In particular, the last statement implies 
  \begin{equation}  \label{icGk_independent_of_else_a}
  \mathbf{P}_{i \cG_k} \text{ and } \mathbf{P}_{(- \cS_k) \cG_k} \text{ are  independent conditional on } \{ i \in \cS_k, \mathbf{P}_{\cdot (-\cG_k)} , {\bf X}\}
  \end{equation}
  and
  \begin{equation}  \label{icGk_independent_of_else_b}
  \mathbf{P}_{i \cG_k} \text{ and } \mathbf{P}_{(- i) \cG_k} \text{ are  independent conditional on } \{ i \in \cS_k, \mathbf{P}_{\cdot (-\cG_k)} , {\bf X}\}.
  \end{equation}
  Then, \condref{valid_monotone_local_PC_pv} and \eqref{icGk_independent_of_else_a}   together give \lemref{supert}$(i)$, and \condref{valid_monotone_local_PC_pv} and \eqref{icGk_independent_of_else_b}   together give \lemref{supert}$(ii)$ since $\mathbf{P}_{(- i) \cG_k} \cup \mathbf{P}_{\cdot (-\cG_k)} = \mathbf{P}_{(-i) \cdot} \cup \mathbf{P}_{i (-\cG_k)}$.
%

\end{proof}

%

\subsection{Additional  lemmas from the literature}
We state two useful lemmas from the existing literature,  with  notations adapted to our present context for the reader's convenience.

\begin{lemma}[Equation (5.4) of \citet{pfilter2019}]\label{lem:storey}
Let $S \in \mathbb{N}$ be a fixed natural number, $\lambda \in (0,1)$ be a fixed tuning parameter, $\nubf = (\nu_1,\dots,\nu_S) \in [0,\infty)^S$ be fixed non-negative weights where $\sum_{\ell \in [S]} \nu_{\ell} = S$, and $\cH_0 \subseteq [S]$ be a fixed subset of $[S]$. Let $\mathbf{p} =(p_1,\dots,p_S)$ be a vector of independent $p$-values where $\Pr(p_\ell \leq t) \leq t $ for all $t \in [0,1]$ if $\ell \in \cH_0$. Then
\begin{equation*}
    \mathbb{E} \left[ \frac{(1 - \lambda) \cdot S}{ \nu_{\max} + \sum_{\ell \in [S] \backslash \{ i \}} \nu_{\ell} \cdot I \left\{ p_\ell > \lambda \right\}} \right] \leq \frac{S}{ \sum_{\ell \in \cH_0} \nu_{\ell} }
\end{equation*}
for any $i \in \cH_0$ where $\nu_{\max} = \max_{\ell \in [S]} \nu_{\ell}$.
\end{lemma}

\begin{proof}[Proof of \lemref{storey}]
    For a vector of constants $(a_1,\dots,a_d) \in [0,1]^d$ where $d \in \bN$, and Bernoulli variables $b_1,\dots,b_d \overset{\tiny \text{iid}}{\sim} \text{Bernoulli}(1 - \lambda)$, the inequality
    \begin{align}\label{binomial_inequality}
        \mathbb{E} \left[ \frac{1}{1 + \sum^d_{i = 1} a_i \cdot b_i}  \right] \leq \frac{1}{(1 - \lambda) \cdot (1 + \sum^d_{i=1} a_i)}
    \end{align}
    holds by Lemma 3 of \citet[p.2805]{pfilter2019} where the right-hand-side is increasing in $\lambda$. Since $\Pr(p_{\ell} > \lambda) > 1 - \lambda$ if $\ell \in \cH_0$, we have that
    \begin{align*}
        \mathbb{E} \left[ \frac{(1 - \lambda) \cdot S}{ \nu_{\max} + \sum_{\ell \in [S] \backslash \{ i \}} \nu_{\ell} \cdot I \left\{ p_\ell > \lambda \right\}} \right]  &= \mathbb{E} \left[ \frac{(1 - \lambda) \cdot \frac{S}{\nu_{\max}} }{ 1 + \sum_{\ell \in [S] \backslash \{ i \} } \frac{\nu_{\ell}}{\nu_{\max}} \cdot \1 \left\{ p_{\ell} > \lambda \right\} } \right] \\
        &\leq \mathbb{E} \left[ \frac{(1 - \lambda) \cdot \frac{S}{\nu_{\max}} }{ 1 + \sum_{\ell \in \cH_0 \backslash \{ i \} } \frac{\nu_{\ell}}{\nu_{\max}} \cdot \1 \left\{ p_{\ell} > \lambda \right\} } \right] \\
        & \overset{\tiny (a)}{\leq} \frac{(1 - \lambda) \cdot \frac{S}{\nu_{\max}} }{ (1 - \lambda) \cdot ( 1 +  \sum_{\ell \in \cH_0 \backslash \{ i \} } \frac{\nu_{\ell}}{\nu_{\max}} ) } \\
        &\leq \frac{ S }{ \nu_{\max} +  \sum_{\ell \in \cH_0 \backslash \{ i \} } \nu_{\ell} } \\
        &\leq \frac{ S }{  \sum_{\ell \in \cH_0 } \nu_{\ell} }. 
    \end{align*}
    where $(a)$ is a result of \eqref{binomial_inequality}.
\end{proof}

\lemref{storey}  is similar to \citet[Corollary 13]{Blanchard2009} when applied to what they call ``modified Storey's estimator''; see \citet[p.2849]{Blanchard2009} for further details.

\bigskip

\begin{lemma}[Proposition 3.7 of \citet{Blanchard2008}]\label{lem:shape}
Let $S \in \mathbb{N}$ be a fixed natural number, $c \geq 0$ be a non-negative constant, $\cH_0 \subseteq [S]$ be a fixed subset of $[S]$, and $V \in [0, \infty)$ be an arbitrary non-negative random variable. Let $\mathbf{p} = (p_1,\dots,p_S)$ be a vector of arbitrarily dependent $p$-values where $\Pr(p_\ell \leq t) \leq t$ for all $t \in [0,1]$ if $\ell \in \cH_0$. Then, for a given  $i \in \cH_0$,
\begin{equation*}
    \mathbb{E} \Bigg[ \frac{ I \{ p_i \leq c \cdot ( \sum_{\ell \in [S]} \frac{1}{\ell} )^{-1} \cdot  V  \} }{ V } \Bigg] \leq c \text{ for all } c \geq 0.
\end{equation*}
\end{lemma}

 The original \citet[Proposition 3.7 ]{Blanchard2008} involves a ``shape'' function $\beta(V)$ (defined in their equation $(6)$). In \lemref{shape}, we restated their proposition with the shape function
\begin{equation*}
    \beta(V) = \left( \sum_{i \in [S]} \frac{1}{i} \right)^{-1} \cdot V,
\end{equation*}
which is justified as a valid shape function in Section 4.2.1 of \citet{Blanchard2008}. We also note that \lemref{shape} is similar to Lemma 1$(c)$ of \citet{pfilter2019}. 

\bigskip

\section{Proofs for \secref{MainPF}}\label{app:error_control}
This appendix contains the proof for all the theoretical results in \secref{MainPF}.
\subsection{Proof of  \lemref{FDR_rep_control_from_partition} ($\FDR_{k}$ control $\Rightarrow$ $\FDR_{\rep}$ control)} \label{app:FDR_k_control_gives_FDR_rep_control}
It suffices to show that
\begin{equation} \label{inclusion_suffice}
\cH^{u/[n]}  \subseteq \bigcup_{k \in [K]} \cH_k,
\end{equation}
i.e. any feature in $\cH^{u/[n]}$  must be in at least one of the sets $\cH_1 , \dots, \cH_K$.
If \eqref{inclusion_suffice} can be proven as  true, then it must be that
    \begin{multline*}
    \FDR_{\rep}(\widehat{\cR}) \equiv \mathbb{E} \left[ \frac{\sum_{i \in \cHun} \1 \{ i \in \widehat{\cR} \} }{1 \vee \sum_{i \in [m]} \1 \{ i \in \widehat{\cR} \} } \right] 
    \leq \sum_{k \in [K]} \mathbb{E} \left[ \frac{\sum_{i \in \cH_k } \1 \{ i \in \widehat{\cR} \} }{1 \vee \sum_{i \in [m]} \1 \{ i \in \widehat{\cR} \} } \right] = \sum_{k \in [K]}  \FDR_{k}(\widehat{\cR}).
    \end{multline*}
As such, if $ \FDR_{k}(\widehat{\cR}) \leq w_k \cdot q$ for each $k \in [K]$, having $\sum_{k=1}^K w_k = 1$ implies $\FDR_{\rep}(\widehat{\cR})  \leq q$.

To show \eqref{inclusion_suffice}, suppose towards a contradiction, there exists a feature $i_1 \in \cH^{u/[n]}  $ such that $i_1 \not \in \cH_k$ \textit{for all} $k \in [K]$. One can then write
    \begin{equation*}
       \left| \left\{ j \in [n] : \mu_{i_1j} \notin \mathcal{A}_{i_1} \right\} \right| \overset{\tiny (a)}{=} \sum_{k \in [K]} \left| \left\{ j \in \cG_k : \mu_{i_1j} \notin \mathcal{A}_{i_1} \right\}  \right| \overset{\tiny (b)}{\geq} \sum_{k \in [K]} u_{i_1 k} =  u
    \end{equation*}
    where $(a)$ is a result of $\cG_1, \dots, \cG_K$ forming a partition of $[n]$, and $(b)$ is true because 
    \[
    \left| \left\{ j \in \cG_k : \mu_{i_1j} \notin \mathcal{A}_{i_1} \right\}  \right| \geq u_{i_1 k},
    \] 
    a direct consequence of the assumption that $i_1 \not \in \cH_k $ for all  $k \in [K]$. However, $\left| \left\{ j \in [n] : \mu_{i_1 j} \notin \mathcal{A}_{i_1} \right\} \right|  \geq u$ means that $i_1 \not \in \cH^{u/[n]}$,  a contradiction.

\subsection{Proof of  \lemref{construct_local_PC} (Validity of local GBHPC $p$-values conditional on ${\bf X}$)}
\label{app:valid_local_GBHPC_pvalue}
 \assumpsref{superuniformp} and \assumpssref{pxindependence} together imply that for any $(i, j) \in [m]\times [n]$,
     \begin{equation} \label{cond_validity_of_each_base_pvalues}
    \Pr(p_{ij} \leq t \mid {\bf X}) \leq t \quad \text{for all $\ t \in [0,1]$ when $H_{ij}$ is true}.
    \end{equation} 
    Since \eqref{cond_validity_of_each_base_pvalues} and  that $p_{i1}, \dots,p_{in}$ are independent conditional on ${\bf  X}$  (\assumpref{indepstudies}), for any $\cJ \subset [\cG_k]$ and the restricted global null hypothesis $H_i^{1/\cJ}: \mu_{ij}  \in \cA_i \text{ for all } j \in \cJ$, any one of the four mentioned combining functions can produce a valid $f(\mathbf{P}_{i \cJ})$ with the property  that
      \begin{equation} \label{cond_validity_for_restricted_global_null}
  \Pr(f_{i, \cJ}( \mathbf{P}_{i \cJ}) \leq t \mid \mathbf{X}) \leq t \text{ for all }t \in [0,1] \text{ whenever }H_i^{1/\cJ} \text{ is true}.
  \end{equation}
    In particular, \eqref{cond_validity_for_restricted_global_null} holds for each $\cJ \subseteq \cG_k$ such that $| \cJ | = |\cG_k| - u_{ik} + 1$, and 
    by Proposition 1 of \citet{Wang2019}(as well as its proof), the local GBHPC $p$-value defined in \eqref{local_GBHPC} is conditionally valid conditional on ${\bf X}$. 

\subsection{Proof of \propref{optimal_thresholds} (Existence of an optimal vector $\hattvec$ in $\cT$)}\label{app:optimal_thresholds}
The fact that $t_k \leq \hatt_k$ for all $k \in [K]$ and any ${\bf t}= (t_1, \dots, t_K)  \in \cT$ simply comes from the definition of $\hat{t}_k$ in \eqref{optimal_tau_repk}. So we only have to  show $\hattvec \in  \cT$. Note that $\cT$ is non-empty because $\textbf{0}_K$ must be one of its elements.

  We first make the following claim, whose   proof will be deferred to the end of this subsection.
 \begin{claim}[$\hat{t}_k$ is achieved in $\cT$] \label{claim:t_hat_k_achieved}
For any given $k \in [K]$,  there must exist $K-1$ non-negative thresholds  $t^{\angleb{k}}_1$, $\dots$, $t^{\angleb{k}}_{k-1}$, $t^{\angleb{k}}_{k+1}$, $\dots$, $t^{\angleb{k}}_{K}$ such that $ \tvec^{\angleb{k}} \equiv (t^{\angleb{k}}_1, \dots, t^{\angleb{k}}_{k-1}, \hat{t}_{k}, t^{\angleb{k}}_{k+1}, \dots, t^{\angleb{k}}_{K}) \in \cT$, where $\hat{t}_{k}$ is defined according to \eqref{optimal_tau_repk}. 
     \end{claim}
 
    For each $k \in [K]$, let $\tvec^{\angleb{k}}$ be as suggested by \claimref{t_hat_k_achieved}. It follows from the definition of $\hattvec$ via \eqref{optimal_tau_repk} that
    \begin{equation}\label{angle_inequality}
        \tvec^{\angleb{k}} \matleq \hattvec \equiv(\hat{t}_{1},\dots,\hat{t}_{K}).
    \end{equation}
    Thus, for each $k \in [K]$, we have that
    \begin{align}\label{FDPhatthatinequality}
        w_k \cdot q \overset{\tiny (a)}{\geq} \FDPhat_{k}(\tvec^{\angleb{k}}) 
        = \frac{ | \cS_k |  \cdot \hat{\pi}_k \cdot \hat{t}_{k}}{|\cR(\tvec^{\angleb{k}})| \vee 1} 
        \overset{\tiny (b)}{\geq}  \frac{|\cS_k | \cdot \hat{\pi}_{k} \cdot \hatt_k }{|\cR(\hattvec)| \vee 1} = \FDPhat_{k}(\hattvec) 
    \end{align}
    where $(a)$ is a result of $\tvec^{\angleb{k}} \in \cT$, and $(b)$ is a result of \lemref{Rrep} and  \eqref{angle_inequality}. Since $\FDPhat_{k}(\hattvec) \leq w_k \cdot q$ for all $k \in [K]$ as shown in \eqref{FDPhatthatinequality}, it must be the case that $\hattvec \in  \cT$ by \defref{set_of_feasible_thresholds}.
Now we furnish the proof of \claimref{t_hat_k_achieved}.

   \begin{proof}[Proof of the \claimref{t_hat_k_achieved}]
    
    
    By the definition of $\hat{t}_k$  in \eqref{optimal_tau_repk}, there must exist a sequence of vectors 
    \begin{equation} \label{seq_converge_to_thatk}
    {\bf t}^{(1)} \equiv (t_1^{(1)}, \dots, t_K^{(1)}), \quad  {\bf t}^{(2)} \equiv (t_1^{(2)}, \dots, t_K^{(2)}),  \quad  {\bf t}^{(3)} \equiv (t_1^{(3)}, \dots, t_K^{(3)}), \quad \dots
    \end{equation}
    in $\bR_{\geq 0}^K$ such that 
     \begin{enumerate}[(i)]
        \item the sequence of their $k$th elements $(t_k^{(\ell)})_{\ell = 1}^\infty$ is increasing and converges to $\hat{t}_k$, and 
     \item   for each $\ell \in \bN$,
     $
    \FDPhat_{k'}(\tvec^{(\ell)})\leq  w_{k'}\cdot  q \text{ for all } k' \in [K]
     $. 
   
     \end{enumerate}     
     Moreover, we also define  the sequence of vectors
         \begin{equation*} \label{to_be_taken_subseq}
    {\bf t}^{(1)}_{\hat{t}_{k}} \equiv (t_1^{(1)}, \dots, \hat{t}_{k}, \dots, t_K^{(1)}), \quad 
     {\bf t}^{(2)}_{\hat{t}_{k}} \equiv (t_1^{(2)}, \dots, \hat{t}_{k}, \dots,  t_K^{(2)}),  \quad 
      {\bf t}^{(3)}_{\hat{t}_{k}} \equiv (t_1^{(3)}, \dots, \hat{t}_{k}, \dots, t_K^{(3)}), \quad  \dots,
    \end{equation*}
    which replace the $k$th element of each vector in \eqref{seq_converge_to_thatk} by $\hat{t}_k$.

     Suppose towards a contradiction that \claimref{t_hat_k_achieved} is not true. In particular, this means that none of the vectors in the sequence $(    {\bf t}^{(\ell)}_{\hat{t}_{k}})_{\ell = 1}^\infty$ can be a member in $\cT$, which in turn implies that there exists $\kappa\in [K]$ and a subsequence $(\ell_g)_{g=1}^\infty$ of increasing natural numbers such that
 \begin{equation}\label{FDP_hat_contradict}
    \FDPhat_{\kappa}(\tvec^{(\ell_g)}_{\hat{t}_k})  
    > w_{\kappa} \cdot q, \quad \text{ for all } g\in \bN.
     \end{equation}
      In what follows, we will leverage the following facts:  Property $(i)$ above about the sequence $({\bf t}^{(\ell)})_{\ell = 1}^\infty$ implies that 
          \begin{equation} \label{ramification_of_prop_2}
   \tvec^{(\ell_g)} \matleq  \tvec^{(\ell_g)}_{\hat{t}_k}  \text{ for all }g \in \bN.
    \end{equation} and  property $(ii)$ above about the sequence $({\bf t}^{(\ell)})_{\ell = 1}^\infty$ implies
    \begin{equation} \label{ramification_of_prop_1}
    |\cS_{k}| \cdot \hat{\pi}_{k} \cdot t_k^{(\ell_g)} \leq w_k \cdot q  \cdot (|\cR(\tvec^{(\ell_g)})| \vee 1)\text{ and all } g \in \bN.
    \end{equation}
We will arrive at a contradiction in two cases:

\begin{itemize}      
      \item 
       Case 1: $\kappa \neq k$. In this case, the $\kappa$-th element of $\tvec^{(\ell_g)}_{\hat{t}_k}$ is $t^{(\ell_g)}_{\kappa}$,
which  implies
     \[
    \FDPhat_{\kappa}(\tvec^{(\ell_g)}_{\hat{t}_k}) = \frac{|\cS_{\kappa}| \cdot \hat{\pi}_{\kappa} \cdot t_{\kappa}^{(\ell_g)}}{|\cR(\tvec^{(\ell_g)}_{\hat{t}_k})| \vee 1} \leq \frac{|\cS_{\kappa}| \cdot \hat{\pi}_{\kappa} \cdot t_{\kappa}^{(\ell_g)}}{|\cR(\tvec^{(\ell_g)})| \vee 1} =     \FDPhat_{\kappa}(\tvec^{(\ell_g)})   \leq w_{\kappa} \cdot q
     \]
     and contradicts \eqref{FDP_hat_contradict}, 
     where the first inequality above comes from \eqref{ramification_of_prop_2} and \lemref{Rrep}, and the second inequality comes from property $(i)$ above about the sequence $({\bf t}^{(\ell)})_{\ell = 1}^\infty$.
     \item 
      Case 2: $\kappa = k$.  One can write
     \begin{equation} \label{contradiction_when_kappa_eq_k}
      w_k \cdot q<   \FDPhat_{\kappa}(\tvec^{(\ell_g)}_{\hat{t}_k}) 
         =  \FDPhat_k(\tvec^{(\ell_g)}_{\hat{t}_k})
          = \frac{|\cS_k| \cdot \hat{\pi}_k \cdot \hat{t}_k}{|\cR(\tvec^{(\ell_g)}_{\hat{t}_k})| \vee 1} 
         \leq   \frac{|\cS_k| \cdot \hat{\pi}_k \cdot \hat{t}_k}{|\cR(\tvec^{(\ell_g)})| \vee 1},
     \end{equation}
     where the first inequality comes from \eqref{FDP_hat_contradict} and $\kappa = k$, and the second inequality comes from \eqref{ramification_of_prop_2} and \lemref{Rrep}. Now we will compare \eqref{ramification_of_prop_1} and \eqref{contradiction_when_kappa_eq_k}. By taking ``$\liminf$''  on both sides of  \eqref{ramification_of_prop_1} and using property $(ii)$ of  the sequence $({\bf t}^{(\ell)})_{\ell = 1}^\infty$ above, we get
\begin{equation} \label{ramification_of_prop_1_liminf}
    |\cS_{k}| \cdot \hat{\pi}_{k} \cdot \hat{t}_k \leq w_k \cdot q  \cdot \liminf_{g \rightarrow \infty} ( |\cR(\tvec^{(\ell_g)})| \vee 1).
    \end{equation}
     Rearranging \eqref{contradiction_when_kappa_eq_k} as 
    $
     w_k\cdot q  \cdot (|\cR(\tvec^{(\ell_g)})| \vee 1) <  |\cS_k| \cdot \hat{\pi}_k \cdot \hat{t}_k
     $
      and taking ``$\liminf$''  on both sides,  we get
          \begin{equation}\label{contradiction_when_kappa_eq_k_liminf}
     w_k\cdot q  \cdot  \liminf_{g \rightarrow \infty} (|\cR(\tvec^{(\ell_g)})| \vee 1) <  |\cS_k| \cdot \hat{\pi}_k \cdot \hat{t}_k,
     \end{equation}
     where, importantly, \eqref{contradiction_when_kappa_eq_k_liminf} remains  a strict inequality because $|\cR(\tvec^{(\ell_g)})| \vee 1$ can only take discrete values in $\bN$. Clearly,  \eqref{ramification_of_prop_1_liminf} and \eqref{contradiction_when_kappa_eq_k_liminf}  contradict each other.
     
 \end{itemize}   
       \end{proof}

\subsection{Proof of \propref{tau_correctness} (Correctness and Termination  of \algref{tau})}\label{app:tau_correctness}

%

 We  first claim that
       \begin{equation} \label{each_iteration_gives_sth_bigger_than_tvec_hat}
    \hattvec^{(s)} = (\hatt^{(s)}_1,\dots,\hatt^{(s)}_K) \matgeq \hattvec \text{ for all }   s \geq 0,
    \end{equation}
    and    
       \begin{align}\label{FDPhat_step2}
     \FDPhat_{k} \left( (\hattvec^{(s)}_{[1:k]}, \hattvec^{(s-1)}_{[(k+1):K]}) \right) \leq w_k \cdot q \text{ for all } k \in [K] \text{ and } s \geq 1;
    \end{align}
they will later be used to finish proving \propref{tau_correctness}.

We  first prove \eqref{each_iteration_gives_sth_bigger_than_tvec_hat} by using an induction argument.  Since $\hattvec^{(0)} = (\infty,\dots,\infty)$ by initialization, we have 
    \begin{equation}\label{initial_inequality}
        \hattvec^{(0)} \matgeq \hattvec = (\hatt_1,\dots,\hatt_K). 
    \end{equation}
 Suppose, for a given $k \in [K]$ and $s \geq 1$, it is true that 
    \begin{equation}\label{inductive_assumption_rep}
        (\hattvec^{(s)}_{[1:(k-1)]},\hatt^{(s-1)}_{k}, \hattvec^{(s-1)}_{[(k+1):K]}) \matgeq \hattvec.
    \end{equation}
    Since $\hatt_k$ is the $k$th element of $\hattvec$, the inductive assumption above in \eqref{inductive_assumption_rep} implies that:
    \begin{align}\label{inductiveresult}
        (\hattvec^{(s)}_{[1:(k-1)]},\hatt_{k}, \hattvec^{(s-1)}_{[(k+1):K]}) \matgeq \hattvec.
    \end{align}
    Since $\hattvec \in \cT$ by \propref{optimal_thresholds}, it follows by \defref{set_of_feasible_thresholds} that 
    \begin{align}
        w_k \cdot q  &\geq \FDPhat_{k} \left( \hattvec \right) \nonumber \\
        &= \frac{|\cS_k| \cdot \hat{\pi}_k \cdot \hat{t}_{k}}{|\cR \left( \hattvec \right)| \vee 1} \nonumber \\
        &\overset{\tiny (a)}{\geq} \frac{|\cS_k| \cdot \hat{\pi}_k \cdot \hat{t}_{k} }{ \left|\cR \left( (\hattvec^{(s)}_{[1:(k-1)]},\hatt_{k}, \hattvec^{(s-1)}_{[(k+1):K]}) \right) \right| \vee 1} \nonumber \\
        &= \FDPhat_{k} \left( (\hattvec^{(s)}_{[1:(k-1)]},\hat{t}_{k}, \hattvec^{(s-1)}_{[(k+1):K]}) \right) \label{FDPhatrepkalg}
    \end{align}
    where (a) is consequence of \lemref{Rrep} and \eqref{inductiveresult}. Result \eqref{FDPhatrepkalg} and \eqref{inductive_assumption_rep} implies that 
    \[
    \hatt^{(s)}_{k}  \in  [\hatt_{k}, \hatt^{(s-1)}_{k} ]
    \]
     by the definition of $\hatt^{(s)}_{k}$ in  Step 2 of \algref{tau}.
    Since $\hatt^{(s)}_{k} \geq \hatt_{k}$, it follows from \eqref{inductiveresult} that
    \begin{equation}\label{otherinductiveresult}
        (\hattvec^{(s)}_{[1:k]}, \hattvec^{(s-1)}_{[(k+1):K]}) \matgeq \hattvec.
    \end{equation}
 Since \eqref{otherinductiveresult} holds by the inductive assumption \eqref{inductive_assumption_rep}, it follows that $\hattvec^{(s)} \matgeq \hattvec$ will be true provided that $\hattvec^{(s-1)} \matgeq \hattvec$ is true. Thus, it follows  from the initial case \eqref{initial_inequality} that  \eqref{each_iteration_gives_sth_bigger_than_tvec_hat} is true.
    
    \ \
    
  Next we prove \eqref{FDPhat_step2}. 
    Suppose, towards a contradiction, that  \eqref{FDPhat_step2} is not true for a given $s \geq 1$ and $k$. By the definition of $\hatt_k^{(s)}$ in Step 2 of \algref{tau},   there exists a sequence of non-decreasing numbers $(\hatt^{(s)}_{k, \ell})_{\ell =1}^\infty$ converging to $\hatt_k^{(s)}$ such that 
    \[
     \FDPhat_{k} \left( (\hattvec^{(s)}_{[1:(k-1)]}, \hatt^{(s)}_{k, \ell}, \hattvec^{(s-1)}_{[(k+1):K]}) \right) \leq w_k \cdot q \quad \text{ for all } \quad \ell \in \bN.
    \]
    In particular, the latter fact implies that, for all $\ell \in \bN$,
    \begin{align*}
  |\cS_k| \cdot \hat{\pi}_k \cdot \hatt^{(s)}_{k, \ell} 
  &\leq w_k \cdot q \cdot \Big|\cR\Big( (\hattvec^{(s)}_{[1:(k-1)]}, \hatt^{(s)}_{k, \ell}, \hattvec^{(s-1)}_{[(k+1):K]})\Big)\Big| \vee 1\\
  & \leq w_k \cdot q \cdot \Big|\cR\Big( (\hattvec^{(s)}_{[1:k]}, \hattvec^{(s-1)}_{[(k+1):K]})\Big)\Big| \vee 1,
    \end{align*}
    where the second inequality above comes from \lemref{Rrep} and that $\hatt^{(s)}_{k, \ell} \leq  \hatt_k^{(s)}$ for all $\ell \in \bN$.
   By taking ``$\liminf$" with respect to $\ell$ on both sides of the last display, we obtain 
   \[
   \lim_{\ell \rightarrow \infty}  |\cS_k| \cdot \hat{\pi}_k \cdot \hatt^{(s)}_{k, \ell} =   |\cS_k| \cdot \hat{\pi}_k \cdot \hatt^{(s)}_k \leq w_k \cdot q \cdot \Big|\cR\Big( (\hattvec^{(s)}_{[1:k]}, \hattvec^{(s-1)}_{[(k+1):K]})\Big)\Big| \vee 1,
   \] 
   where the last inequality  contradicts that \eqref{FDPhat_step2} is not true.

\begin{proof}[Proof of \propref{tau_correctness}$(i)$]
   The  proof borrows similar induction arguments made in Section 7.2 of \citet{pfilter2019}.
 If $\hattvec^{(s^*-1)} = \hattvec^{(s^*)}$ holds for some time iteration $s^*$, then \eqref{FDPhat_step2} implies that, for all $k \in [K]$, 
    \begin{align}
         w_k \cdot q &\geq \FDPhat_{k} \left( (\hattvec^{(s^*)}_{[1:k]}, \hattvec^{(s^*-1)}_{[(k+1):K]}) \right) \nonumber \\
        &= \FDPhat_{k} \left( (\hattvec^{(s^*)}_{[1:k]}, \hattvec^{(s^*)}_{[(k+1):K]}) \right) \nonumber \\
        &= \FDPhat_{k} \left( \hattvec^{(s^*)} \right).\label{FDPhat_step2s}
    \end{align}
    It follows from \eqref{FDPhat_step2s} and \defref{set_of_feasible_thresholds} that 
    \begin{equation}\label{hattvecsstar}
        \hattvec^{(s^*)} \in \cT.
    \end{equation}
      Since $\hattvec^{(s^*)} \matgeq \hattvec$  by \eqref{each_iteration_gives_sth_bigger_than_tvec_hat} and $\hattvec$ is the largest element-wise vector in $\cT$ (\propref{optimal_thresholds}), it must be the case that  $\hattvec^{(s^*)} = \hattvec$ by \eqref{hattvecsstar}, and \propref{tau_correctness}$(i)$ is proven.

\end{proof}
\begin{proof}[Proof of \propref{tau_correctness}$(ii)$]

By construction, $(\hattvec^{(s)})_{s \in \bN}$ is a non-increasing sequence lowered bounded by $\hattvec$ as proven in \eqref{each_iteration_gives_sth_bigger_than_tvec_hat}, that is, 
\[
\hattvec^{(0)} \matgeq \hattvec^{(1)}  \matgeq  \hattvec^{(2)}  \cdots
\]
Hence, we can define the vector of the component-wise limit
\[
\hattvec^{(\infty)} = ( \hatt^{(\infty)}_1, \dots, \hatt^{(\infty)}_K),
\] 
where $\hatt^{(\infty)}_k \equiv \lim_{s\rightarrow \infty}\hatt^{(s)}_k$ for each $k \in [K]$ and 
\begin{equation} \label{decreasing_convegence}
\hattvec^{(s)} \searrow \hattvec^{(\infty)} \text{ as } s \rightarrow \infty.
\end{equation}

Suppose, towards a contradiction, \algref{tau} doesn't terminate. This implies there exists a $\kappa \in [K]$ such that  the sequence $(\hatt^{(s)}_\kappa)_{s= 0}^\infty$ doesn't stabilize in the tail, i.e.
\begin{equation} \label{not_stabilize}
 \text{for all $N \in \bN$, there exists $N \leq s_1 \leq s_2$ such that $\hatt^{(s_1)}_\kappa \neq \hatt^{(s_2)}_\kappa$}. 
 \end{equation}
 Now, we define, for each $k\neq \kappa$,
\[
\epsilon_k \equiv  \min \Big\{ \big|p_i^{u_{ik}/\cG_k}/\nu_{ik} -   \hatt^{(\infty)}_k \big|: i \in [m]  \text{ and }|p_i^{u_{ik}/\cG_k}/\nu_{ik} -   \hatt^{(\infty)}_k| \neq 0 \Big\} \vee 0.001,
\]
which is the minimal among the absolute differences in $\{\big|p_i^{u_{ik}/\cG_k}/\nu_{ik} -   \hatt^{(\infty)}_k \big|\}_{i = 1}^m$  that are positive, or set to $0.001$  if no such positive absolute differences exist. We also let
\[
\epsilon \equiv \min_{ k \neq \kappa} \epsilon_k 
\] 
By the definition of $\cR(\hattvec)$ in \eqref{RrepEasy}, we first re-write
\begin{equation} \label{rewrite}
\cR\Big( (\hattvec^{(s)}_{[1:(\kappa-1)]},t_\kappa,\hattvec^{(s-1)}_{[(\kappa+1):K]}) \Big) = 
\cR(t_\kappa)
  \cap \cR\Big(\hattvec^{(s)}_{[1:(\kappa-1)]} \Big)  
   \cap \cR \Big(\hattvec^{(s-1)}_{[(\kappa+1):K]} \Big),
  \end{equation}
where 
\begin{align*}
\cR(t_\kappa) 
& \equiv \left\{ i \in \cS_k :  p^{\uiG}_{i}/\nu_{ik} \leq  t_\kappa \wedge( \lambda_k /\nu_{ik})\right\},\\
\cR\Big(\hattvec^{(s)}_{[1:(\kappa-1)]} \Big)   
&\equiv
 \bigcap_{k \in [1:(\kappa-1)]} \left\{ i \in \cS_k :  p^{u_{ik}/\cG_k}_{i} / \nu_{ik} \leq 
 \hatt_k^{(s)} 
  \wedge (\lambda_k  / \nu_{ik} ) \right\} \text{ and } \\
  \cR \Big(\hattvec^{(s-1)}_{[(\kappa+1):K]} \Big) 
&\equiv \bigcap_{k \in [(\kappa+1):K]} \left\{ i \in \cS_k :  p^{u_{ik}/\cG_k}_{i} / \nu_{ik} \leq 
 \hatt_k^{(s-1)} 
  \wedge (\lambda_k  / \nu_{ik} ) \right\}. 
\end{align*}

Since $\hattvec^{(s)} $ converges decreasingly to $\hattvec^{(\infty)} $ as in   \eqref{decreasing_convegence}, let $S \in \bN$ be such that
\[
 |\hatt_k^{(s)} - \hatt_k^{(\infty)}| \leq \epsilon/2 \text{ for all } k \neq \kappa \text{ and } s \geq S.
\] 
In particular, by how $\epsilon$ is defined, it must be the case that
\begin{equation} \label{rej_set_partially_stabilize}
\cR\Big(\hattvec^{(s)}_{[1:(\kappa-1)]} \Big)    = \cR\Big(\hattvec^{(\infty)}_{[1:(\kappa-1)]} \Big) \text{ and }
\cR \Big(\hattvec^{(s-1)}_{[(\kappa+1):K]} \Big) 
  = \cR \Big(\hattvec^{(\infty)}_{[(\kappa+1):K]} \Big)  \text{ for all } s \geq S+1,
\end{equation}
where $\cR(\hattvec^{(\infty)}_{[1:(\kappa-1)]})$ and $\cR (\hattvec^{(\infty)}_{[(\kappa+1):K]} )$ are defined analogously to $\cR(\hattvec^{(s)}_{[1:(\kappa-1)]} )$ and $\cR (\hattvec^{(s-1)}_{[(\kappa+1):K]} ) $, except with $\hatt^{(s)}_1, \dots, \hatt^{(s)}_{\kappa -1}, \hatt^{(s-1)}_{\kappa +1} \cdots \hatt^{(s-1)}_K $ respectively replaced by $\hatt^{(\infty)}_1, \dots, \hatt^{(\infty)}_{\kappa -1}, \hatt^{(\infty)}_{\kappa +1} \cdots \hatt^{(\infty)}_K $.
From \eqref{rej_set_partially_stabilize} and \eqref{rewrite}, we further get that
\[
 \quad \cR\Big( (\hattvec^{(s)}_{[1:(\kappa-1)]},t_\kappa,\hattvec^{(s-1)}_{[(\kappa+1):K]}) \Big) = \cR(t_\kappa) \cap  \cR\Big(\hattvec^{(\infty)}_{[1:(\kappa-1)]} \Big) \cap \cR \Big(\hattvec^{(\infty)}_{[(\kappa+1):K]} \Big)     \text{ for all  }s \geq  S +1,
\]
which in turn implies that
\begin{equation} \label{non_changing_step_function}
\Big|\cR\Big( (\hattvec^{(s)}_{[1:(\kappa-1)]},t_\kappa,\hattvec^{(s-1)}_{[(\kappa+1):K]}) \Big) \Big|
= \Big|\cR\Big( (\hattvec^{(\infty)}_{[1:(\kappa-1)]},t_\kappa,\hattvec^{(\infty)}_{[(\kappa+1):K]}) \Big) \Big| \text{ for all } s \geq S+1.
\end{equation}
From \eqref{non_changing_step_function}, 
we have
    \begin{equation*}
       \hatt^{(s)}_\kappa  =  \max \Bigg\{ t_\kappa \in [0, \hatt_\kappa^{(s-1)} ] :
        \frac{|\cS_\kappa| \cdot \hat{\pi}_\kappa \cdot t_\kappa}{\big|\cR\big( (\hattvec^{(\infty)}_{[1:(\kappa-1)]},t_\kappa,\hattvec^{(\infty)}_{[(\kappa+1):K]}) \big) \big|} \leq w_\kappa \cdot q \Bigg\} \text{ for all } s \geq S+1.
    \end{equation*}
    Since   the ratio
    \[
    \frac{|\cS_\kappa| \cdot \hat{\pi}_\kappa }{\big|\cR\big( (\hattvec^{(\infty)}_{[1:(\kappa-1)]},t_\kappa,\hattvec^{(\infty)}_{[(\kappa+1):K]}) \big) \big|}
    \]
    is the same for all $s \geq S+1$, it is apparent that $ \hatt^{(S+1)}_\kappa =  \hatt^{(S+2)}_\kappa =  \hatt^{(S+3)}_\kappa = \cdots$, and therefore we have contradicted \eqref{not_stabilize}.
\end{proof}

\subsection{Proof of \thmref{error_control} ($\FDR_{\rep}$ control under independence)} \label{app:main_theorem_pf}

Under cases $(i)$ and $(ii)$, it suffices to show that for any group $k \in [K]$ and feature $i \in [m]$,
\begin{equation} \label{single_hypothesis_control}
      \mathbb{E} \left[ \left. \frac{ \1 \left\{  i \in \cR(\hattvec) \right\} }{|\cR(\hattvec)| \vee1 } \right| \mathbf{P}_{\cdot (- \cG_k) }, \mathbf{P}_{(-\cS_k) \cG_k}, \mathbf{X} \right] \leq \frac{\nu_{ik} \cdot w_k \cdot q}{(|\cS_k| \vee 1) \cdot \pi_{k}  } \quad \text{if $H^{\uiG}_i$ is true}.
\end{equation}
Provided that \eqref{single_hypothesis_control} is true, control of $\FDR_{k}$ at level $w_k \cdot q$ under cases $(i)$ and $(ii)$ can be established as follows:
\begin{align*}
    \FDR_{k}(\cR(\hattvec)) &= \mathbb{E} \left[ \frac{ \sum_{i \in \cH_k }  \1 \left\{ i \in \cR(\hattvec) \right\} }{|\cR(\hattvec)| \vee 1 } \right] \\
    &\overset{\tiny (a)}{=} \mathbb{E} \left[ \frac{ \sum_{i \in \cS_k \cap \cH_k  }  \1 \left\{ i \in \cR(\hattvec) \right\}  }{|\cR(\hattvec)| \vee 1 } \right] \\
    &= \mathbb{E} \left[ \mathbb{E} \left[ \frac{ \sum_{i \in \cS_k \cap \cH_k  } \1 \left\{ i \in \cR(\hattvec) \right\} }{|\cR(\hattvec)| \vee 1 } \Bigg| \mathbf{P}_{\cdot (- \cG_k) }, \mathbf{P}_{(-\cS_k) \cG_k }, \mathbf{X} \right] \right] \\
    &\overset{\tiny (b)}{=} \mathbb{E} \left[ \sum_{i \in \cS_k \cap \cH_k  } \mathbb{E} \left[ \frac{  \1 \left\{ i \in \cR(\hattvec) \right\} }{|\cR(\hattvec)| \vee 1 } \Bigg| \mathbf{P}_{\cdot (- \cG_k) }, \mathbf{P}_{(-\cS_k) \cG_k}, \mathbf{X} \right] \right] \\
    &\overset{\tiny (c)}{\leq} \mathbb{E} \left[  \sum_{i \in \cS_k \cap \cH_k } \frac{ \nu_{ik} \cdot w_{k} \cdot q }{ |\cS_k| \cdot \pi_{k}  }  \right] \\
    &\overset{\tiny (d)}{=} \mathbb{E} \left[   \pi_{k} \cdot  \frac{ w_{k} \cdot q }{\pi_{k}  }  \right] \\
    &= w_k \cdot q 
    \end{align*}
    where $(a)$ is a result of $\cR(\hattvec) \subseteq \cS_k$, $(b)$ holds because conditioning on $\{ \mathbf{P}_{\cdot (- \cG_k) }, \mathbf{P}_{(-\cS_k) (\cG_k)}, \mathbf{X} \}$ fixes $\cS_k$ by \condref{indep}, $(c)$ holds by \eqref{single_hypothesis_control}, and $(d)$ holds by the definition of $\pi_{k}$ in \eqref{group_pi}. 
    
    $\FDR_{\rep}$ control at level $q$ for both case $(i)$ and $(ii)$ then follows from a subsequent application of \lemref{FDR_rep_control_from_partition}. Hence, we will focus on proving \eqref{single_hypothesis_control} under cases $(i)$ and $(ii)$, starting with case $(ii)$ for expositional purposes.  
 
      \begin{proof}[Proof for case (ii)]  The proof shares similarities with the proof of Lemma S1 $(iv)$ in the supplementary materials of \citet{Katsevich2023}.
    
The weighted null proportion estimator $\hat{\pi}_{k}$ under this case is in the form of \eqref{pi_hat_adaptive}. Let
    \begin{equation*}
        \hat{\pi}^{(- i)}_{k} \equiv \frac{\max_{\ell \in \cS_k} \nu_{\ell k} + \sum_{\ell \in \cS_k \backslash \{ i \} } \nu_{\ell k} \cdot \1 \{ p^{u_{\ell k} / \cG_k}_{\ell} > \lambda_k \} }{(1 - \lambda_k) \cdot |\cS_k|}
    \end{equation*}
    where $\hat{\pi}^{(-i)}_{k}$ is $\hat{\pi}_{k}$ but without the summand indexed by $i$ in the numerator. Thus,  
    \begin{equation}\label{pihatrepk0inequality}
        \hat{\pi}^{(-i)}_{k} \leq \hat{\pi}_{k}.
    \end{equation}
    
    By substituting
    \begin{equation*}
        \text{ $S = |\cS_k|$, $\lambda = \lambda_k$, $\nubf = (\nu_{\ell k})_{\ell \in \cS_k}$,  $\mathbf{p} = (p^{u_{\ell k}/\cG_k}_{\ell})_{\ell \in \cS_k}$, and $\cH_0 = \cH_k \cap \cS_k$}
    \end{equation*}  
    into \lemref{storey}, we have that
        \begin{align}\label{blanchard_result_rep}
        \mathbb{E} \left[ \left. \frac{1}{\hat{\pi}^{(-i)}_{k}} \right| i \in \cS_k, \mathbf{P}_{\cdot (- \cG_k)}, \mathbf{P}_{(-\cS_k) \cG_k}, \mathbf{X}  \right] \leq \frac{|\cS_k| \vee 1}{\sum_{\ell \in \cS_k \cap \cH_k } \nu_{\ell k}} = \frac{1}{\pi_{k}}.
    \end{align}
    \lemref{storey} is applicable above because, when $\{ i \in \cS_k, \mathbf{P}_{\cdot (- \cG_k)}, \mathbf{P}_{(-\cS_k) \cG_k}, \mathbf{X} \}$ is given,
    \begin{itemize}
    \item  both $S = |\cS_k|= |\cS_k| \vee 1$ and $\cH_0 = \cH_k \cap \cS_k$  are fixed by \condref{indep};
    \item $\nubf = (\nu_{\ell k})_{\ell \in \cS_k}$ is fixed by \condref{nu}$(a)$; and
    \item the local PC $p$-values in $\mathbf{p} = (p^{u_{\ell k}/\cG_k}_{\ell})_{\ell \in \cS_k}$ are conditionally valid as described in \lemref{supert}$(i)$.
    \end{itemize}

For a given fixed matrix $\tilde{\mathbf{P}}_{(-i)\cdot} \in [0,1]^{(m-1) \times n}$,  vector $\tilde{\mathbf{P}}_{i (-\cG_k)} \in [0,1]^{n - |\cG_k|}$, and covariate matrix $\tilde{\mathbf{X}}$, we define the  event
    \begin{equation*}
        \mathcal{E}_{ik}( \tilde{\mathbf{P}}_{(-i) \cdot}, \tilde{\mathbf{P}}_{i (-\cG_k)}, \tilde{\mathbf{X}}) \equiv \{  \text{$\mathbf{P}_{(-i) \cdot} = \tilde{\mathbf{P}}_{(-i) \cdot} \text{, } \mathbf{P}_{i (-\cG_k)} = \tilde{\mathbf{P}}_{i (-\cG_k)} \text{, } \mathbf{X} = \tilde{\mathbf{X}}$, and $i \in \cR(\hattvec)$}  \},
    \end{equation*}
    which shares similarities to the event defined in equation (S3) in the supplementary materials of \citet{Katsevich2023}. In words, $\mathcal{E}_{ik}( \tilde{\mathbf{P}}_{(-i) \cdot},\tilde{\mathbf{P}}_{i (-\cG_k)}, \tilde{\mathbf{X}})$ is the event that $\mathbf{P}$ takes values so that $i \in \cR(\hattvec)$ when $\mathbf{P}_{(-i) \cdot} = \tilde{\mathbf{P}}_{(-i) \cdot}$, $\mathbf{P}_{i (-\cG_k)} = \tilde{\mathbf{P}}_{i (-\cG_k)}$, and $\mathbf{X} = \tilde{\mathbf{X}}$. 
    \\~\\
  Now,  if $\mathbf{P}_{(-i) \cdot}$, $\mathbf{P}_{i (-\cG_k)}$, and $\mathbf{X}$ are given, then the stability of $\cR(\hattvec)$ from  \lemref{stable} will imply that
    \begin{equation} \label{equal_Rminusi_on_event}
        \cR(\hattvec) =  \widehat{\cR}^{(-i)}  \text{ on the event } \mathcal{E}_{ik}(\mathbf{P}_{(-i) \cdot},\mathbf{P}_{i (-\cG_k)},\mathbf{X})
    \end{equation}
    for some set $\widehat{\cR}^{(-i)}$ that depends only on $\mathbf{P}_{(-i) \cdot}$ and $\mathbf{X}$, and is such that $i  \in \widehat{\cR}^{(-i)} $. Thus, 
    
    \begin{align*}
        &\mathbb{E} \left[ \left. \frac{ \1 \left\{ i \in \cR(\hattvec) \right\} }{|\cR(\hattvec)| \vee 1} \right| \mathbf{P}_{\cdot ( -\cG_k)}, \mathbf{P}_{(-\cS_k) \cG_k},\mathbf{X}  \right] \\
        &\overset{\tiny (a)}{=} \mathbb{E} \left[ \left. \frac{ \1 \left\{ i \in \cR(\hattvec) \right\} }{|\cR(\hattvec)| \vee 1} \1 \{ i \in \cS_k \} \right| \mathbf{P}_{\cdot ( -\cG_k)}, \mathbf{P}_{(-\cS_k) \cG_k},\mathbf{X}  \right] \\
        &\overset{\tiny (b)}{=} \mathbb{E} \left[ \left. \frac{ \1 \left\{ i \in \cR(\hattvec) \right\} }{|\cR(\hattvec)| \vee 1}  \right| i \in \cS_k, \mathbf{P}_{\cdot ( -\cG_k)}, \mathbf{P}_{(-\cS_k)\cG_k},\mathbf{X}  \right] \\
        &= \mathbb{E} \left[   \mathbb{E} \left[ \left. \left. \frac{ \1 \left\{ i \in \cR(\hattvec) \right\} }{|\cR(\hattvec)|\vee 1 } \right| i \in \cS_k, \mathbf{P}_{(-i) \cdot}, \mathbf{P}_{i  (-\cG_k)}, \mathbf{X}  \right] \right|  i \in \cS_k, \mathbf{P}_{\cdot ( -\cG_k)}, \mathbf{P}_{(-\cS_k)\cG_k},\mathbf{X}  \right] \\
         &= \mathbb{E} \left[   \mathbb{E} \left[ \left. \left. \frac{ \1 \left\{ i \in \cR(\hattvec) \right\} }{|\cR(\hattvec)| \vee 1 } \1 \Big\{    \mathcal{E}_{ik}( \mathbf{P}_{(-i) \cdot}, \mathbf{P}_{i (-\cG_k)}, \mathbf{X}) \Big\}  \right| i \in \cS_k, \mathbf{P}_{(-i) \cdot}, \mathbf{P}_{i  (-\cG_k)}, \mathbf{X}  \right] \right| i \in \cS_k, \mathbf{P}_{\cdot ( -\cG_k)}, \mathbf{P}_{(-\cS_k)\cG_k},\mathbf{X}  \right] \\
       &\overset{\tiny (c)}{\leq} \mathbb{E} \left[   \mathbb{E} \left[ \left. \left.  \frac{ \1 \left\{ p^{\uiG}_i \leq \frac{(|\widehat{\cR}^{(-i)}| \vee 1) \cdot \nu_{ik} \cdot w_k \cdot q}{(|\cS_k| \vee 1) \cdot \hat{\pi}_{k} } \right\} }{|\widehat{\cR}^{(-i)}| \vee 1} 
       \right| i \in \cS_k, \mathbf{P}_{(-i ) \cdot}, \mathbf{P}_{i (-\cG_k)}, \mathbf{X}  \right] \right| i \in \cS_k, \mathbf{P}_{\cdot ( -\cG_k)}, \mathbf{P}_{(-\cS_k) \cG_k},\mathbf{X}  \right] \\
       &\overset{\tiny (d)}{\leq} \mathbb{E} \left[   \mathbb{E} \left[ \left. \left.  \frac{ \1 \left\{ p^{\uiG}_i \leq \frac{(|\widehat{\cR}^{(-i)}| \vee 1) \cdot \nu_{ik} \cdot w_k \cdot q}{(|\cS_k| \vee 1) \cdot \hat{\pi}_{k}^{(-i)} } \right\} }{|\widehat{\cR}^{(-i)}| \vee 1} 
       \right| i \in \cS_k, \mathbf{P}_{(-i ) \cdot}, \mathbf{P}_{i (-\cG_k)}, \mathbf{X}  \right] \right| i \in \cS_k, \mathbf{P}_{\cdot ( -\cG_k)}, \mathbf{P}_{(-\cS_k) \cG_k},\mathbf{X}  \right] \\
        &\overset{\tiny (e)}{=} \mathbb{E} \left[ \left. \frac{ \Pr \left( \left. p^{\uiG}_i \leq \frac{(|\widehat{\cR}^{(-i)}| \vee 1) \cdot \nu_{ik} \cdot w_k \cdot q}{(|\cS_k| \vee 1) \cdot \hat{\pi}^{(-i)}_{k}} \right| i \in \cS_k, \mathbf{P}_{(-i) \cdot}, \mathbf{P}_{i  (-\cG_k)} , \mathbf{X} \right) }{|\widehat{\cR}^{(-i)}| \vee 1} \right| i \in \cS_k, \mathbf{P}_{\cdot ( -\cG_k)}, \mathbf{P}_{(-\cS_k) \cG_k},\mathbf{X} \right] \\
        &\overset{\tiny (f)}{\leq}  \mathbb{E} \left[ \left. \frac{ \left( \frac{(|\widehat{\cR}^{(-i)}| \vee 1) \cdot \nu_{ik} \cdot w_k \cdot q}{(|\cS_k| \vee 1)\cdot \hat{\pi}^{(-i)}_{k}} \right)  }{|\widehat{\cR}^{(-i)}| \vee 1 } \right| i \in \cS_k,  \mathbf{P}_{\cdot ( -\cG_k)}, \mathbf{P}_{(-\cS_k) \cG_k},\mathbf{X}  \right] \\
        &\leq \frac{\nu_{ik} \cdot w_k \cdot q}{|\cS_k| \vee 1} \cdot  \mathbb{E} \left[ \left. \frac{1}{\hat{\pi}^{(-i)}_{k}} \right| i \in \cS_k, \mathbf{P}_{\cdot ( -\cG_k)}, \mathbf{P}_{(-\cS_k) \cG_k},\mathbf{X}  \right] \\
        &\overset{\tiny (g)}{\leq} \frac{\nu_{ik} \cdot w_k \cdot q}{(|\cS_k| \vee 1)\cdot \pi_{k}  } 
    \end{align*}
    where 
    \begin{itemize}
    \item $(a)$ is a result of $\cR(\hattvec) \subseteq \cS_k$;
    \item $(b)$ is a result of $\cS_k$ being fixed when $\{ \mathbf{P}_{\cdot ( -\cG_k)}, \mathbf{P}_{(-\cS_k)\cG_k},\mathbf{X}  \}$ is given, by \condref{indep};
    \item $(c)$ is a result of \lemref{gw_sc} and \eqref{equal_Rminusi_on_event};
    \item $(d)$ is a result of \eqref{pihatrepk0inequality};
    \item $(e)$ is a result of $\widehat{\cR}^{(-i)}$ only depending on $\mathbf{P}_{(-i) \cdot}$ and $\mathbf{X}$;
    \item $(f)$ is a result of \lemref{supert}$(ii)$ because conditioning on $\{ i \in \cS_k, \mathbf{P}_{(-i) \cdot}, \mathbf{P}_{i (-\cG_k)}, \mathbf{X} \}$ fixes $\cS_k$ by \condref{indep}, fixes $\nu_{ik}$ by \condref{nu}$(a)$, and fixes both $\hat{\pi}^{(-i)}_{k}$ and $\widehat{\cR}^{(-i)}$ by their definitions; and
    \item $(g)$ is a result of \eqref{blanchard_result_rep}.
    \end{itemize}
    \end{proof}
 
  \begin{proof}[Proof for case (i)]
 The proof is nearly identical to the proof for case $(ii)$ above. The difference is that we set $\hat{\pi}_{k} = \hat{\pi}^{(-i)}_{k} = 1$ and verify that they satisfy properties \eqref{pihatrepk0inequality} and \eqref{blanchard_result_rep}. Property \eqref{pihatrepk0inequality} apparently holds.  Furthermore, property \eqref{blanchard_result_rep} will hold because
    \begin{align}
        \mathbb{E} \left[ \left. \frac{1}{\hat{\pi}^{(-i)}_{k}} \right| i \in \cS_k, \mathbf{P}_{\cdot (- \cG_k)}, \mathbf{P}_{ (-\cS_k) \cG_k}, \mathbf{X}  \right] = \mathbb{E} \left[ \left. \frac{1}{1} \right| i \in \cS_k, \mathbf{P}_{\cdot (- \cG_k)}, \mathbf{P}_{(-\cS_k) \cG_k}, \mathbf{X}  \right] = 1 \leq  \frac{1}{\pi_{k}}.
    \end{align}
  \end{proof}

\subsection{Proof of \thmref{error_control_dep} ($\FDR_{\rep}$ control under arbitrary dependence)} \label{app:main_theorem_dep_pf}

It suffices to show that for any group $k \in [K]$ and feature $i \in [m]$,
\begin{equation} \label{single_hypothesis_control_dep}
      \mathbb{E} \left[ \left. \frac{ \1 \left\{  i \in \cR(\hattvec) \right\} }{|\cR(\hattvec)| \vee1 } \right| \mathbf{P}_{\cdot (- \cG_k) }, \mathbf{X} \right] \leq \frac{\nu_{ik} \cdot w_k \cdot q}{(|\cS_k| \vee 1)\cdot \pi_{k}  } \quad \text{if $H^{\uiG}_i$ is true}.
\end{equation}

Provided that \eqref{single_hypothesis_control_dep} is true, control of $\FDR_{k}$ at level $w_k \cdot q$ can be established as follows:
\begin{align*}
    \FDR_{k}(\cR(\hattvec)) &= \mathbb{E} \left[ \frac{ \sum_{i \in \cH_k }  \1 \left\{ i \in \cR(\hattvec) \right\} }{|\cR(\hattvec)| \vee 1 } \right] \\
    &\overset{\tiny (a)}{=} \mathbb{E} \left[ \frac{ \sum_{i \in \cS_k \cap \cH_k  }  \1 \left\{ i \in \cR(\hattvec) \right\}  }{|\cR(\hattvec)| \vee 1 } \right] \\
    &\overset{\tiny (b)}{=} \mathbb{E} \left[ \sum_{i \in \cS_k \cap \cH_k  } \mathbb{E} \left[ \frac{  \1 \left\{ i \in \cR(\hattvec) \right\} }{|\cR(\hattvec)| \vee 1 } \Bigg| \mathbf{P}_{\cdot (- \cG_k) }, \mathbf{X} \right] \right] \\
    &\overset{\tiny (c)}{\leq} \mathbb{E} \left[  \sum_{i \in \cS_k \cap \cH_k } \frac{ \nu_{ik} \cdot w_{k} \cdot q }{ (|\cS_k| \vee 1 )\cdot \pi_{k}  }  \right] \\
    &\overset{\tiny (d)}{=} \mathbb{E} \left[   \pi_{k} \cdot  \frac{ w_{k} \cdot q }{\pi_{k}  }  \right] \\
    &= w_k \cdot q 
    \end{align*}
    where $(a)$ is a result of $\cR(\hattvec) \subseteq \cS_k$; $(b)$ holds because conditioning on $\{ \mathbf{P}_{\cdot (- \cG_k) },  \mathbf{X} \}$ fixes $\cS_k$ by \condref{indep}; $(c)$ holds by \eqref{single_hypothesis_control_dep}; and $(d)$ holds by the definition of $\pi_{k}$ in \eqref{group_pi}. 
    
    $\FDR_{\rep}$ control at level $q$ for both cases then follows from a subsequent application of \lemref{FDR_rep_control_from_partition}. Hence, we will focus on proving \eqref{single_hypothesis_control_dep}. 
    \\~\\
   \textit{Proof of \eqref{single_hypothesis_control_dep}.} By $\cR(\hattvec) \subseteq \cS_k$ and \condref{indep}, we have 
    \begin{align*}
         \mathbb{E} \left[ \left. \frac{ \1 \left\{ i \in \cR(\hattvec) \right\} }{|\cR(\hat{\tvec})| \vee 1 } \right| \mathbf{P}_{\cdot (-\cG_k)}, \mathbf{X}  \right]  
           &= \mathbb{E} \left[ \left. \frac{ \1 \left\{ i \in \cR(\hattvec) \right\} }{|\cR(\hat{\tvec})| \vee 1 }  \right|  i \in \cS_k , \mathbf{P}_{\cdot (-\cG_k)}, \mathbf{X}  \right] \\
         &\overset{\tiny (a)}{\leq} \mathbb{E} \left[\left.  \frac{ \1 \Big\{ p^{\uiG}_i \leq \frac{  (|\cR (\hat{{\bf t}} ) |\vee 1) \cdot \nu_{ik} \cdot   w_k \cdot q}{ \left( \sum_{\ell \in [|\cS_k| \vee 1]} \frac{1}{\ell} \right) \cdot (|\cS_k| \vee 1) } \Big\} }{|\cR(\hattvec)| \vee 1} \right|  i \in \cS_k , \mathbf{P}_{\cdot (-\cG_k) }, \mathbf{X}  \right] \\
         &\overset{\tiny (b)}{\leq} \frac{\nu_{ik} \cdot w_k \cdot q}{|\cS_k| \vee 1 } \\
         &\leq \frac{\nu_{ik} \cdot w_k \cdot q}{(|\cS_k| \vee 1) \cdot \pi_{k} } 
    \end{align*}
where $(a)$ is a result of \lemref{gw_sc} with $\hat{\pi}_{k} =  \sum_{\ell \in [|\cS_{k}| \vee 1]}\ell^{-1} $, and $(b)$ is a result of substituting 
    \begin{center}
        $S = |\cS_k|$, $\mathbf{p} = (p^{u_{\ell k}/ \cG_k}_\ell )_{\ell \in \cS_k}$, $\cH_0 = \cS_k \cap \cH_k$, $V = |\cR(\hattvec)| \vee1$, and $c = \frac{\nu_{ik} \cdot w_k \cdot q}{|\cS_k| \vee 1}$
    \end{center}
    into \lemref{shape}. \lemref{shape} is applicable above because, when $\{ i \in \cS_k, \mathbf{P}_{\cdot (-\cG_k)},\mathbf{X} \}$ is given,  
	\begin{itemize}
			\item both $S = |\cS_k|$ and $\cH_0 = \cS_k \cap \cH_k$ are fixed by \condref{indep};
			\item the local PC $p$-values in $\mathbf{p} = (p^{u_{\ell k}/ \cG_k}_\ell )_{\ell \in \cS_k}$ is valid conditional on $ i \in \cS_k , \mathbf{P}_{\cdot (-\cG_k) }, \mathbf{X}$ as in \lemref{supertX}; and
			\item $c = \frac{\nu_{ik} \cdot w_k \cdot q}{|\cS_k| }$ is fixed since $\nu_{ik}$ is fixed by  \condref{nu}$(b)$.
    	\end{itemize}

\section{Extended results from \secref{deprived_power_sim}}\label{app:deprived_power}
The simulation studies in this section extend  those performed in \secref{deprived_power_sim} by considering $n \in \{ 2, 3, 4, 5, 6, 7, 8 \}$ studies and $m \in \{ 1, 10, 20, \dots, 1000 \}$ hypotheses, where all the base null hypotheses $H_{ij}$ are \emph{false} with  highly right skewed  $p$-values generated as 
\[
p_{ij} \overset{\tiny \text{iid}}{\sim} \text{Beta}(0.26,7) \text{ for each } (i, j) \in [m] \times [n],
\]
as in \secref{deprived_power_sim}. 

\figref{stouffer_deprived_power} extends \figref{PCdemo} and shows the power of the BH procedure applied to  Stouffer-GBHPC $p$-values. For $m = 1$ and any given $n$, we observe a sharp drop in power when moving from $u = n - 1$ to $u = n$. For example, when $m = 1$ and $n = 8$, the power is $0.72$ at $u = 7$ and only $0.142$ at $u = 8$. Moreover, for any fixed $n$, the power at $u = n$ rapidly declines toward zero as $m$ increases, whereas for $u < n$, the power stabilizes around a value well above zero. For instance, when $n = 8$, the power at $u = 8$ drops from $0.142$ to $0$ as $m$ increases from $1$ to $1000$, while at $u = 7$, the power decreases from $0.72$ and converges to $0.45$.

We repeat this experiment using PC $p$-values constructed via the generalized Benjamini–Heller formulation in \eqref{u_combine_p}, with $f_\cJ$ taken to be Fisher’s combining function. The yields a PC $p$-values with the  form
\begin{equation*}
p^{u/[n]}_i = 1 - F_{2(n-u+1)} \left( -2 \sum^n_{j = u} \log(p_{i(j)}) \right),
\end{equation*}
where $p_{i(1)} \leq \cdots \leq  p_{i(n)}$ are the order statistics of the base $p$-values for feature $i$. The results shown in \figref{fisher_deprived_power} exhibit a similar phenomenon to what was observed earlier.
Overall, Fisher-GBHPC $p$-values are generally less powerful than their Stouffer-GBHPC counterparts, particularly as $n$ increases. 

\begin{figure}[H]
    \centering
    \includegraphics[width=0.9\textwidth]{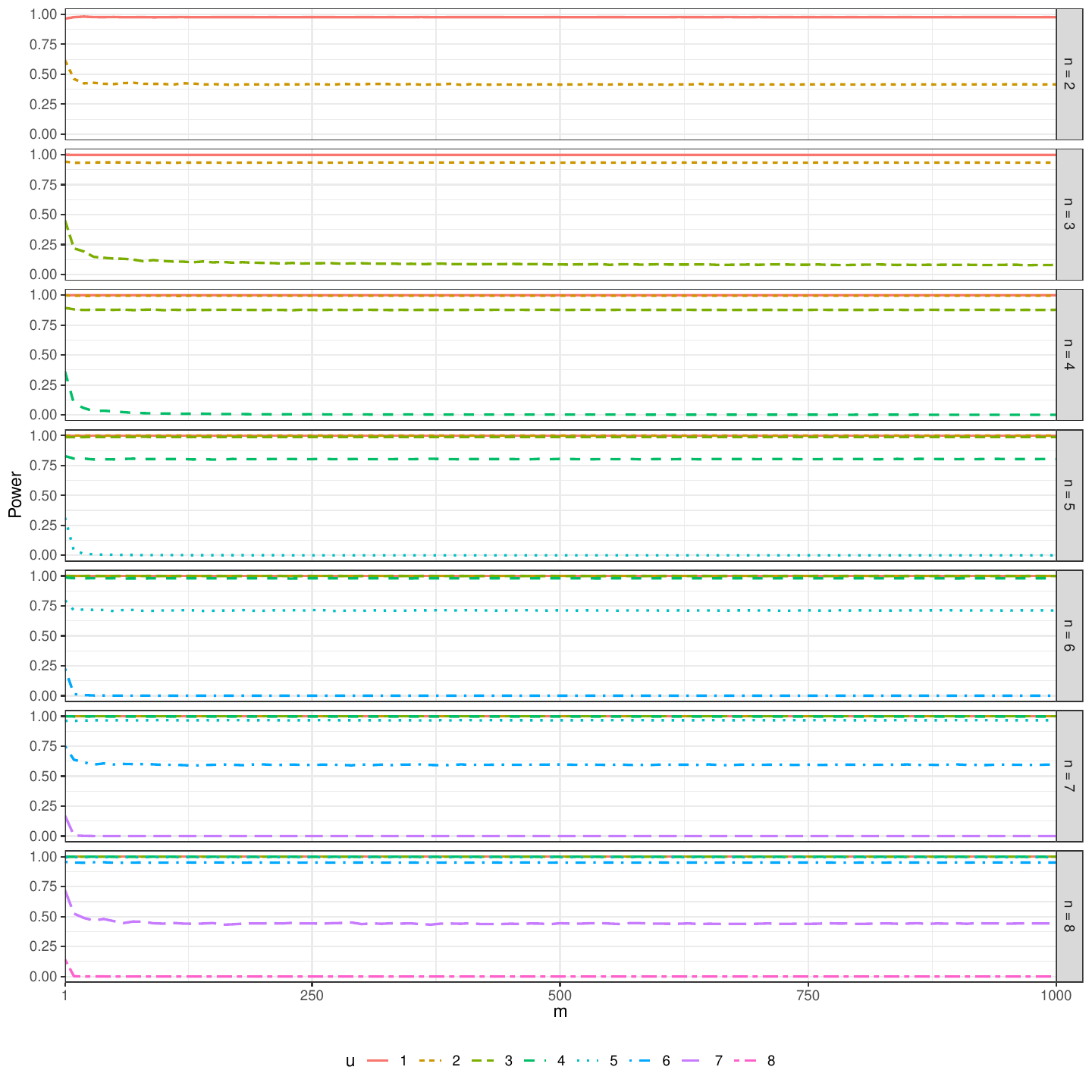}
    \caption{Power of the BH procedure (with FDR target 0.05) applied to $m = 1,10,20,\dots,1000$ Stouffer-GBHPC $p$-values under replicability levels $u = 1,\dots,n$ for a simulation experiment with $n = 2,3,\dots,8$ studies.}
    \label{fig:stouffer_deprived_power}
\end{figure}

\begin{figure}[H]
    \centering
    \includegraphics[width=0.9\textwidth]{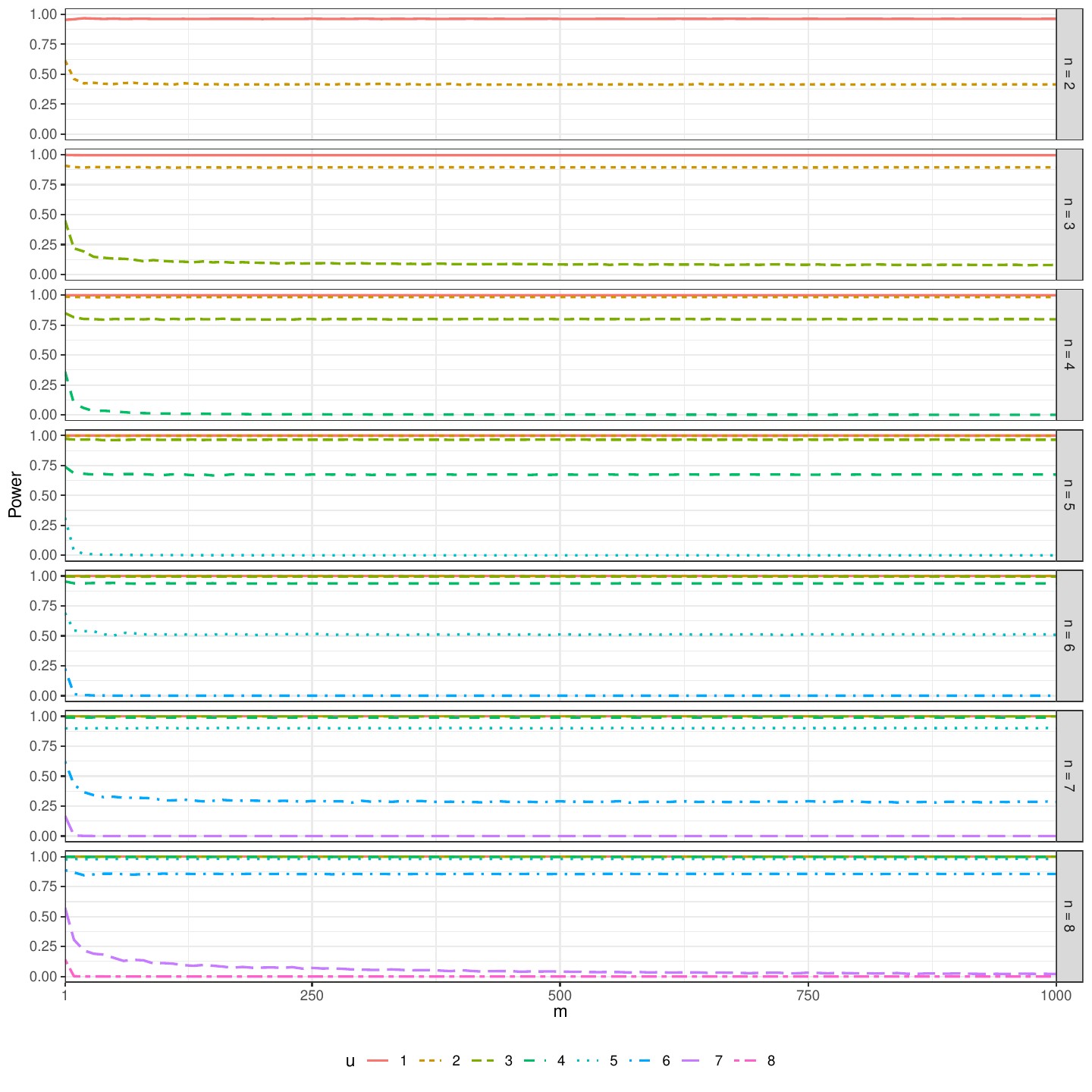}
    \caption{Power of the BH procedure (with FDR target 0.05) applied to $m = 1,10,20,\dots,1000$ Fisher-GBHPC $p$-values under replicability levels $u = 1,\dots,n$ for a simulation experiment with $n = 2,3,\dots,8$ studies.}
    \label{fig:fisher_deprived_power}
\end{figure}

\section{Additional simulations}\label{app:additional_sim}
Below, we introduce several additional $p$-value-based methods that will be evaluated alongside those in \secref{sim_methods} in the simulation studies presented in \appref{full_main_sim}--\ref{app:depueqn}. 
\begin{enumerate}[(a)]
  \setcounter{enumi}{9} 
\item \textbf{Adaptive-BH:} The plug-in adaptive step-up procedure described in \citet[Definition 8]{Blanchard2009} applied to the PC $p$-values in \eqref{maxPCp}, using the modified Storey's estimator \citep[Sec. 3.2]{Blanchard2009}, with tuning parameter $\lambda = 0.5$, for the reciprocated null proportion $m/|\cH^{{}_{u/[n]}}|$.
\item \textbf{BY:} The BY procedure by \citet{Benjamini2001} applied to the PC $p$-values in \eqref{maxPCp}. The BY procedure is similar to the BH procedure but multiplies the input target FDR level by a factor of $1 / \sum^m_{i = 1} \frac{1}{i}$. 
\item \textbf{Adaptive-CoFilter-BH:} The CoFilter procedure described in Algorithm 2.1 of \citet{Dickhaus2024} with the adaptive-BH procedure, described in item (j) above, implemented in step $(iii)$ of the algorithm.

\item \textbf{Inflated-\methodnamens:} An implementation of \algref{PF} similar to that of item $(a)$ in \secref{sim_methods}, but with:  \begin{enumerate}[1)] \item the selection rule
\begin{equation*}
    \cS_k = \bigcap_{\ell \in [K]\backslash \{k\} } \left\{ i \in [m] : p^{u_{i \ell}/\cG_{\ell}}_i \leq \left( w_{\ell} \cdot q / \left( \sum^m_{i = 1} \frac{1}{i} \right) \right) \wedge \lambda_{\ell}  \right\} ,
\end{equation*}
 \item local PC weights constructed according to \appref{localPCweights_b}, \item weighted null proportion estimators taken in the form of \eqref{misnomer}, and \item tuning parameter $\lambda_k = 1$  for each $k \in [K]$. \end{enumerate}
\item \textbf{Non-adaptive-\methodnamens:}  An implementation of \algref{PF} similar to that of item (a) in \secref{sim_methods}, but with the weighted null proportion estimator $\hat{\pi}_k = 1$ and tuning parameter $\lambda_k = 1$  for each group $k \in [K]$. 
 \item \textbf{Oracle:} A computationally simpler substitute for the the optimal $\FDR_{\rep}$ procedure. The details of its implementation are given in \appref{oracle}. While the oracle procedure is not implementable in real-world scenarios, as it requires knowledge of the density function of the base $p$-values, it serves as an interesting theoretical benchmark. 
\end{enumerate}

To avoid confusion, we clarify that the  qualifier ``adaptive'' as in the Adaptive-BH, Adaptive-CoFilter-BH and Non-adapative-ParFilter above refers to adaptiveness to the post-selection weighted proportions of local PC nulls in \eqref{group_pi} (as opposed to adaptiveness to the covariates ${\bf X}$).


\subsection{Simulations in \secref{sim} with additional methods}\label{app:full_main_sim}
We incorporate the  additional methods described above into the simulation study conducted in \secref{sim}. The $\FDR_{\rep}$ control guarantees provided by each of these additional methods under this simulation study, along with whether they are covariate-adaptive, are summarized in \tabref{method_summary_additional}. The empirical $\FDR_{\rep}$ and $\text{TPR}_{\rep}$ results are presented in \figref{Indepueqnfull}.

\begin{table}[H]
\small
\begin{tabular}{|l|l|l|}
\hline
Method                                               & Covariate-adaptive? & $\FDR_{\rep}$ control \\ \hline
{Adaptive-BH}               & No             & Finite (Corollary 13, \citet{Blanchard2009})  \\
{BY}                      & No              & Finite (Theorem 1.3, \citet{Benjamini2001})         \\
{Adaptive-CoFilter-BH}               & No             & Asymptotic (Proposition 4.2, \citet{Dickhaus2024})  \\
{Inflated-\methodnamens}                & Yes               & Finite (\thmref{error_control_dep})          \\
{Non-adaptive-\methodnamens:}                       & Yes              & Finite (\thmref{error_control}$(i)$) \\
{Oracle}                     & Yes              & Finite (\appref{oracle})                   \\
 \hline
\end{tabular}
\caption{A summary of the additional methods compared in the simulations of \appref{full_main_sim}. Specifically, whether the method is covariate-adaptive and the type of $\FDR_{\rep}$ control it possesses, i.e. finite-sample or asymptotic as $m \longrightarrow \infty$.}\label{tab:method_summary_additional}
\end{table}

When the covariates are uninformative ($\gamma = 0$), AdaFilter-BH (\AdaFilterBHCol line) outpowered the oracle (\OracleCol line) across all settings of $n$ and $\xi$. This may seem surprising, but recall that the oracle is a computationally simpler substitute for the actual optimal $\FDR_{\rep}$ procedure. AdaFilter-BH also only guarantees asymptotic $\FDR_{\rep}$ control, so its gain in power over the oracle comes at the cost of potentially violating $\FDR_{\rep}$ in finite samples  (e.g., when $(n,\gamma_1) = (5, 0)$). 

When the covariates are informative ($\gamma = 1.0$ or $1.5$), the oracle method is by far the most powerful, particularly as $n$ increases. The Adaptive-BH procedure (\AdaptiveBHCol line) was more conservative than the BH procedure, suggesting that the modified Storey estimator may not be effective when applied to PC $p$-values. The Adaptive-CoFilter-BH (\AdaptiveCoFilterBHCol line) tends to exhibit slightly greater power than the CoFilter-BH (\CoFilterBHCol line).  The BY procedure (\BYCol line) generally exhibited the worst power among all the compared methods.

The Inflated-\methodname procedure (\InflatedParFilterCol line) exhibited consistently low power across all settings. This is  unsurprising, given that it uses conservative weighted null proportion estimators designed to account for dependence among base $p$-values, despite the fact that the base $p$-values here are independent. Under mildly to highly informative covariates ($\gamma = 1.0$ and $\gamma = 1.5$, respectively), the Non-adaptive-\methodname (\NonAdaptiveParFilterCol line) was not competitive with the original \methodname implementation (\ParFilterCol line). When the covariates are non-informative ($\gamma = 0$), the \methodname and Non-adaptive-\methodname performed similarly, with the latter showing slightly better power. This suggests that, for \algref{PF}, the data-adaptive weighted null proportion estimator provided in \eqref{pi_hat_adaptive} is generally preferable to setting $\hat{\pi}_k = 1$.

\subsection{Simulations for  the case of $u < n$}\label{app:uleqnsim}
We repeat the simulations described in \appref{full_main_sim} but with the following settings for $u$ and $n$:
\begin{equation*}
    (u,n) \in \{ (2,3),(3,4),(3,5),(4,5) \}.
\end{equation*}
To accomodate for these new settings of $u$ and $n$, let the PC $p$-value for each feature $i \in [5000]$ be
\begin{equation*}
p^{u/[n]}_i \equiv 1 - \Phi \left( \frac{\sum_{j = u}^n \Phi^{-1}(1 - p_{i(j)})}{\sqrt{n - u + 1}} \right),
\end{equation*}
i.e., a GBHPC $p$-value (\defref{gbhpc}) where each $f_{\cJ}$ is taken as Stouffer's combining function in $n - u + 1$ arguments. Moreover, let $\mathcal{U}(u)$ denote the set of all valid choices for the local replicability levels $\mathbf{u}_i = (u_{i1},u_{i2})$ when $u$ is the replicability level and $K = 2$ is the number of groups. For instance,
\begin{equation*}
\mathcal{U}(4) = \{ (1,3), (3,1), (2,2) \}.
\end{equation*}
We reimplement the \methodname procedure described in item $(a)$ of \secref{sim_methods}, but under a new testing configuration where $K=2$, $\{ \cG_k, w_k \}_{k \in [2]}$ is as specified in \eqref{uleqnconfig}, and $\mathbf{u}_i$ is randomly sampled uniformly from $\mathcal{U}(u)$ for each $i \in [m]$. Similarly, we reimplement the No-Covar-ParFilter, Inflated-\methodname and Non-adaptive-\methodname procedures  under the same testing configuration. All other methods from \appref{full_main_sim} are implemented as originally described. 

The  empirical $\FDR_{\rep}$ and $\text{TPR}_{\rep}$ results of this new set of simulations are presented in \figref{Indepuleqn}. When the covariates are non-informative ($\gamma_1 = 0$), we see that AdaFilter-BH (\AdaFilterBHCol line) outpowers the remaining methods except the oracle (\OracleCol line) for all settings of $(u,n)$ and $\xi$. When the covariates are mildly ($\gamma_1 = 1.0$) or most ($\gamma_1 = 1.5$) informative, CAMT (\CAMTCol line) and AdaPT (\AdaPTCol line) display  high levels of power, especially for $(u,n) \in \{ (3,4), (3,5), (4,5) \}$. In comparison, the original \methodname (\ParFilterCol line), generally the most powerful among other  implementations of \algref{PF},   exhibits weaker power.
Even when the covariates are the most informative ($\gamma_1 = 1.5$), its power remained visibly lower than that of AdaFilter-BH, a method that does not adapt to covariates.

These  results  suggest that, at present, our various implementations of \algref{PF} may not be competitive for replicability analyses when $u < n$. This is likely due to the local replicability levels $\mathbf{u}_i$ being assigned randomly rather than derived from covariates. This randomness leads to a large number of imbalanced features (\defref{imbalance}), which in turn diminishes the power of \algref{PF}. Developing a general algorithm for constructing these levels using covariate information when $u < n$ would be a valuable direction for future research, though challenging due to the wide variability in the type, dimension, and informativeness of covariates across scientific contexts; see our discussion in \appref{uleqn}. In the meantime, CAMT and AdaPT appear to be  competitive choices  when side information is available in the $u < n$ setting.

\subsection{Simulations under negative AR(1) correlation}\label{app:depueqn}
We repeat the simulations in \secref{sim},  with the settings of $m$ and $(u,n)$, and the way we generate $\mathbf{X}$ and $\{ \1 \{ \text{$H_{ij}$ is true} \} \}_{(i,j) \in [m] \times [n]}$ remaining exactly the same as in \secref{sim}. However, we inject $\mathbf{P}_{\cdot j}$ with autoregressive dependence of order 1 (AR(1) correlation) across features, for each study $j \in [n]$. 
We first generate $z$-values $(z_{1j}, \dots, z_{mj})$ with strong negative AR(1) correlation:
\begin{equation*}
    z_{ij} \sim \text{N}(0,1) \text{ for each $i \in [m]$} 
\end{equation*}
where
\begin{equation*}
\text{cor}(z_{ij}, z_{i'j}) = (-0.8)^{|i - i'|} \text{ for any $i,i' \in [m]$;}
\end{equation*}
across studies, the $z$-values are independent.
From the generated $z$-values, we create base $p$-values by letting:
\begin{equation*}
    p_{ij} = \begin{cases}
        \Phi^{-1}(z_{ij}) &\text{if $H_{ij}$ is true}, \\
        F^{-1}_{\text{Beta}}(z_{ij}; 1 - \xi, 7) &\text{if $H_{ij}$ is false,}
    \end{cases}
    \quad \text{for each $(i,j) \in [m] \times [n]$}
\end{equation*}
where $\Phi^{-1}(\cdot)$ is an inverse standard normal distribution and $F^{-1}_{\text{Beta}}(\cdot; \alpha_1 , \alpha_2)$ is an inverse beta distribution function with shape parameters $\alpha_1$ and $\alpha_2$. The range of the signal strength parameter is the same as in \secref{sim}, i.e., $\xi \in \{ 0.72,0.74,0.76,0.78,0.80,0.82 \}$. The base $p$-values defined above have the same marginal distribution as the base $p$-values defined in \secref{sim}. But unlike \secref{sim}, the base $p$-values generated here are negatively dependent within each study.

The methods explored in \appref{full_main_sim} are reimplemented in this simulation, except for the oracle due to its computational complexity\footnote{It requires evaluating a multivariate normal density of dimension $m = 5000$.} under dependence. We summarize the theoretical $\FDR_{\rep}$ guarantees of each compared method under the data-generating setup of this simulation study in \tabref{method_summary_dep}. 

\begin{table}[h]
\small
\begin{tabular}{|l|l|l|}
\hline
Method                                               & $\FDR_{\rep}$ control \\ \hline
{\methodname}                & No proven guarantee       \\
{No-Covar-\methodname} & No proven guarantee \\
{AdaFilter-BH}               & Asymptotic  (Theorem 4.4, \citet{adafilter})  \\
{Inflated-AdaFilter-BH}      & Asymptotic  (Implied by Theorem 4.4, \citet{adafilter})     \\
{CoFilter-BH} & No proven guarantee \\
{BH}                &  No proven guarantee  \\
{AdaPT}                      & No proven guarantee  \\
{CAMT}                       & Asymptotic (Theorem 3.8, \citet{Zhang2022}) \\
{IHW}                        & No proven guarantee  \\

{Adaptive-BH}               & No proven guarantee  \\
{BY}                      & Finite (Theorem 1.3, \citet{Benjamini2001})         \\
{Adaptive-CoFilter-BH} & No proven guarantee \\
{Inflated-\methodnamens}                & Finite (\thmref{error_control_dep})          \\
{Non-adaptive-\methodnamens}                 & No proven guarantee \\
 \hline
\end{tabular}
\caption{Summary of the type of $\FDR_{\rep}$ control provided by each method compared in the simulations of \appref{depueqn}—i.e., finite-sample, asymptotic as $m \to \infty$, or none at all.}\label{tab:method_summary_dep}
\end{table}

The results for this simulation are presented in \figref{depueqn}. Overall, $\FDR_{\rep}$ control at $q = 0.05$ was observed by every method. The power of the compared methods remains largely consistent with their performance in \appref{full_main_sim}. Specifically, the \methodname (\ParFilterCol line) generally outperforms AdaFilter-BH (\AdaFilterBHCol line) when the covariates are mildly ($\gamma_1 = 1.0$) or most ($\gamma_1 = 1.5$) informative, while the AdaFilter-BH maintains its superior power when the covariates are non-informative ($\gamma_1 = 0$). The Inflated-\methodname (\InflatedParFilterCol line) exhibited relatively poor power across all simulation settings; however, it consistently outperformed the BY procedure (\BYCol line), Adaptive-BH (\AdaptiveBHCol line), and Inflated-AdaFilter-BH (\InflatedAdaFilterBHCol line). Recall that, aside from the Inflated-\methodnamens, the BY procedure is the only other method compared here that guarantees finite-sample $\FDR_{\rep}$ control with dependent base $p$-values; see \tabref{method_summary_dep}. 

The simulation results indicate that our original implementation of the \methodname (item (a) of \secref{sim_methods}) has empirical $\FDR_{\rep}$ control under highly negatively correlated base $p$-values. A direction for future work is to investigate whether the \methodname as implemented in item $(a)$ of \secref{sim_methods} has asymptotic $\FDR_{\rep}$ control under arbitrary dependence. If this is indeed the case, then implementing the Inflated-\methodname may not be necessary for large-scale meta-analyses with many features, as it was shown here to have  low power despite its appealing .

\begin{figure}[H]
    \centering
    \includegraphics[width=0.9\textwidth]{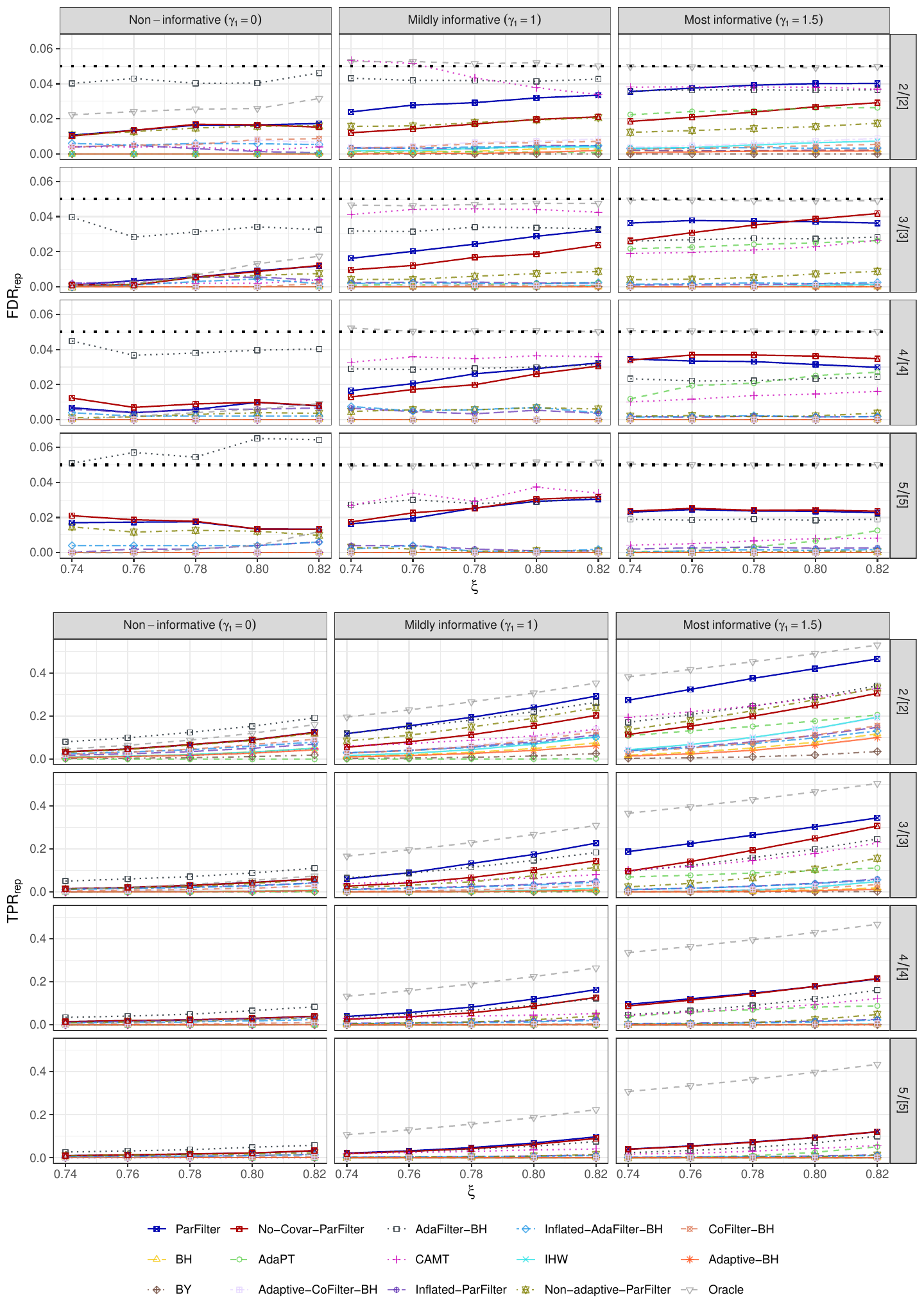}
    \caption{Empirical $\text{FDR}_{\rep}$ and $\text{TPR}_{\rep}$ of the compared methods for the simulations in \appref{full_main_sim}  based on $500$ simulated datasets, for testing $2/[2]$, $3/[3]$, $4/[4]$, and $5/[5]$ replicability  under  independence of base $p$-values across features.  Logistic parameters $\gamma_1 = 0, 1,$ and $1.5$, denoted as non-informative, mildly informative, and most informative respectively. Effect size parameters $\xi = 0.72, 0.74, 0.76, 0.78, 0.80, 0.82$. .}
    \label{fig:Indepueqnfull}
\end{figure}

\begin{figure}[H]
    \centering
    \includegraphics[width=0.9\textwidth]{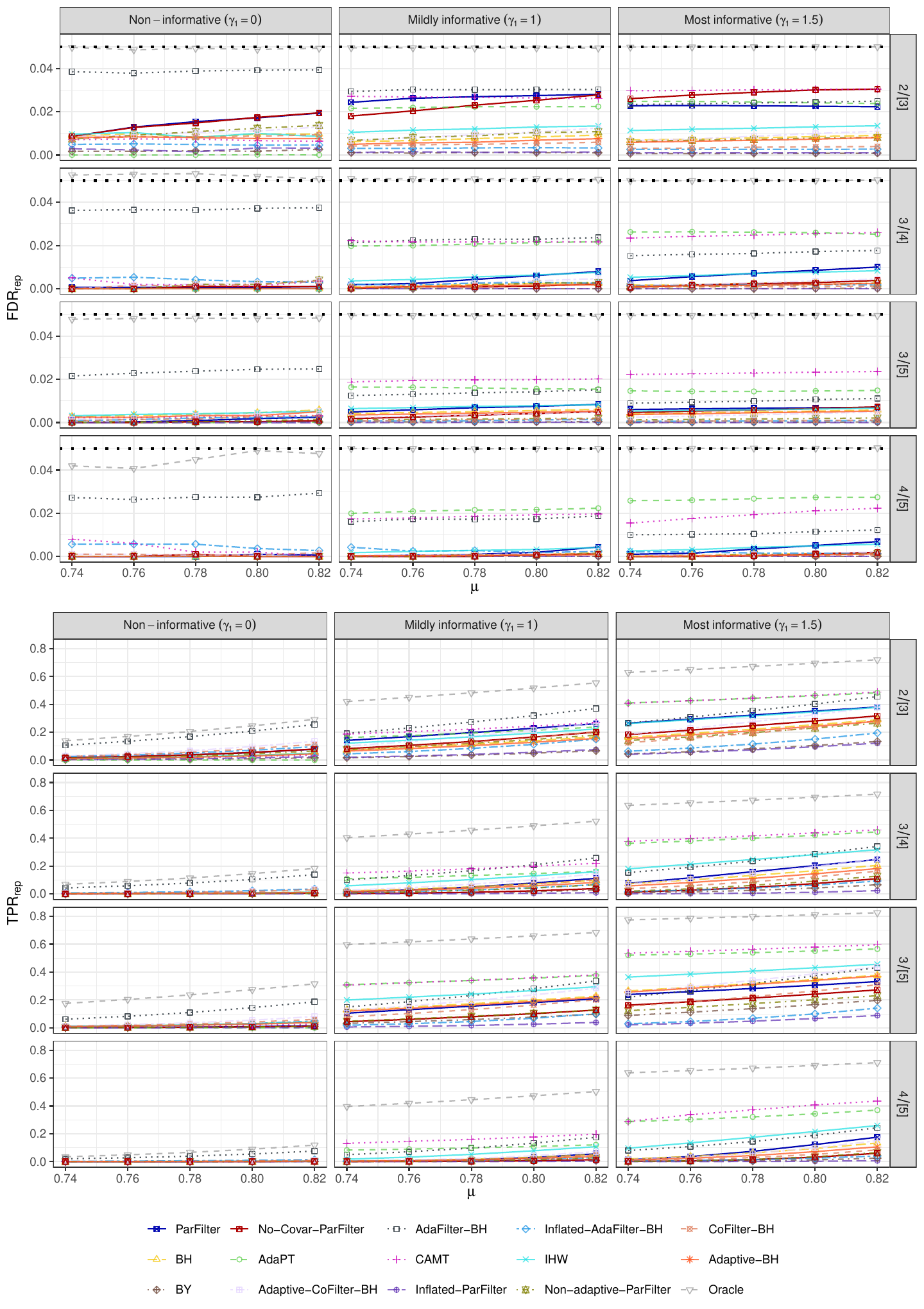}
    \caption{Empirical $\text{FDR}_{\rep}$ and $\text{TPR}_{\rep}$ of the compared methods for the simulations in \appref{uleqnsim} based on $500$ simulated datasets, for testing $2/[3]$, $3/[4]$, $3/[5]$, and $4/[5]$ replicability under  independence of base $p$-values across features. Logistic parameters $\gamma_1 = 0, 1,$ and $1.5$, denoted as non-informative, mildly informative, and most informative respectively. Effect size parameters $\xi = 0.72, 0.74, 0.76, 0.78, 0.80, 0.82$.  }
    \label{fig:Indepuleqn}
\end{figure}

\begin{figure}[H]
    \centering
    \includegraphics[width=0.9\textwidth]{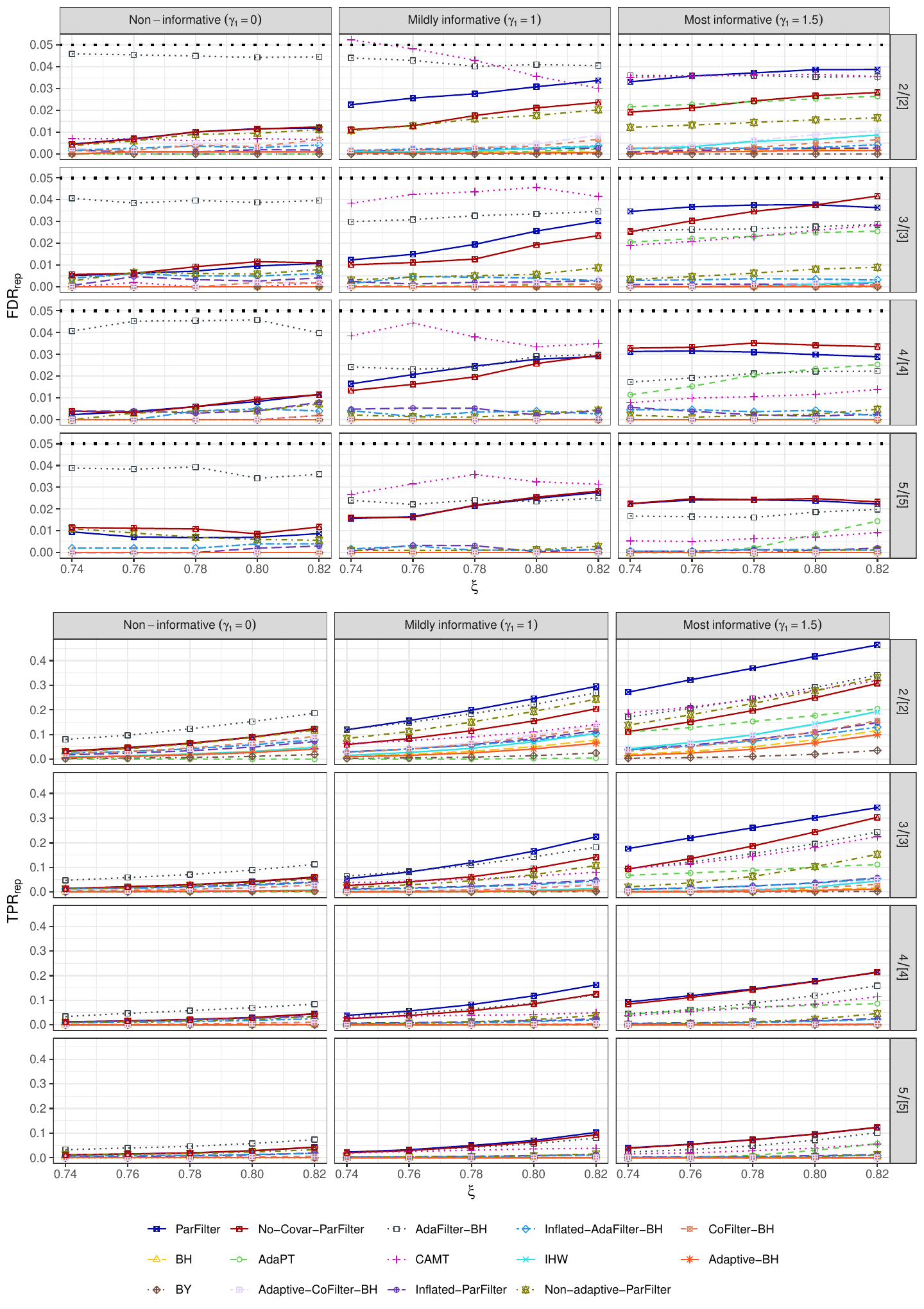}
    \caption{Empirical $\text{FDR}_{\rep}$ and $\text{TPR}_{\rep}$ of the compared methods for the simulations in \appref{depueqn}  based on $500$ simulated datasets, for testing $2/[2]$, $3/[3]$, $4/[4]$, $5/[5]$ replicability under AR(1) correlation of base $p$-values across features. Logistic parameters $\gamma_1 = 0, 1,$ and $1.5$, denoted as non-informative, mildly informative, and most informative respectively. Effect size parameters $\xi = 0.72, 0.74, 0.76, 0.78, 0.80, 0.82$.}
    \label{fig:depueqn}
\end{figure}

\section{Implementation of the oracle procedure}\label{app:oracle}
The \textit{optimal} $\FDR$ procedure achieves the highest expected number of true discoveries among all procedures that control finite-sample $\FDR$ \citep[Sec. 2.1]{Heller2020}. In this section, we describe the implementation of the oracle procedure used in \appref{additional_sim}, which served as a computationally simpler alternative to the optimal FDR procedure. The oracle procedure follows a similar approach to that of \citet{Sun2007} and avoids the need to solve the complex optimization problem required by the optimal FDR procedure \citep[Sec. 2.4]{Heller2020}.

For each $i \in [m]$, let $\Clfdr^{{}_{u/[n]}}_{i} \equiv \Pr(\text{$H^{{}_{u/[n]}}_i$ is true}|\mathbf{P},\mathbf{X})$ denote the probability that $H^{{}_{u/[n]}}_i$ is true given $\mathbf{P}$ and $\mathbf{X}$. Moreover, let $\Clfdr^{{}_{u/[n]}}_{(1)} \leq \cdots \leq \Clfdr^{{}_{u/[n]}}_{(m)}$ denote the order statistics of $\Clfdr^{{}_{u/[n]}}_{1}, \cdots, \Clfdr^{{}_{u/[n]}}_{m}$. For a target $\FDR_{\rep}$ level $q$, the oracle procedure computes a rejection index
\begin{equation}\label{lhat}
    \hat{l} \equiv \max \left\{ l \in [m] : \frac{\sum_{i \in [l]} \Clfdr^{u/[n]}_{(i)} }{l \vee 1} \leq q \right\}
\end{equation}
and returns the rejection set 
\begin{equation*}
    \widehat{\cR}^{\text{Oracle}} \equiv
    \begin{cases}
        \{ i \in [m] : \Clfdr^{u/[n]}_{i} \leq \Clfdr^{u/[n]}_{(\hat{l})} \} & \text{ if $\hat{l}$ exists}; \\
        \emptyset & \text{ otherwise}.
    \end{cases}
\end{equation*}
By design, $|  \widehat{\cR}^{\text{Oracle}}  | = \hat{l}$ if $\hat{l}$ exists.

The oracle procedure has finite-sample $\FDR_{\rep}$ control at level $q$ because:
\begin{align*}
    \FDR_{\rep}(\widehat{\cR}^{\text{Oracle}}) &\equiv \mathbb{E} \left[ \frac{ \sum_{i \in \widehat{\cR}^{\text{Oracle}}} \1 \left\{ i \in   \cH^{u/[n]} \right\} }{ |\widehat{\cR}^{\text{Oracle}}| \vee 1} \right] \\
    &= \mathbb{E}  \left[ \mathbb{E}  \left[  \frac{\sum_{i \in \widehat{\cR}^{\text{Oracle}}} \1 \left\{ i \in   \cH^{u/[n]} \right\} }{ |\widehat{\cR}^{\text{Oracle}}| \vee 1} \right] \Bigg|  \mathbf{P}, \mathbf{X}  \right] \\
    &= \mathbb{E}  \left[ \mathbb{E}  \left[  \frac{ \sum_{i \in [\hat{l}]} \Clfdr^{u/[n]}_{(i)} }{ \hat{l} \vee 1} \right] \Bigg|  \mathbf{P}, \mathbf{X}  \right] \\
    &\overset{\tiny (a)}{\leq} q
\end{align*}
where $(a)$ follows from the definition of $\hat{l}$ in \eqref{lhat}. Note, the proof above assumes $\hat{l}$ exists when $\mathbf{P}$ and $\mathbf{X}$ are given. If $\hat{l}$ does not exist, then the proof still holds because then
\begin{equation*}
    \mathbb{E}  \left[ \mathbb{E}  \left[  \frac{\sum_{i \in \widehat{\cR}^{\text{Oracle}}} \1 \left\{ i \in   \cH^{u/[n]} \right\} }{ |\widehat{\cR}^{\text{Oracle}}| \vee 1} \right] \Bigg|  \mathbf{P}, \mathbf{X}  \right] = \mathbb{E}  \left[ \mathbb{E}  \left[  \frac{0}{|\emptyset| \vee 1} \Big| \mathbf{P}, \mathbf{X} \right]   \right] = 0 \leq q.
\end{equation*}

\section{Additional replicability analysis results from \secref{realdata}}\label{app:real_data_extra}
Below, we highlight seven genes from \tabref{GeneRanks}, aside from F11r, Mknk2, and Mreg, that have known associations with autoimmune diseases.
\begin{itemize}
	  \item Escsr: \citet{Feng2024} found a strong association between Escsr expression and ulcerative colitis, an autoimmune related disease affecting the digestive tract.
   	 \item Jarid2: \citet{Pereira2014} discovered that Jarid2 is associated with the maturation of invariant natural killer T (iNKT) cells. iNKT cells are known to profoundly influence autoimmune disease outcomes in humans and animals \citep{Hammond2003}.
	\item Ktn1: \citet{Lu2005} found links between Ktn1 and Behçet's syndrome, a rare autoimmune related disease that causes inflammation in blood vessels.
	 \item Ncl: Lupus erythematosus is an autoimmune disease associated with dysfunction of C1QTNF4, a protein strongly linked to Ncl \citep{Vester2021}.
	   \item Bcl2l2: While primarily recognized for its role in regulating cell death, Bcl2l2 has also been associated with autoimmunity in mice and human males \citep{Tischner2010}
	\item Fxyd3: \citet{Yang2023} concluded that targeting Fxyd3 may offer a potential therapeutic approach for treating psoriasis, a common autoimmune skin disease.
	\item Rell1: The tumor necrosis factor receptor superfamily (TNFRSF) are proteins that are important therapeutic targets for autoimmune diseases \citep{Croft2024}. Rell1 is functionally associated with Relt, a member of the TNFRSF. 
\end{itemize}

\end{document}